\documentclass[11pt,reqno]{article}
\usepackage[letterpaper,top=1.5cm, bottom=1.8cm, left=1.8cm, right=1.8cm]{geometry}
\usepackage{amssymb}
\usepackage{amsmath}
\usepackage{amsfonts}
\usepackage{amsthm}
\usepackage{latexsym}
\usepackage[usenames]{color} 
\usepackage{bm}
\usepackage{bbm}
\usepackage[normalem]{ulem}
\usepackage{overpic}
\usepackage{graphicx}
\usepackage{mathrsfs}
 \usepackage{enumerate}
 \usepackage{fmtcount}

\usepackage{exscale}

\usepackage{mathtools}
\usepackage{commath}
\usepackage{hyperref}
\hypersetup{
    colorlinks,
    citecolor=black,
    filecolor=black,
    linkcolor=blue,
    urlcolor=black
}

\newtheorem{theorem}{Theorem}[section]
\newtheorem{proposition}[theorem]{Proposition}
\newtheorem{lemma}[theorem]{Lemma}
\newtheorem{corollary}[theorem]{Corollary}

\theoremstyle{definition}
\newtheorem{definition}[theorem]{Definition}

\newtheorem{remark}[theorem]{Remark}

\DeclareMathOperator{\Id}{Id}
\newcommand{\sums}{\sum_{\sigma \in \left\{+,-\right\}}}
\newcommand{\sumt}{\sum_{\tau \in \left\{+,-\right\}}}

\def\Ind#1{{\mathbbmss 1}_{_{\scriptstyle #1}}}
\def\lra{\longrightarrow}
\def\eps{\varepsilon}
\def\ignore#1{}

\renewcommand{\baselinestretch}{1}\normalsize

\title{High-frequency limit of Nash equilibria in a market impact game \\ with transient price impact}
\author{\normalsize Alexander Schied\footnote{Department of Statistics and Actuarial Science, University of Waterloo, and Department of Mathematics, University of Mannheim, Email: {\tt alex.schied@gmail.com}}\qquad Elias Strehle\setcounter{footnote}{6}\footnote{Department of Mathematics, University of Mannheim, Email: {\tt estrehle@mail.uni-mannheim.de}} \qquad Tao Zhang\setcounter{footnote}{3}\footnote{Department of Mathematics, University of Mannheim, Email: {\tt taozhang.de@gmail.com}\hfill\break
The authors gratefully acknowledge financial support by Deutsche Forschungsgemeinschaft (DFG) through Research Grants SCHI/3-1 and SCHI/3-2. In preliminary form, parts of this paper were published in the third author's doctoral dissertation~\cite{ZhangThesis}.}}
\date{\normalsize First version: September 28, 2015\\
\normalsize This version: May 9, 2017}                                   
\begin{document}
\maketitle
\begin{abstract}
We study the high-frequency limits of strategies and costs in a Nash equilibrium for two agents that are competing to minimize liquidation costs in a discrete-time market impact model with exponentially decaying price impact and quadratic transaction costs of size $\theta\ge0$. We show that, for $\theta=0$, equilibrium strategies and costs will oscillate indefinitely between two accumulation points. For $\theta>0$, however, strategies,  costs, and total transaction costs will converge towards limits that are independent of $\theta$. We then show that the limiting strategies form a Nash equilibrium for a continuous-time version of the model with $\theta$ equal to a certain critical value $\theta^*>0$, and that the corresponding expected costs coincide with the high-frequency limits of the discrete-time equilibrium costs. For $\theta\neq\theta^*$, however, continuous-time Nash equilibria will typically not exist.  Our results permit us to give mathematically rigorous proofs of numerical observations made in~\cite{SchiedZhangHotPotato}. In particular, we provide a range of model parameters for which the limiting expected costs of both agents are decreasing functions of $\theta$. That is, for sufficiently high trading speed, raising additional transaction costs can reduce the expected costs of all agents.
\end{abstract}

\noindent MSC classification: 91A05, 91A10, 91A25, 49N70, 91A60, 91A80, 91G10, 91G80\\
Keywords: Market impact game, transient price impact, Nash equilibrium, high-frequency limit, transaction costs

\section{Introduction}

In this paper, we continue the analysis of a market impact game  started in~\cite{Schoeneborn} and~\cite{SchiedZhangHotPotato}. In this game, two agents compete to minimize liquidation costs in a discrete-time market impact model with transient price impact. Such market impact models are typically used to describe situations in which orders arrive  fast enough to be affected by the reversion of the price impact created by previous trades~\cite{Bouchaudetal,ow,AFS2,ASS}. The observations in~\cite{Schoeneborn} and~\cite{SchiedZhangHotPotato} show that 
the qualitative behavior of Nash equilibria in models with transient impact is drastically different from the one of Nash equilibria in Almgren--Chriss-type models, in which orders are supposed to arrive slow enough to not be affected by any transient price impact component; we refer to~\cite{Carlinetal,SchoenebornSchied,CarmonaYang,SchiedZhangCARA} for analyses of market impact games in the context of the Almgren--Chriss model. A discrete-time market impact game with asymmetric information was analyzed in~\cite{Moallemietal}. Other applications of game theory to issues of market microstructures include~\cite{BressanFacchi1,BressanFacchi2,CarmonaWebster2,Lachapelleetal,GaydukNadtochiy}. General background on market impact models and the corresponding optimization problems can be found in the book~\cite{CarteaJaimungalPenalva}, the two surveys~\cite{GatheralSchiedSurvey,Lehalle}, and the references therein.

Specifically, it was shown in~\cite{Schoeneborn} and~\cite{SchiedZhangHotPotato} that the equilibrium strategies of both agents typically oscillate  between buy and sell trades, a behavior that is reminiscent of the \lq\lq hot-potato game" during the 2010 flash crash~\cite{SEC}. This behavior    can be interpreted as a protection against predatory trading by the competitor. Oscillations of trading strategies can of course be  dampened by adding additional transaction costs for each trade.  To study the effects of transaction costs on equilibrium strategies, \cite{SchiedZhangHotPotato} introduced quadratic transactions costs whose size is controlled by a parameter $\theta\ge0$. Varying  $\theta$ leads to a number of surprising effects. For instance, it was shown in    \cite[Theorem 2.7]{SchiedZhangHotPotato} that there exists an explicitly given critical value $\theta^*>0$ such that the equilibrium strategies show at least some oscillations for $\theta<\theta^*$, whereas all oscillations disappear for $\theta\ge\theta^*$.

The most surprising observations, however, concern the behavior of the expected liquidation costs. First, by means of numerical computations, it was illustrated  in~\cite{SchiedZhangHotPotato} that the expected costs of both agents can in some regimes be \emph{decreasing} functions of the size of transaction costs. That is, all market participants can be better off on average if additional transaction costs are imposed. Second, numerical simulations show that, for small $\theta$, 
expected costs can be (essentially) \emph{increasing} functions of the trading frequency, although a higher frequency means that agents can draw from a larger class of strategies and  thus should in principle be able to apply more cost-efficient strategies. This picture changes, however,  if transaction costs are increased. If $\theta$ is sufficiently large (e.g., if $\theta\ge\theta^*$) the expected costs  become decreasing functions of the trading frequency. 

The two phenomena described in the preceding paragraph can again be explained by the need of protecting against predatory trading. First, an increase of transaction costs discourages predatory trading so that  both traders need to take fewer precautions  and thus can use more cost-efficient strategies. The corresponding savings can  exceed the extra expenses in  transaction costs and thus lead to an overall reduction of costs for all market participants. Second, an increase of the trading frequency also increases the opportunities for predatory trading so that additional protective measures need to be taken if transaction costs are small. The cost of these measures can outweigh the benefit of higher trading frequency, thus leading to an increase of the expected costs. The situation changes if predatory trading has been sufficiently discouraged through higher transaction costs. In this case we observe the usual decrease of the expected costs as a function of the trading frequency.

One of the main goals of the present paper is to provide mathematical proof and justification for the numerical observations made in~\cite{SchiedZhangHotPotato}. Our main corresponding result will be Corollary~\ref{cost comparison cor}, which provides a range of model parameters that are sufficient for the expected costs of both agents to be decreased by raising additional  transaction costs.

The proof of Corollary~\ref{cost comparison cor} is based on  a thorough analysis of the behavior of equilibrium strategies and expected costs in the high-frequency limit, which is also interesting in its own right and exhibits a number of surprising features of our market impact game. In Theorem~\ref{limit strategies kappa thm} we will study the asymptotics of the accumulate equilibrium strategies. We show that, for $\theta=0$, these strategies oscillate indefinitely between two limiting curves, which we identify explicitly. For all $\theta>0$, however, the strategies converge to continuous-time limits that are given explicitly and are independent of $\theta$. This result will also be needed as input for the proof of 
Theorem~\ref{costs asymptotics thm}, which deals with the asymptotics of the expected costs in equilibrium. Theorem~\ref{costs asymptotics thm} states that, for $\theta=0$, the expected costs oscillate indefinitely and in the limit have exactly two distinct accumulation points, which are given in closed form. Again, the picture is different for $\theta>0$. In this case, the expected costs converge to an explicit limit that is independent of $\theta$. A comparison between this limit and the two  accumulation points for $\theta=0$ then yields the above-mentioned Corollary~\ref{cost comparison cor}. 

The convergence of the equilibrium strategies and costs raises the question whether the limiting strategies and costs can be associated with a continuous-time version of our market impact game. A corresponding continuous-time setup is provided in Section~\ref{high-frequency limit section}
 by drawing on the existing  literature for continuous-time market impact models with transient price impact~\cite{Gatheral,GSS,PredoiuShaikhetShreve,LorenzSchied, AS12, AlfonsiBlanc}. Theorem~\ref{main th continuous}, our corresponding main result, states that, for $\theta$ equal to the above-mentioned critical value $\theta^*$, there exists a unique Nash equilibrium that consists exactly of the limiting strategies found in the high-frequency limit of our discrete-time market impact game. Moreover, the expected costs of that equilibrium are equal to the high-frequency  limits of the discrete-time costs (Corollary~\ref{cont costs cor}). However, Theorem~\ref{main th continuous} states also that for $\theta\neq\theta^*$ no Nash equilibrium exists if at least one agent holds nontrivial inventory.  Preliminary versions of the results from Section~\ref{high-frequency limit section} and their proofs were stated earlier in the third author's doctoral dissertation~\cite{ZhangThesis}.
 
 To sum up, we show that, for $\theta=0$, equilibrium strategies and costs will oscillate indefinitely and not converge to any limit. For $\theta>0$, however, both strategies and costs will converge towards limits that are independent of $\theta$. These limits also appear in the unique continuous-time Nash equilibrium for $\theta=\theta^*$, which is essentially the only Nash equilibrium that exists in continuous time. 

This paper is organized as follows. In Section~\ref{prelim section} we recall the setup and the main results from~\cite{SchiedZhangHotPotato} in the form in which they are needed here. In Section~\ref{high-frequency limit section}
we state our asymptotic results on the high-frequency limits of equilibrium strategies and costs. Continuous-time Nash equilibria are discussed in Section~\ref{Continuous section}. Section~\ref{closed form section} prepares for the proofs of our main results by algebraically establishing explicit formulas for the discrete-time equilibrium strategies; the corresponding main result is  Theorem~\ref{omega and nu closed form thm}. 
The  sections in the appendix contain the proofs of our main results.

\section{Preliminaries}\label{prelim section}

Let us briefly recall the setup and the main results of~\cite{SchiedZhangHotPotato} in the special form in which they will be need here (the setup of~\cite{SchiedZhangHotPotato}  is more general). 
We consider two financial agents  who are active in a  market impact model for one risky asset. Price impact will be transient and modeled  by means of the exponential decay kernel,
$$G(t)=\lambda e^{-\rho t},\qquad t\ge0,
$$
where $\rho>0$. There is no loss of generality in taking $\lambda=1$ for the remainder of this paper, as all other quantities can be scaled accordingly. Transient price impact was proposed in~\cite{Bouchaudetal}, and the first analysis of a model with exponential decay of price impact was given in 
\cite{ow}. See, e.g.,~\cite{AFS2,Gatheral,ASS,GSS,PredoiuShaikhetShreve,LorenzSchied, AS12, AlfonsiBlanc, Donier} for further analyses and extensions  of this model, which is sometimes called a propagator model.

If none of the two agents is active, asset prices are described by a right-continuous martingale $S^0=(S^0_t)_{t\ge0}$ on a filtered probability space $(\Omega,(\mathscr{F}_t)_{t\ge0},\mathscr{F},\mathbb{P})$, for which $(\mathscr{F}_t)_{t\ge0}$ is right-continuous and  $\mathscr{F}_0$ is $\mathbb{P}$-trivial. The process $S^0$ is often called the unaffected price process. 
Trading takes place at the discrete trading times of a \emph{time grid} $\mathbb{T} =\{t_0,t_1,\dots, t_N\}$, where $0=t_0<t_1<\cdots<t_N= T$. Both agents are assumed to use trading strategies that are admissible in the following sense:

\begin{definition}An \emph{admissible trading strategy} for $\mathbb{T}$ and $z\in\mathbb{R}$ is a vector $\bm  \zeta=(\zeta_0,\dots,\zeta_N)$ of random variables such that 
 each $\zeta_i$ is $\mathscr{F}_{t_i}$-measurable and bounded, and 
 $\displaystyle z=\zeta_0+\cdots+\zeta_N$ $\mathbb{P}$-a.s.
The set of all admissible strategies for given $\mathbb{T}$ and $z$ is denoted by $\mathscr{X}(z,\mathbb{T})$.
\end{definition}

For $\bm \zeta\in\mathscr{X}(z,\mathbb{T})$, the value of $\zeta_i$ is taken as the number of shares traded  at time $t_i$, with a positive sign indicating a sell order and a negative sign indicating a purchase. 
Thus, $z$ is the inventory of the agent at time $0=t_0$,  and by time $t_N=T$  the agent must have a zero inventory. The assumption that each $\zeta_i$ is bounded can  be made without loss of generality from an economic point of view. 

If the two agents  apply the respective strategies $\bm \xi\in\mathscr{X}(x,\mathbb{T})$ and $\bm \eta\in\mathscr{X}(y,\mathbb{T})$,  the asset price is given by 
$$
S^{\bm \xi,\bm \eta}_t=S^0_t-\sum_{t_k<t}e^{-\rho(t-t_k)}(\xi_k+\eta_k).
$$
Based on this definition of an impacted price process, one can derive the following definition of transaction costs for each of the two agents.

\begin{definition}\label{discrete strategies cost def}Suppose that $\mathbb{T} =\{t_0,t_1,\dots, t_N\}$,  $x$ and $y$ are given.  
Let furthermore $\theta\ge0$ and $(\eps_i)_{i=0,1,\dots}$ be an i.i.d.~sequence of Bernoulli\,$(\frac12)$-distributed random variables that are independent of $\sigma(\bigcup_{t\ge0}\mathscr{F}_t)$. Then the \emph{costs of $\bm \xi\in\mathscr{X}(x,\mathbb{T})$ given $\bm \eta\in\mathscr{X}(y,\mathbb{T})$} are defined as
\begin{equation}\label{costs xi def eq}
\mathscr{C}_\mathbb{T}(\bm \xi|\bm \eta)=xS_0^0+\sum_{k=0}^N\Big(\frac{G(0)}2\xi_k^2-S_{t_k}^{\bm \xi,\bm \eta}\xi_k+\eps_k{G(0)}\xi_k\eta_k+\theta\xi_k^2\Big)
\end{equation}
and the \emph{costs of $\bm \eta$ given $\bm \xi$} are 
\begin{equation}\label{costs eta def eq}
\mathscr{C}_\mathbb{T}(\bm \eta|\bm \xi)=yS_0^0+\sum_{k=0}^N\Big(\frac{G(0)}2\eta_k^2-S_{t_k}^{\bm \xi,\bm \eta}\eta_k+(1-\eps_k){G(0)}\xi_k\eta_k+\theta\eta_k^2\Big).
\end{equation}
\end{definition}

In the preceding definition, $\varepsilon_k$ is a random variable that determines whether the first or the second agent's order is executed first, when both agents place orders at the same instance. In addition to the endogenous liquidation costs, each trade $\zeta_k$ also incurs quadratic transaction costs  of the form $\theta\zeta^2_k$, whose size is determined by the parameter $\theta\ge0$. Such quadratic transaction costs are often used to model   \lq\lq slippage" arising from temporary price impact; see  \cite{BertsimasLo,AlmgrenChriss2} and \cite[Section 2.2]{Gatheral}. Nevertheless, proportional transaction costs might be more realistic in many situations, and so the question arises to what extend  results will change when the quadratic transaction costs $\theta\xi_k^2$ are replaced by (piecewise) linear transaction costs. This question will be discussed in Remark~\ref{transaction costs remark} below.

As usual, for a given time grid $\mathbb{T}$ and initial values $x$, $y\in\mathbb{R}$, a \emph{Nash equilibrium} is  a pair $(\bm \xi^*,\bm \eta^*)$ of strategies in $\mathscr{X}(x,\mathbb{T})\times\mathscr{X}(y,\mathbb{T})$ such that 
$$\mathbb{E}[\,\mathscr{C}_\mathbb{T}(\bm \xi^*|\bm \eta^*)\,]=\min_{\bm \xi\in\mathscr{X}(X_0,\mathbb{T})}\mathbb{E}[\,\mathscr{C}_\mathbb{T}(\bm \xi|\bm \eta^*)\,]\qquad\text{and}\qquad \mathbb{E}[\,\mathscr{C}_\mathbb{T}(\bm \eta^*|\bm \xi^*)\,]=\min_{\bm \eta\in\mathscr{X}(Y_0,\mathbb{T})}\mathbb{E}[\,\mathscr{C}_\mathbb{T}(\bm \eta|\bm \xi^*)\,].
$$
This definition assumes implicitly that each agent has full knowledge of the other agent's initial inventory ($x$ or $y$, respectively) and strategy ($\bm\xi^*$ or $\bm\eta^*$, respectively).

The existence of  a unique Nash equilibrium in the class of deterministic strategies and with $\theta=0$ was established in    \cite[Theorem 9.1]{Schoeneborn}. This result was then extended in~\cite{SchiedZhangHotPotato} to general decay kernels,  adapted strategies,  and arbitrary $\theta\ge0$. Even more importantly,  an explicit formula for the equilibrium strategies was given that allows for the qualitative analysis of equilibrium strategies and costs, which will be continued in the present paper. We now recall this  existence result from~\cite{SchiedZhangHotPotato} in the specific setting of exponential decay of price impact and for the equidistant time grids,
\begin{equation*}
\mathbb{T}_N=\Big\{\frac{kT}N\,\Big|\,k=0,1,\dots,N\Big\},\qquad N=2,3,\dots
\end{equation*}
To this end, we fix $N$ and we define the lower triangular $\left(N+1\right)\times\left(N+1\right)$-matrix $\tilde{\Gamma}$ by
\begin{align*}
	(\left.\tilde{\Gamma}\right.)_{ij} &= \begin{cases} 0 &\text{ if } i<j,\\ 1/2 &\text{ if } i=j,\\ e^{-\rho({i-j})T/N} &\text{ if } i>j.\end{cases}
\end{align*}
We then denote the transposition of matrices and vectors with the symbol $\top$, and let
$$\Gamma := \tilde{\Gamma} + \tilde{\Gamma}^\top $$
be the so-called Kac--Murdock--Szeg\H o matrix. 
Let $\Id$ denote the $\left(N+1\right)\times\left(N+1\right)$-identity matrix and $\mathbf{1}$ the $\left(N+1\right)$-column vector that contains only ones.
Define the $\left(N+1\right)$-column vectors 
$${\bm\nu} := (\left.{\Gamma}+\tilde{\Gamma}+2\theta\Id\right.)^{-1}\mathbf{1}\qquad\text{and}\qquad {\bm\omega} := (\left.{\Gamma}-\tilde{\Gamma}+2\theta\Id\right.)^{-1}\mathbf{1};$$
 both vectors are well-defined by \cite[Lemma 3.2]{SchiedZhangHotPotato}. In Theorem~\ref{omega and nu closed form thm} we will derive closed-form representations for both $\bm\nu$ and $\bm\omega$. Let furthermore 
\begin{equation}\label{v and w}
\bm v:=\frac1{\bm1^\top\bm\nu}\bm\nu\qquad\text{and}\qquad 
\bm w:=\frac1{\bm1^\top\bm\omega}\bm\omega
\end{equation}
denote the corresponding normalized vectors; again, the denominators are nonzero by \cite[Lemma 3.2]{SchiedZhangHotPotato}.

\begin{theorem}[Theorem 2.5 in~\cite{SchiedZhangHotPotato}]\label{th1}For any  $\rho>0$, $\theta\ge0$, time grid $\mathbb{T}$, and initial values $x$, $y\in\mathbb{R}$, there exists a unique  Nash equilibrium $(\bm \xi^*,\bm \eta^*)\in\mathscr{X}(x,\mathbb{T})\times\mathscr{X}(y,\mathbb{T})$. The optimal strategies $\bm \xi^*$ and $\bm \eta^*$ are deterministic and given by
$$
\bm \xi^*=\frac12(x+y)\bm v+\frac12(x-y)\bm w\qquad\text{and}\qquad
\bm \eta^*=\frac12(x+y)\bm v-\frac12(x-y)\bm w.
$$
\end{theorem}

It was observed in~\cite{Schoeneborn} and~\cite{SchiedZhangHotPotato} that, for vanishing transaction costs (i.e., $\theta=0$), the equilibrium strategies may display strong oscillations with alternating buy and sell trades. These oscillations can be interpreted as a protection against predatory trading by the competitor~\cite{Schoeneborn,SchiedZhangHotPotato}. When transaction costs increase, these oscillations are dampened and they finally disappear when transaction costs reach a specific critical level $\theta^*$. This is the content of the following result from~\cite{SchiedZhangHotPotato}  for the specific case of exponential decay of price impact.

\begin{theorem}[Theorem 2.7 in~\cite{SchiedZhangHotPotato}]\label{nonosc thm}The following conditions are equivalent.
\begin{enumerate}
\item For every $N\in\mathbb{N} $ and $\rho>0$, all components of $\bm v$ are nonnegative.
\item For every $N\in\mathbb{N}  $ and $\rho>0$, all components of $\bm w$ are nonnegative.
\item $\theta\ge\theta^*=1/4$.
\end{enumerate}
\end{theorem}

A very surprising numerical observation made in \cite[Section 2.4]{SchiedZhangHotPotato} is that the expected costs of both  agents can \emph{decrease} when transaction costs \emph{increase}; see Figure~\ref{costs figure}. Economically, this phenomenon can be interpreted as follows. 
An increase of transaction costs leads to a decrease of predatory trading. Therefore, both traders need to take fewer precautions against predatory trading and thus can use more efficient strategies. The savings from the preceding effect can sometimes exceed the extra expenses in increased transaction costs, thus leading to an overall reduction of costs for all market participants. An additional effect is that an increase of transaction costs  reduces the volume that is traded and thus can  ``calm the market".

One of the main goals of the present paper is to give a mathematical proof and a quantitative analysis of the above-mentioned  numerical observation from  \cite[Section 2.4]{SchiedZhangHotPotato}.   To this end, the following sections will analyze the high-frequency limit of equilibrium strategies and costs, i.e., the limit when $N\uparrow\infty$. Before that, however, we conclude this section by recalling a discussion from \cite{SchiedZhangHotPotato} on how results might be impacted if our quadratic transactions were replaced by (piecewise) linear ones.

\begin{remark}[Quadratic versus proportional transaction costs]\label{transaction costs remark} It was shown in \cite[Proposition 2.6]{SchiedZhangHotPotato} that in the context of Theorem~\ref{th1} there exists a piecewise linear function $\tau$ of the form 
\begin{equation}\label{piecewise linear transction costs}
\tau(|x|)=\theta_0|x|+\sum_{k=1}^{M}\theta_k(|x|-c_k)\Ind{[c_k,\infty)}(|x|)
\end{equation}
with certain coefficients $\theta_k>0$ and thresholds $0<c_1<\cdots c_{M}$ such that $(\bm \xi^*,\bm \eta^*)$  is also a Nash equilibrium in $\mathscr{X}(x,\mathbb{T})\times\mathscr{X}(y,\mathbb{T})$ for the the modified expected cost functional in which the quadratic transaction cost function $x\mapsto \theta x^2$ in \eqref{costs xi def eq} and \eqref{costs eta def eq} is replaced with $x\mapsto\tau(|x|)$. Transaction costs of the form \eqref{piecewise linear transction costs} can  model a transaction tax that is subject to tax progression. With such a tax, small orders, such as those  placed by small investors, are taxed at a lower rate than large orders, which may be placed with the intention of moving the market. Moreover, since the main difference of quadratic and proportional transaction costs is their behavior at the origin, one may guess  that similar results as recalled in this section  for quadratic transaction costs and fixed $N$ might also hold for proportional transaction costs. Nevertheless, the function $\tau$ for which the above-mentioned result holds depends on all model parameters and in particular on $N$. We can therefore not expect that the limiting results obtained in the following sections remain valid if our quadratic transaction cost function is replaced by proportional transaction costs in a neighborhood of the origin. Indeed, in the limit $N\uparrow\infty$, individual trades of the equilibrium strategies can become arbitrarily small, and so the differences between quadratic and proportional transaction costs become crucial.\end{remark}

\begin{figure}
\begin{center}
\includegraphics[width=8cm]{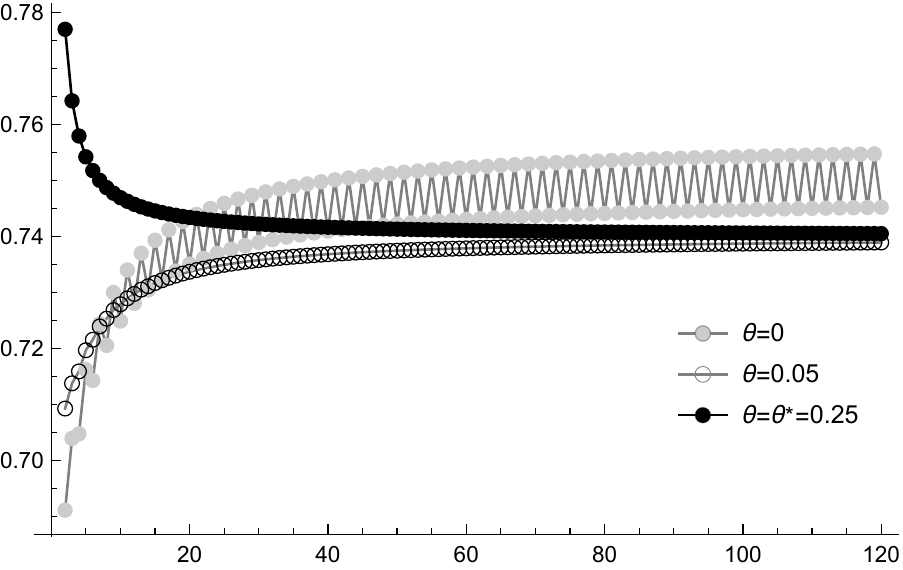}
\caption{Expected costs $\mathbb{E} [\mathscr{C}_{\mathbb{T}_N}(\bm\xi|\bm\eta) ]$ as a function of trading frequency, $N$,  for $\theta=0$, $\theta=0.05$, and $\theta=\theta^*=1/4$ if  $\rho T=1$ and $x=y=1$. }
\label{costs figure}
\end{center}
\end{figure}

 \section{High-frequency limits for equilibrium strategies and costs}\label{high-frequency limit section}
 
 Let us introduce 
 \begin{equation*}
 n_t := \Big\lceil\frac{ N t}T\Big\rceil
 \end{equation*}
 and the renormalized strategies
$$
		V_t^{\left(N\right)} := 1-\sum_{k=1}^{n_t} v_k \qquad\text{and}\qquad W^{(N)}_t=1-\sum_{k=1}^{n_t} w_k,\qquad 0\le t\le T,
		$$
	where $\bm v=(v_1,\dots,v_{N+1})^\top$ and $\bm w=(w_1,\dots,w_{N+1})^\top$ are as in \eqref{v and w}. To keep the notation simple, we will not make the dependence  of $\bm v$, $\bm w$, and of many other quantities on $N$ explicit. 
 Our first  main result  will deal with the asymptotics of $V^{(N)}$ and $W^{(N)}$.

\begin{theorem}[\bfseries Asymptotics of  strategies] \label{limit strategies kappa thm} The renormalized strategies $V^{(N)}$ and $W^{(N)}$ behave as follows as $N\uparrow\infty$. 
\begin{enumerate}
\item If $\theta>0$, then $V_0^{\left(N\right)} = 1$ and
	\begin{align*}
		V_t^{\left(N\right)} \lra \frac{e^{3\rho T}\left(6\rho\left(T-t\right)+4\right)-4e^{3\rho t}}{2e^{3\rho T}\left(3\rho T+5\right)-1},\qquad 0<t\le T.
	\end{align*} \item
 For $\theta=0$, we define the functions $f_\pm, g_\pm:[0,T]\to\mathbb{R}$ by
	\begin{align*}
		f_\pm\left(t\right) &:= \left(2e^{6\rho T}\left(3\rho T+5\right)+e^{3\rho T}+3\rho T+7\right)^{-1}\Big(\pm 3e^{3\rho\left(T-t\right)}\pm6e^{3\rho\left(2T-t\right)}+{}\\
		&\hspace{1cm}{}+e^{6\rho T}\left(6\rho\left(T-t\right)+4\right)+3\rho\left(T-t\right)+2e^{3\rho T}+4e^{3\rho t} -4e^{3\rho\left(T+t\right)}+3\Big),\nonumber\\
		g_\pm\left(t\right) &:= \left(2e^{6\rho T}\left(3\rho T+5\right)-3e^{3\rho T}-3\rho T-7\right)^{-1}\Big(\pm 3e^{3\rho\left(T-t\right)}\pm6e^{3\rho\left(2T-t\right)}+{}\\
		&\hspace{1cm}{}+e^{6\rho T}\left(6\rho\left(T-t\right)+4\right)-3\rho\left(T-t\right)-2e^{3\rho T}-4e^{3\rho t} -4e^{3\rho\left(T+t\right)}-3\Big).\nonumber
	\end{align*}
	Then $V_0^{\left(N\right)} = 1$, and for $0<t\le T$ the sequence $(\left.V_t^{\left(2N\right)}\right.)_{N \in \mathbb{N}}$ has exactly the two cluster points  $f_+\left(t\right)$ and $f_-(t)$, and $(\left.V_t^{\left(2N+1\right)}\right.)_{N \in \mathbb{N}}$ has exactly the two cluster points  $g_+\left(t\right)$ and $g_-(t)$.
	\item If $\theta>0$,  then 	
$$
	W^{(N)}_t\lra\frac{\rho(T-t)+1}{\rho T+1},\qquad 0\le t\le T.
$$
\item For $\theta=0$, we define the functions $\varphi_\pm, \psi_\pm:[0,T]\to\mathbb{R}$ by
\begin{align*}
\varphi_\pm(t)&:=\frac{1+\rho( T-t) \pm e^{-\rho(T-t)} }{1+\rho T+e^{-\rho T}},\\
\psi_\pm(t)&:=\frac{1+\rho( T-t) \pm e^{-\rho(T-t)} }{1+\rho T-e^{-\rho T}}.
\end{align*}
Then $W_T^{\left(N\right)} = 0$, and for $0\le t< T$ the sequence $( W_t^{\left(2N\right)})_{N \in \mathbb{N}}$ has exactly the two cluster points  $\varphi_+\left(t\right)$ and $\varphi_-(t)$, and $(\left.W_t^{\left(2N+1\right)}\right.)_{N \in \mathbb{N}}$ has exactly the two cluster points  $\psi_+\left(t\right)$ and $\psi_-(t)$.
\end{enumerate}
	 	\end{theorem}
		
			\begin{figure}
 \centering
 \begin{minipage}[b]{8cm}
 \begin{overpic}[width=8cm]{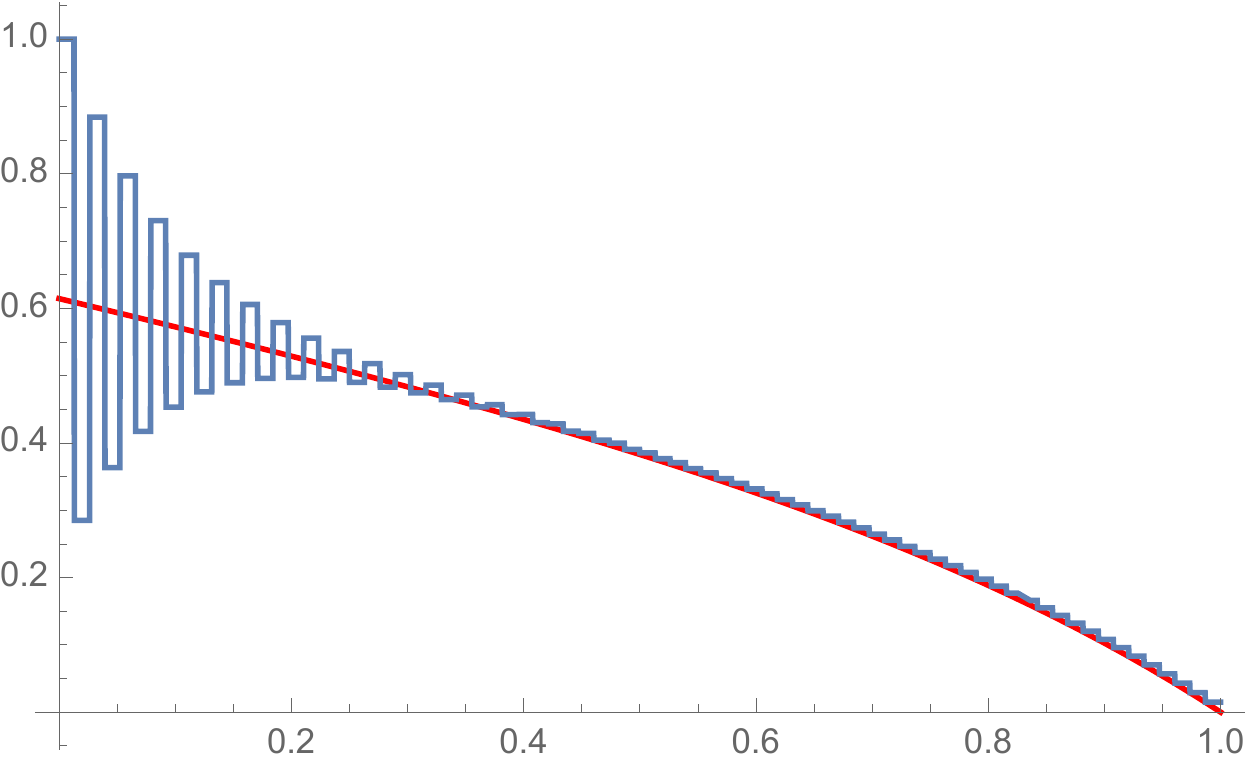}
\put(-8,60){\small$V^{(50)}_t$}
 \put(100,-3){\small$t$}
 \end{overpic}\end{minipage}
 \begin{minipage}[b]{8cm}
 \begin{overpic}[width=8cm]{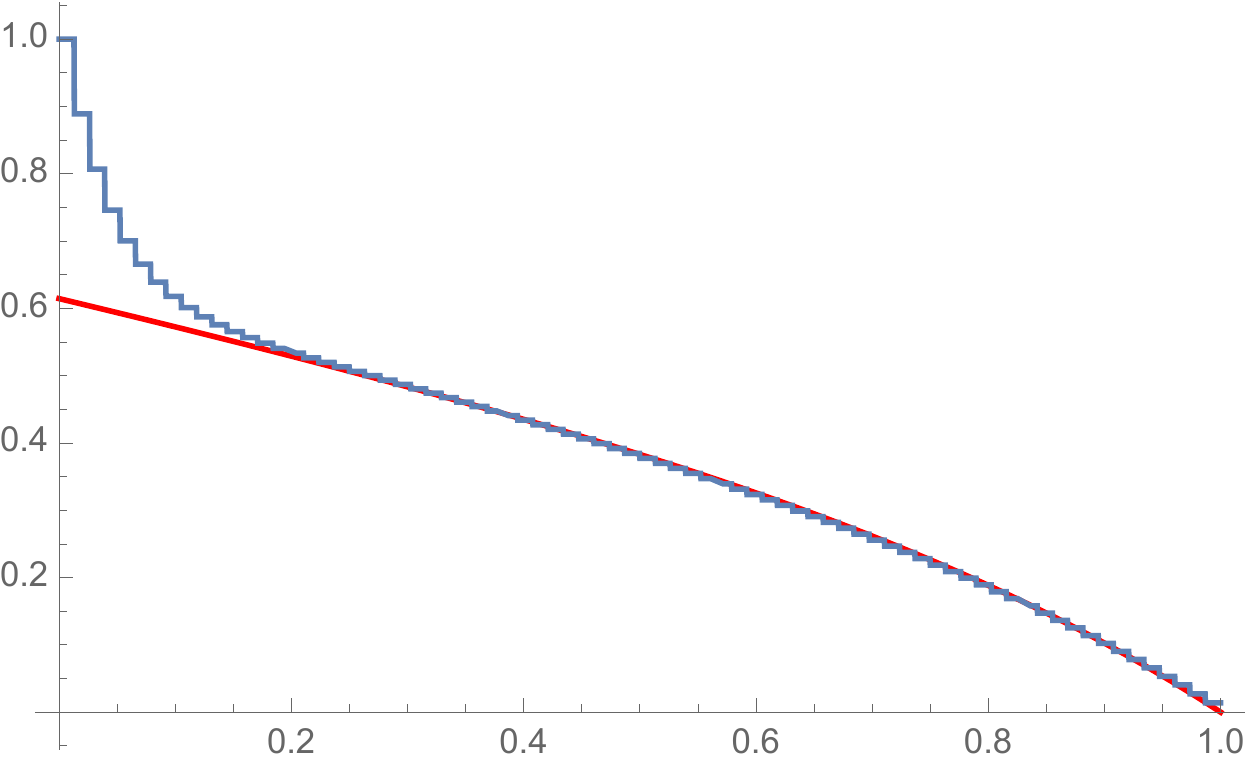}
\put(-8,60){\small$V^{(50)}_t$}
 \put(100,-3){\small$t$}
\end{overpic}\end{minipage}

\caption{Convergence of $V^{(75)}_t$ (blue) to $V$  (red) for small but nonzero $\theta$  (left) and large $\theta$ (right).}\end{figure}

		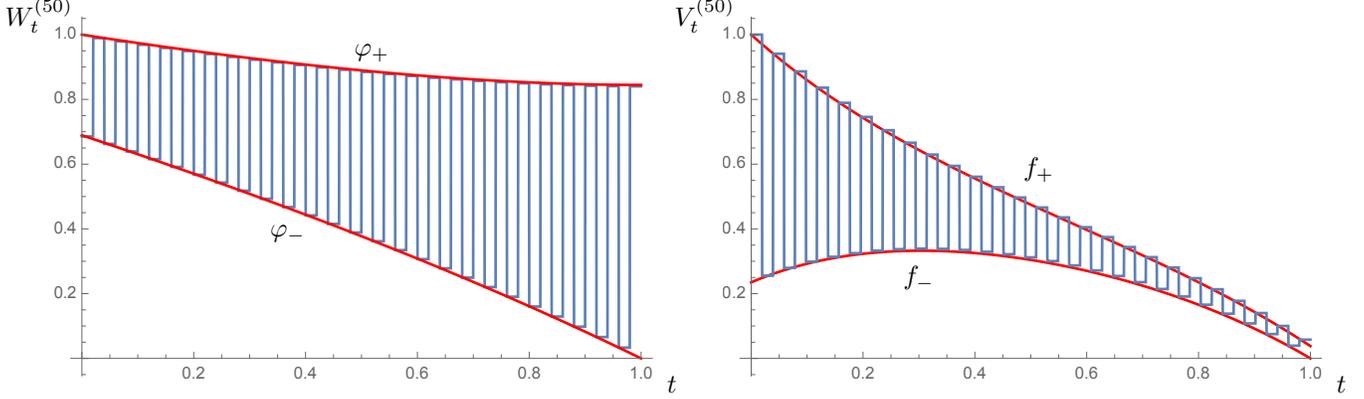
\begin{figure}
		\begin{center}
		\begin{minipage}[b]{8cm}
		\begin{overpic}[width=8cm]{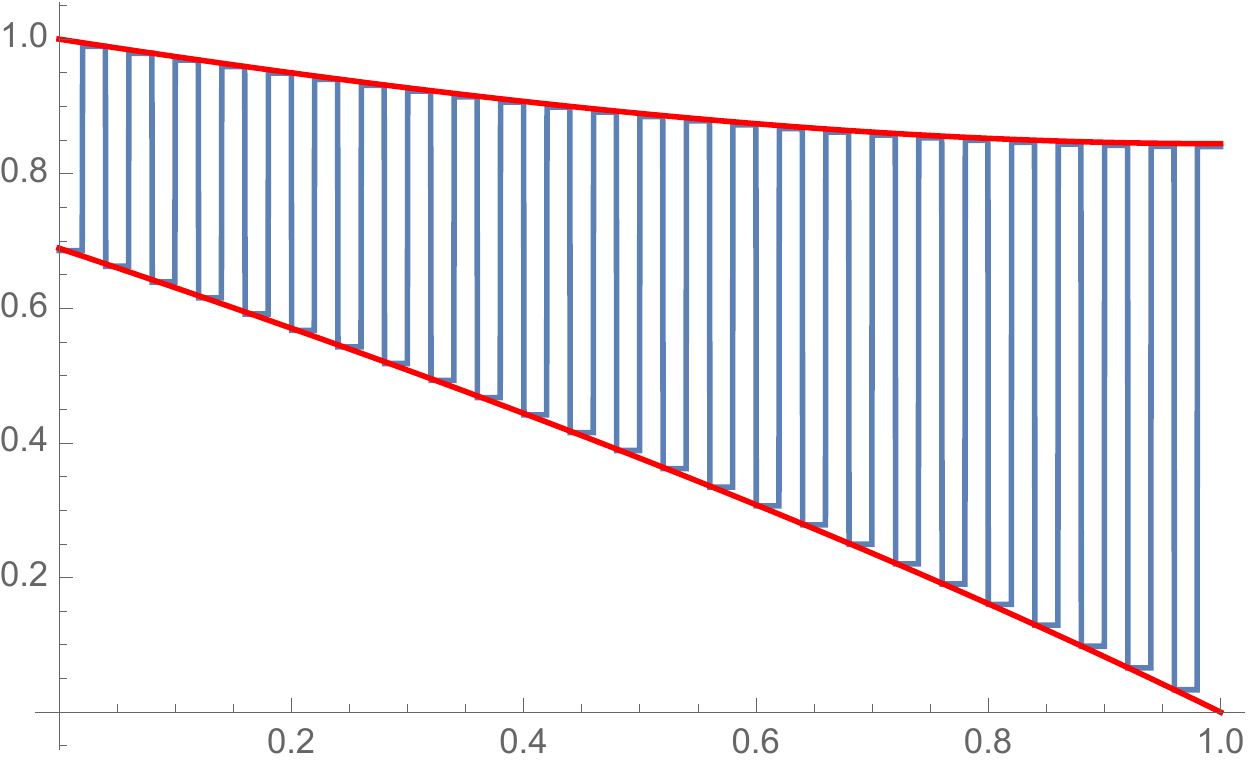}
		\put(-8,60){\small$W^{(50)}_t$}
		\put(102,-1){\small${t}$}
		\put(50,55){\small$ \varphi_+$}
	\put(36,25){\small$\varphi_-$}
		\end{overpic}
		\end{minipage}	\qquad
		\begin{minipage}[b]{8cm}
		\begin{overpic}[width=8cm]{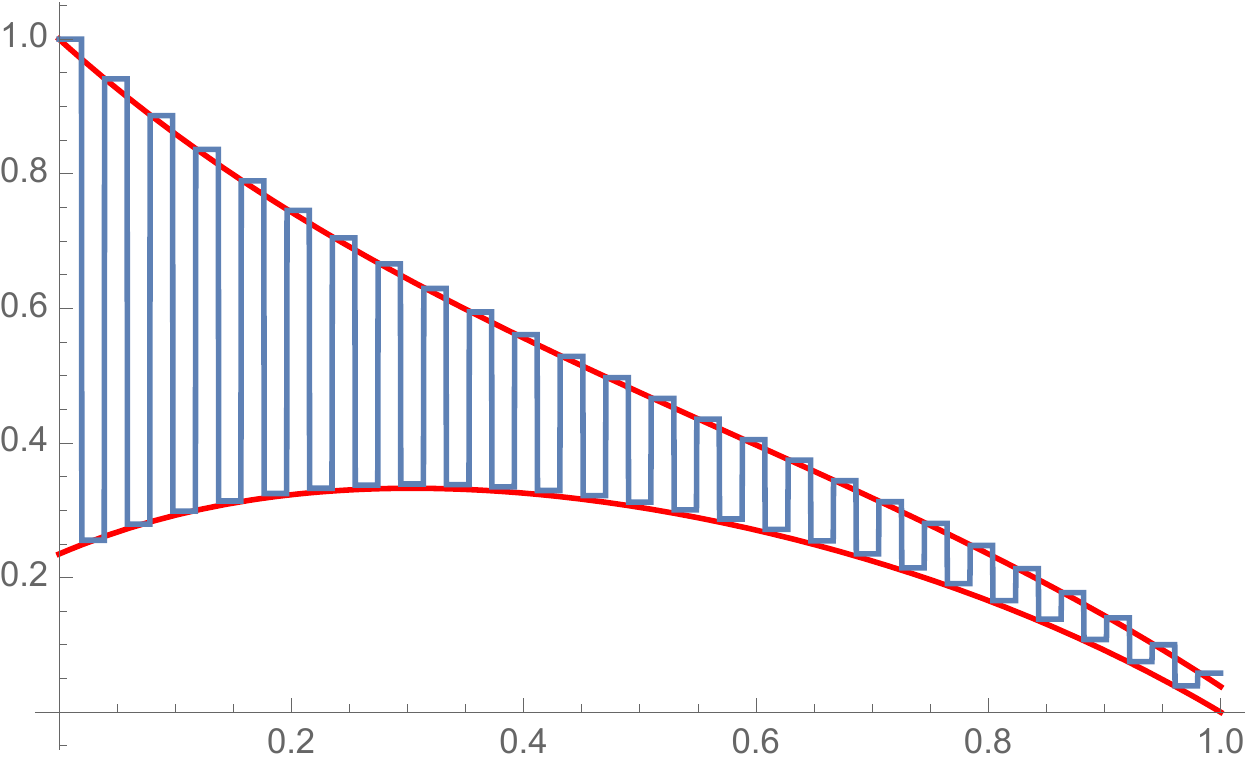}
		\put(-8,60){\small$V^{(50)}_t$}
		\put(102,-1){\small${t}$}
	\put(50,35){\small$f_+$}
	\put(30,17){\small$f_-$}

		\end{overpic}
		\end{minipage}		
		\caption{$W^{(50)}$ (left) and $V^{(50)}$ (right) in blue for $\theta=0$, together with the corresponding limits of oscillations from Theorem~\ref{limit strategies kappa thm} (b) and (d) in red.}
		\end{center}
		\end{figure}

		Note that the limits in parts (a) and (c) of the preceding theorem are independent of the particular value of $\theta>0$. A similar effect occurs in the asymptotics of the expected costs, as can be seen from part (a) of the following result.

\begin{theorem}[\bfseries Asymptotics of expected costs]\label{costs asymptotics thm}
For $x,y\in\mathbb{R}$ and $N\ge2$, let $\bm\xi^{(N)}\in\mathscr{X}(x,\mathbb{T}_N)$ and $\bm\eta^{(N)}\in\mathscr{X}(y,\mathbb{T}_N)$ be the corresponding equilibrium strategies. 
\begin{enumerate}
\item If $\theta>0$, then	
\begin{align*}
		\lim_{N\uparrow\infty}\mathbb{E}\left[\mathscr{C}_{\mathbb{T}_N}(\bm\xi^{(N)}|\bm\eta^{(N)})\right] =\frac{\left(x+y\right)^2\left(36e^{6\rho T}\left(8\rho T+13\right)-60e^{3\rho T}-3\right)}{16\left(2e^{3\rho T}\left(3\rho T+5\right)-1\right)^2}+\frac{x^2-y^2}{2\left(\rho T+1\right)}+\frac{\left(x-y\right)^2}{16\left(\rho T+1\right)^2}.
	\end{align*}
	\item If $\theta=0$, then
	\begin{align*}
		\lim_{N\uparrow\infty}\mathbb{E}\left[\mathscr{C}_{\mathbb{T}_{2N}}(\bm\xi^{(2N)}|\bm\eta^{(2N)})\right]=  \frac{\left(x+y\right)^2 \left(6e^{6\rho T}+3\right)}{2\left(2e^{6\rho T}\left(3\rho T+5\right)+e^{3\rho T}+3\rho T+7\right)}+\frac{x^2-y^2}{2\left(e^{-\rho T}+\rho T+1\right)},
	\end{align*}
	and
	\begin{align*}
		\lim_{N\uparrow\infty}\mathbb{E}\left[\mathscr{C}_{\mathbb{T}_{2N+1}}(\bm\xi^{(2N+1)}|\bm\eta^{(2N+1)})\right]=\frac{\left(x+y\right)^2 \left(6e^{6\rho T}-3\right)}{2\left(2e^{6\rho T}\left(3\rho T+5\right)-3e^{3\rho T}-3\rho T-7\right)}+\frac{x^2-y^2}{2\left(-e^{-\rho T}+\rho T+1\right)}.
	\end{align*}
	\end{enumerate}
\end{theorem}

With the cost limits stated in the preceding theorem, we can now  prove in a  mathematically rigorous fashion that increasing the transaction costs level $\theta$ can sometimes lower the expected costs of all market participants. This statement is made precise in the following corollary and illustrated in Figure~\ref{cost regions figure}.

\begin{corollary}\label{cost comparison cor}
	Let $x=y$ and $\rho T>\log(\left.4+\sqrt{62}/2\right.)/3 \approx 0.69$. Then the limit of the expected costs $\mathbb{E} [\mathscr{C}_{\mathbb{T}_N}(\bm\xi^{(N)}|\bm\eta^{(N)}) ]$ for 
	$\theta>0$ is strictly lower   than the limit inferior of the expected costs for $\theta=0$.
\end{corollary}

 \begin{figure}
 \begin{center}
 \includegraphics[width=6cm]{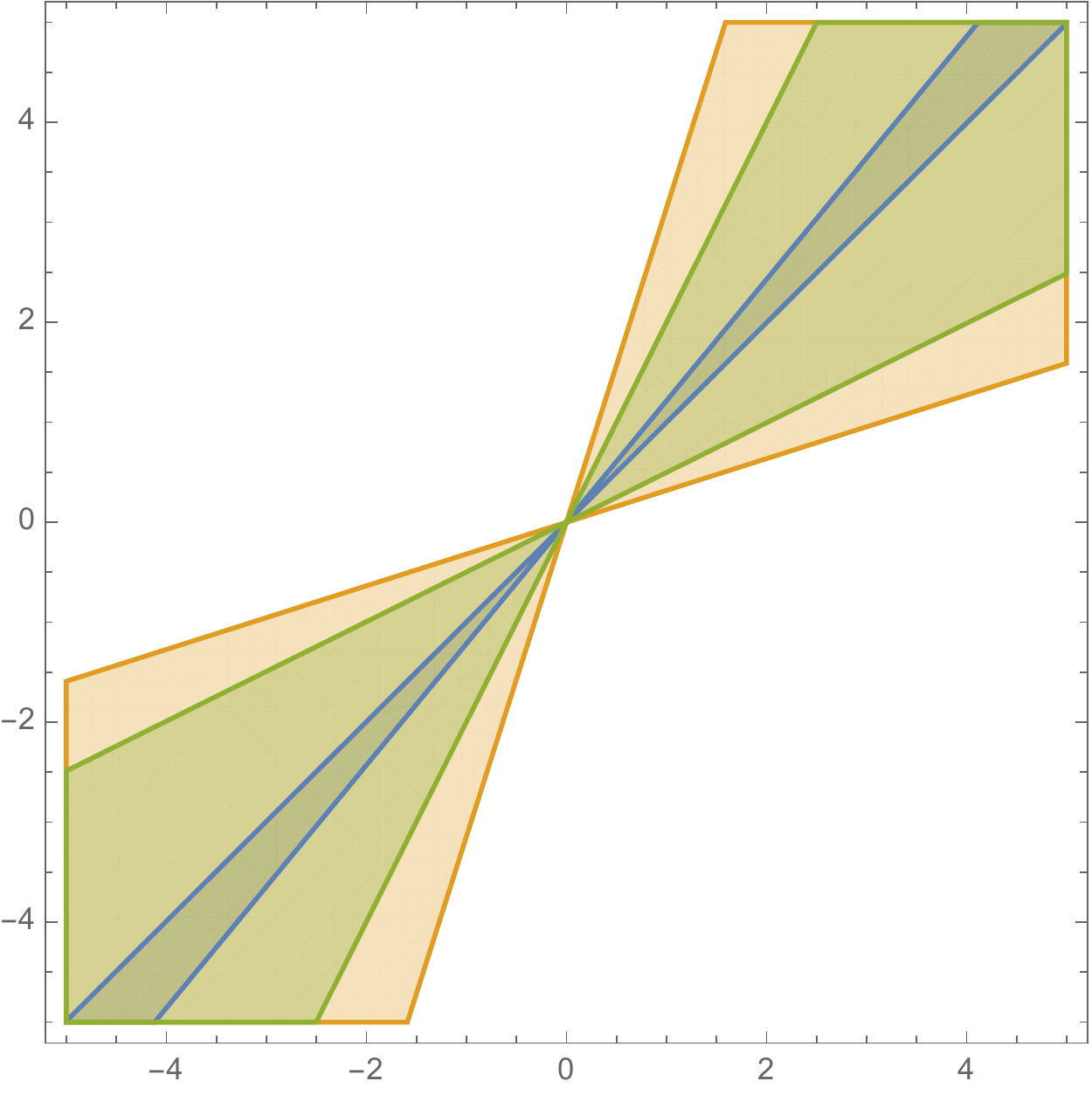}
\caption{Regions of $(x,y)\in[-5,5]\times[-5,5]$ for which the limit of the expected costs $\mathbb{E} [\mathscr{C}_{\mathbb{T}_N}(\bm\xi^{(N)}|\bm\eta^{(N)}) ]$ for $\theta=\theta^*=1/4$ is strictly lower   than the limit inferior of the expected costs if $\theta=0$, for $\rho T=0.69$ (blue), $\rho T=3$ (green), and $\rho T=6$ (orange).}
\end{center}
\label{cost regions figure}
 \end{figure}

%

In our analysis, we have considered the high-frequency limit $N\uparrow\infty$ while taking all other model parameters, notably $\rho$, as constant. In reality, one might expect that an increase of the trading frequency also leads to a more resilient market. That is, $\rho$ might increase with $N$. Note that Corollary~\ref{cost comparison cor} requires only a condition of the form $\rho >c$ for some constant $c$. Therefore we can expect  the conclusion of Corollary~\ref{cost comparison cor} to be robust if the framework is replaced by a variant in which $\rho$ increases with $N$. Moreover, it was observed numerically in~\cite{SchiedZhangHotPotato} that, for fixed and finite $N$, the expected costs $\mathbb{E} [\mathscr{C}_{\mathbb{T}_N}(\bm\xi^{(N)}|\bm\eta^{(N)}) ]$ can be a decreasing function of $\theta$ as long as $\theta$ remains sufficiently small. 

The quadratic transaction costs in our model can also be interpreted as a quadratic transaction tax, with $\theta \ge 0$ acting as the tax rate.
Corollary~\ref{cost comparison cor} shows that such a tax can sometimes make everyone better off and also generate some additional tax revenue. To analyze the tradeoff between the generation of tax revenues and additional costs in greater detail, we introduce the following definitions.

\begin{definition}If $N\in\mathbb{N}$, the tax rate is $\theta$, and  $\bm\xi^{(N)}\in\mathscr{X}(x,\mathbb{T}_N)$ and $\bm\eta^{(N)}\in\mathscr{X}(y,\mathbb{T}_N)$ are the corresponding equilibrium strategies, the \emph{total tax revenues} are
$$
	\text{TR}_N \coloneqq \theta\big( \bm\xi^{(N)}\big)^\top \bm\xi^{(N)} + \theta\big(\bm\eta^{(N)}\big)^\top \bm\eta^{(N)} $$
The total costs are defined as
$$C_N(\theta):=\mathbb{E} [\mathscr{C}_{\mathbb{T}_N}(\bm\xi^{(N)}|\bm\eta^{(N)}) ]+\mathbb{E} [\mathscr{C}_{\mathbb{T}_N}(\bm\eta^{(N)}|\bm\xi^{(N)}) ],
$$
and the \emph{total taxation costs} are 
$$\text{TC}_N:=C_N(\theta)-C_N(0).
$$
\end{definition}

\medskip 

The asymptotic behavior of the total taxation costs $\text{TC}_N$ is known from Theorem~\ref{costs asymptotics thm}. 
For the asymptotic behavior of the total tax revenues, we have the   following result. It shows that the tax revenues are asymptotically independent of the tax rate $\theta$ and that they  dominate the total taxation costs; see also Figure~\ref{tax figure}. In this sense, it is beneficial in our model  to levy a small transaction tax.
\begin{figure}\centering
\begin{overpic}[width=8cm]{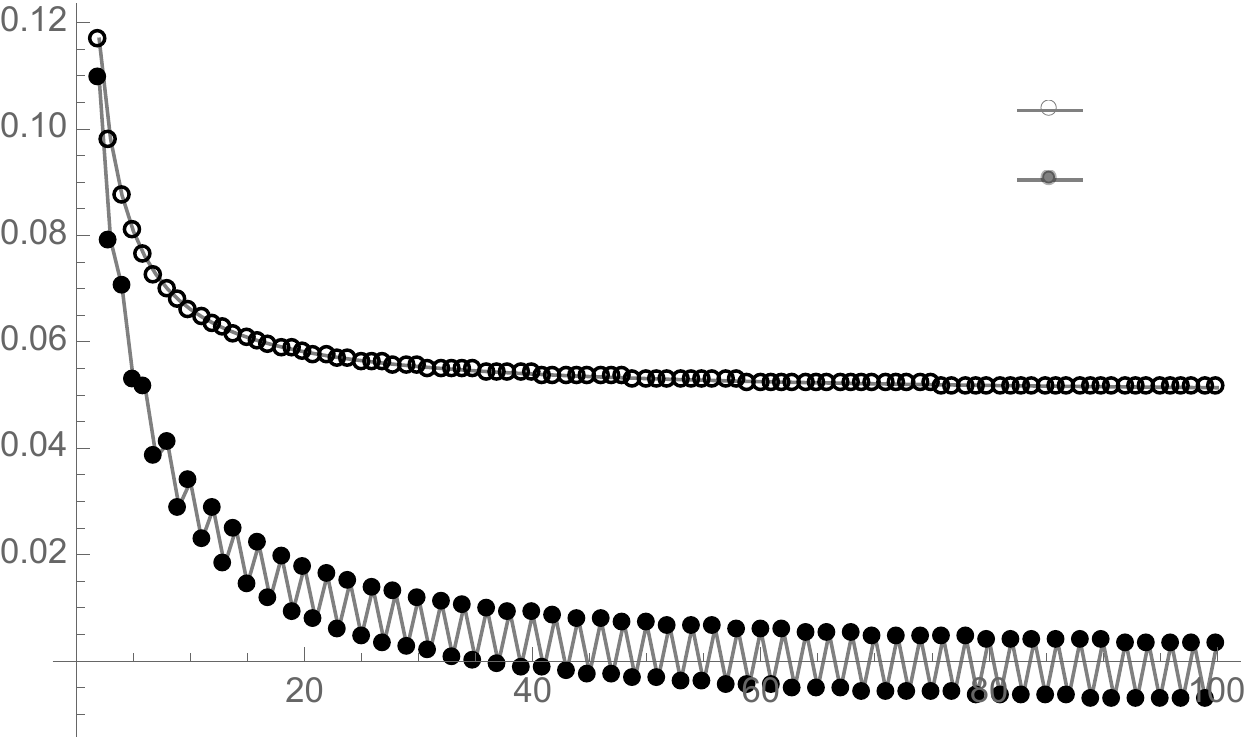}
\put(88,49){\small$\text{TR}_N$}
\put(88,43.5){\small$\text{TC}_N$}
\end{overpic}
\caption{Total tax revenues, $\text{TR}_N$, and total taxation costs, $\text{TC}_N$, for $N=2,\dots, 100$ and $x=1$, $y=1/2$, $\rho T=1$, and $\theta=\theta^*=1/4$. }
\label{tax figure}
\end{figure}

\begin{corollary}\label{tax corollary}For $\theta>0$ and initial positions $x,y\in\mathbb{R}$,
$$\text{\rm TR}_N\longrightarrow \frac{ (x+y)^2\,9(1+2e^{3\rho T})^2}{8(1-2e^{3\rho T}(5+3\rho T))^2}+  \frac{(x-y)^2}{8(\rho T+1)^2}\qquad\text{as $N\uparrow\infty$.}
$$
Moreover,
\begin{equation*}
\liminf_{N\uparrow\infty}(\text{\rm TR}_N-\text{\rm TC}_N)=\frac{(x+y)^2\,3 \left(2 e^{3 \rho T }+1\right)^2 \left(3 (\rho T +3)+2 e^{6 \rho T } (3 \rho T +5)-e^{3
   \rho T } (12 \rho T +19)\right)}{2 \left(1-2 e^{3 \rho T } (3 \rho T +5)\right)^2 \left(3 \rho T
   +e^{3 \rho T }+2 e^{6 \rho T } (3 \rho T +5)+7\right)}
\end{equation*}
and this expression is always nonnegative and strictly positive if $x\neq-y$.
\end{corollary}

	  \section{Nash equilibrium in continuous time}\label{Continuous section}
 
 In view of the convergence results from the preceding section, it is natural to ask whether the obtained limits are possibly related to a Nash equilibrium in a continuous-time extension of our model. To this end, we now introduce a continuous-time version of the primary model. Previous versions of this section's  statements and their proofs were first stated in the third author's doctoral  dissertation~\cite{ZhangThesis}. 
 
 \subsection{Definition of admissible strategies and Nash equilibria}
 
 We first define admissible strategies in continuous time. Various definitions of such strategies have been given in the literature~\cite{Gatheral,GSS,PredoiuShaikhetShreve,LorenzSchied, AS12, AlfonsiBlanc}; here we use the one from 
~\cite{LorenzSchied}, where strategies are right-continuous but may jump immediately at time $t=0$ and thus need  a starting value $Z_{0-}=z$ immediately before time $t=0$.
 
\begin{definition}\label{def ad cont strategy}
A strategy $(Z_t)_{t\ge0-}$ is called \emph{admissible}, if it satisfies the following conditions:
\begin{itemize}
\item  $(Z_t)_{t\ge0}$ is adapted to the filtration $(\mathscr{F}_t)_{t\ge0}$;
  \item the function $t\mapsto Z_t$ is $\mathbb{P}$-a.s.~right-continuous and bounded;
  \item the function $t\mapsto Z_t$ has finite and $\mathbb{P}$-a.s.~bounded total variation;
  \item there exists $T>0$ such that $Z_t=0$ $\mathbb{P}$-a.s.~for all $t\ge T$.
\end{itemize}
We denote the class of all admissible strategies $Z$ with initial value $Z_{0-}=z$ and time horizon $T$ by $\mathscr{X}(z,[0,T])$.
\end{definition}

 If two agents use the admissible strategies $X$ and $Y$, the \emph{affected price} $S^{X,Y}$ is defined as
 $$
S_t^{X,Y}=S_t^0+ \int_{[0,t)}e^{-\rho(t-s)}\,d X_s+ \int_{[0,t)}e^{-\rho(t-s)}\,d Y_s,
$$
where the integrals are Stieltjes integrals.
As in~\cite{LorenzSchied}, our definition of  the liquidation costs in a two-player setting will be motivated by a discrete-time approximation. To this end, let $X\in\mathscr{X}(x,[0,T])$, $Y\in\mathscr{X}(y,[0,T])$, and $N\in\mathbb{N}$ be   given. For $t^N_k:=kT/N\in\mathbb{T}_N$, we define the following discretized trades
\begin{equation*}
\begin{split}
\xi^N_0:=X_0-X_{0-}\quad&\text{ and }\quad\xi^N_k:=X_{t^N_k}-X_{t^N_{k-1}}\text{ for }k=\{1,2,\dots,N\};\\
\eta^N_0:=Y_0-Y_{0-}\quad&\text{ and }\quad\eta^N_k:=Y_{t^N_k}-Y_{t^N_{k-1}}\text{ for }k=\{1,2,\dots,N\}.
\end{split}
\end{equation*}
Then  $\bm\xi^N\in\mathscr{X}(x,\mathbb{T}_N)$ and $\bm\eta^N\in\mathscr{X}(y,\mathbb{T}_N)$. Furthermore, for each $N\in\mathbb{N}$ let $(\eps^N_k)_{k\in\{0,1,\dots,N\}}$ be a sequence of i.i.d.~Bernoulli\,$(\frac12)$-distributed random variables that are independent of $\sigma(\bigcup_{t\ge0}\mathscr{F}_t)$. According to Definition~\ref{discrete strategies cost def}, we have,
\begin{equation*}
\begin{split}
\mathscr{C}_{\mathbb{T}_N}(\bm\xi^N|\bm\eta^N):=&xS^0_{0-}+\sum_{k=0}^N\Big(\frac1{2}(\xi^N_k)^2-S^{\bm\xi^N,\bm\eta^N}_{t^N_k}\xi^N_k+\eps^N_k \xi^N_k\eta^N_k+\theta(\xi^N_k)^2\Big),\\
\mathscr{C}_{\mathbb{T}_N}(\bm\eta^N|\bm\xi^N):=&yS^0_{0-}+\sum_{k=0}^N\Big(\frac1{2}(\eta^N_k)^2-S^{\bm\xi^N,\bm\eta^N}_{t^N_k}\eta^N_k+(1-\eps^N_k) \xi^N_k\eta^N_k+\theta(\eta^N_k)^2\Big).
\end{split}
\end{equation*}
In the following lemma, we obtain the convergence of the expected liquidation costs. Its proof is analogous to the one of \cite[Lemma 1]{LorenzSchied} and hence omitted. 

\begin{lemma}\label{conv exp costs}
As $N\uparrow\infty$, we have
\begin{align*}
\mathbb{E}[\mathscr{C}_{\mathbb{T}_N}(\bm\xi^N|\bm\eta^N)]\longrightarrow&\mathbb{E}\bigg[\frac12\int_{[0,T]}\int_{[0,T]}  e^{-\rho|t-s|}\,d X_s\,d X_t+\int_{[0,T]}\int_{[0,t)}  e^{-\rho(t-s)}\,d Y_s\,d X_t\label{continuous-time ciost limit 1}\\
&\qquad+\frac1{2}\sum_{t\in[0,T]}\Delta X_t\Delta Y_t+\theta\sum_{t\in[0,T]}(\Delta X_t)^2\bigg].
\end{align*}
\end{lemma}

Clearly, by interchanging the roles of $X$ and $Y$ in the preceding lemma, we obtain the convergence of the expected costs $\mathbb{E}[\mathscr{C}_{\mathbb{T}_N}(\bm\eta^N|\bm\xi^N)]$. Motivated by this lemma, we can now state the following definition.

\begin{definition}
Given initial asset positions $x,y\in\mathbb{R}$ and $T>0$, the \emph{liquidation costs} of $X\in\mathscr{X}(x,[0,T])$ given $Y\in\mathscr{X}(y,[0,T])$ are defined as
\begin{align*}
\mathscr{C}(X|Y)&=\frac12\int_{[0,T]}\int_{[0,T]}  e^{-\rho|t-s|}\,d X_s\,d X_t+\int_{[0,T]}\int_{[0,t)}  e^{-\rho(t-s)}\,d Y_s\,d X_t\\
&\qquad+\frac1{2}\sum_{t\in[0,T]}\Delta X_t\Delta Y_t+\theta\sum_{t\in[0,T]}(\Delta X_t)^2.
\end{align*}
\end{definition}

We can now define the concept of a Nash equilibrium in this continuous-time setting.

\begin{definition}\label{Nash Def continuous}For a given time horizon $[0,T]$ and initial asset positions $x$, $y\in\mathbb{R}$, a \emph{Nash equilibrium} is a pair $(X^*,Y^*)$ of strategies in $\mathscr{X}(x,[0,T])\times\mathscr{X}(y,[0,T])$ such that 
$$\mathbb{E}[\,\mathscr{C}(X^*|Y^*)\,]=\inf_{X\in\mathscr{X}(x,[0,T])}\mathbb{E}[\,\mathscr{C}(X|Y^*)\,]\quad\text{and}\quad \mathbb{E}[\,\mathscr{C}(Y^*|X^*)\,]=\inf_{Y\in\mathscr{X}(y,[0,T])}\mathbb{E}[\,\mathscr{C}(Y|X^*)\,].
$$
\end{definition}

\subsection{Existence and nonexistence of a Nash equilibrium}

The following result gives a complete solution to the questions of existence, uniqueness, and characterization of Nash equilibria in continuous time. It states in particular that a Nash equilibrium exists for all initial positions $x,y\in\mathbb{R}$ if and only  $\theta$ is equal to the critical value $\theta^*=1/4$; in this case, the Nash equilibrium is unique and given by the continuous-time limits of equilibrium strategies derived in Theorem~\ref{limit strategies kappa thm}  for $\theta>0$.

\begin{theorem}\label{main th continuous}
Let $\rho>0$, $T>0$, and $x,y\in\mathbb{R}$ be given.
\begin{enumerate}
\item For $\theta=\theta^*=1/4$, there exists a unique Nash equilibrium $(X^*,Y^*)$ in the class $\mathscr{X}(x,[0,T])\times\mathscr{X}(y,[0,T])$ of adapted strategies.  The optimal strategies $X^*$ and $Y^*$ are deterministic and given by
\begin{equation}\label{equilibrium strategies continuous}
X^*_t=\frac12(x+y)V_t+\frac12(x-y)W_t\qquad\text{and}\qquad
Y^*_t=\frac12(x+y)V_t-\frac12(x-y)W_t,
\end{equation}
where 
\begin{equation*}
\begin{split}
&V_t=\frac{e^{3\rho T}\big(6\rho(T-t)+4\big)-4e^{3\rho t}}{2 e^{3\rho T}(3\rho T+5)-1}\quad\text{if $t\in[0,T]$, and }V_{0-}=1,\\
&W_t=\frac{\rho(T-t)+1}{\rho T+1}\quad \text{if $t\in[0,T)$, $W_{0-}=1$, and  $W_T=0$.}
\end{split}
\end{equation*}
\item For $\theta\neq\theta^*$, a Nash equilibrium exists if and only if $x = y = 0$, in which case the Nash equilibrium is unique and the equilibrium strategies vanish identically.
\end{enumerate} 
\end{theorem}


A natural question in view of Theorem~\ref{main th continuous} (a) is whether the convergence of the discrete-time expected costs, as observed in Theorem~\ref{costs asymptotics thm} for $\theta>0$, can   be interpreted as the convergence toward the costs of the continuous-time equilibrium. The following corollary answers this question in the affirmative.

\begin{corollary}\label{cont costs cor}For $x,y\in\mathbb{R}$, let $X^*$ and $Y^*$ be as in \eqref{equilibrium strategies continuous} and assume that $\theta=\theta^*=1/4$. Then $\mathscr{C}(X^*|Y^*)$ is equal to the limit of the discrete-time expected costs in Theorem~\ref{costs asymptotics thm} {\rm(a)}.
\end{corollary}

\subsection{Uniqueness and first-order conditions}

In this section, we analyze Nash equilibria in continuous time. In particular, we will prove the uniqueness of Nash equilibria, show that Nash equilibria in the class of deterministic strategies are also Nash equilibria in the larger class of adapted strategies, and provide a first-order condition for optimal strategies in a Nash equilibrium extending the Fredholm integral equations derived in~\cite{GSS}.
For admissible strategies $X$ and $Y$, we define the following expressions,
\begin{equation*}
\begin{split}
&C( X, Y):=\mathbb{E}\bigg[\int_{[0,T]}\int_{[0,T]}  e^{-\rho|t-s|}\,d Y_s\,d X_t\bigg],\\
&C_1(  X,Y):=\mathbb{E}\bigg[\int_{[0,T]}\int_{[0, t)}  e^{-\rho(t-s)}\,d Y_s\,d X_t\bigg],\quad C_2(  X,Y):=\mathbb{E}\bigg[\sum_{t\in[0,T]}\Delta X_t\Delta Y_t\bigg].
\end{split}
\end{equation*}
Then,
\begin{equation}\label{continuous-times costs shorthand decomposition}
\mathbb{E}\big[\mathscr{C}(X|Y)\big]=\frac{1}{2}C(X,X)+C_1(X,Y)+\frac1{2}C_2(X,Y)+\theta C_2(X,X).
\end{equation}
A first step toward proving the uniqueness of Nash equilibria is the strict convexity of the map $X\mapsto\mathbb{E}[\mathscr{C}(X|Y)]$, which is established in the following lemma.

\begin{lemma}\label{convex lm}
Given $T>0, \rho>0$, $\theta\ge0$, initial asset positions $x, y\in\mathbb{R}$ and an admissible strategy $Y\in\mathscr{X}(y,[0,T])$, the functional $\mathbb{E}[\mathscr{C}(X|Y)]$ is strictly convex with respect to $X\in\mathscr{X}(x,[0,T])$.
\end{lemma}

\begin{proof}
Let $\alpha\in(0,1)$ and $X^0, X^1\in\mathscr{X}(x,[0,T])$ be two distinct admissible strategies. Since the function $t\mapsto   e^{-\rho t}$ is positive definite in the sense of Bochner, we obtain
\begin{equation}\label{pos def Bochner costs eq}
C(X^1-X^0,X^1-X^0)=\mathbb{E}\bigg[\int_{[0,T]}\int_{[0,T]}  e^{-\rho|s-t|}\,d\big(X^1_s-X_s^0\big)\,d\big(X^1_t-X_t^0\big)\bigg]>0;
\end{equation}
see \cite[Proposition 2.6]{GSS}.
Hence,
\begin{equation*}
\begin{split}
\lefteqn{C(\alpha X^1+(1-\alpha)X^0,\alpha X^1+(1-\alpha)X^0)}\\
&<C(\alpha X^1+(1-\alpha)X^0,\alpha X^1+(1-\alpha)X^0)+\alpha(1-\alpha)C(X^1-X^0,X^1-X^0)\\
&=\alpha^2C(X^1,X^1)+(1-\alpha)^2C(X^0,X^0)+2\alpha(1-\alpha)C(X^0,X^1)\\
&\qquad+\alpha(1-\alpha)C(X^1,X^1)-2\alpha(1-\alpha)C(X^0,X^1)+\alpha(1-\alpha)C(X^0,X^0)\\
&=\alpha C(X^1,X^1)+(1-\alpha)C(X^0,X^0).
\end{split}
\end{equation*}
Moreover, $C_1(X,Y)$ and $C_2(X,Y)$ are clearly affine in $X$, whereas $C_2(X,X)$ is convex. The result thus follows from \eqref{continuous-times costs shorthand decomposition}.
\end{proof}

 We can now establish the uniqueness of Nash equilibria.

\begin{proposition}\label{unique prop}
Given $T>0, \rho>0$, $\theta\ge0$ and initial asset positions $x, y\in\mathbb{R}$, there exists at most one Nash equilibrium in the class $\mathscr{X}(x,[0,T])\times\mathscr{X}(y,[0,T])$ of adapted strategies.
\end{proposition}

\begin{proof}We use a similar reasoning as in \cite[Lemma 3.3]{SchiedZhangHotPotato} and \cite[Lemma 4.1]{SchiedZhangCARA}. We assume by way of contradiction that there exist two distinct Nash equilibria $(X^0,Y^0)$ and $(X^1,Y^1)$ in $\mathscr{X}(x,[0,T])\times\mathscr{X}(y,[0,T])$. Then we define for $\alpha\in[0,1]$
$$
X^{\alpha}:=\alpha X^1+(1-\alpha)X^0\qquad\text{and}\qquad Y^{\alpha}:=\alpha Y^1+(1-\alpha)Y^0.
$$
We furthermore let
$$
f(\alpha):=\mathbb{E}\Big[\mathscr{C}(X^{\alpha}|Y^0)+\mathscr{C}(Y^{\alpha}|X^0)+\mathscr{C}(X^{1-\alpha}|Y^1)+\mathscr{C}(Y^{1-\alpha}|X^1)\Big].
$$
According to Lemma~\ref{convex lm} and the assumption that the two Nash equilibria  $(X^0,Y^0)$ and $(X^1,Y^1)$ are distinct,  $f(\alpha)$ is strictly convex in $\alpha$ and thus has a unique minimum at $\alpha=0$. It follows that
\begin{equation}\label{contradiction >}
\lim_{h\downarrow0}\frac{f(h)-f(0)}{h}=\frac{df(\alpha)}{d\alpha}\Big|_{\alpha=0+}\ge0.
\end{equation}
On the other hand, we have
\begin{equation*}
\begin{split}
\frac{d}{d\alpha}\Big|_{\alpha=0}\mathbb{E}[\mathscr{C}(X^{\alpha}|Y^0)]&=C(X^1-X^0,X^0)+C_1(X^1-X^0,Y^0)\\
&+\frac1{2}C_2(X^1-X^0,Y^0)+2\theta C_2(X^1-X^0,X^0).
\end{split}
\end{equation*}
Taking derivatives of $\mathbb{E}[\mathscr{C}(Y^{\alpha}|X^0)]$, $\mathbb{E}[\mathscr{C}(X^{1-\alpha}|Y^1)]$ and $\mathbb{E}[\mathscr{C}(Y^{1-\alpha}|X^1)]$ in the same way gives
\begin{equation*}
\begin{split}
\lefteqn{\frac{df(\alpha)}{d\alpha}\Big|_{\alpha=0}}\\
&=-C(X^1-X^0,X^1-X^0)-C(Y^1-Y^0,Y^1-Y^0)-C(Y^1-Y^0,Y^1-Y^0)\\
&\qquad-2\theta \Big(C_2(X^1-X^0,X^1-X^0)+C_2(Y^1-Y^0,Y^1-Y^0)\Big)\\
&<-\frac12C(X^1-X^0,X^1-X^0)-\frac12C(Y^1-Y^0,Y^1-Y^0)-\frac12C(X^1-X^0,Y^1-Y^0)\\
&<0,
\end{split}
\end{equation*}
which contradicts \eqref{contradiction >}. \end{proof}

 The following lemma will allow us to concentrate on deterministic strategies when searching for a Nash equilibrium. It is similar to \cite[Lemma 3.4]{SchiedZhangHotPotato}.

\begin{lemma}\label{det is adp lm}
A Nash equilibrium in the class $\mathscr{X}_{\text{det}}(x,[0,T])\times\mathscr{X}_{\text{det}}(y,[0,T])$ of deterministic strategies is also a Nash equilibrium in the class $\mathscr{X}(x,[0,T])\times\mathscr{X}(y,[0,T])$ of adapted strategies.
\end{lemma}

\begin{proof}
Let $(X^*,Y^*)\in\mathscr{X}_{\text{det}}(x,[0,T])\times\mathscr{X}_{\text{det}}(y,[0,T])$ be a Nash equilibrium in the class of deterministic strategies. For any strategy $X\in\mathscr{X}(x,[0,T])$, we have 
$$
\mathscr{C}(X(\omega)|Y^*)\ge\mathscr{C}(X^*|Y^*)\quad\text{ for $\mathbb{P}$-almost all }\omega\in\Omega.
$$
Therefore we obtain $\mathbb{E}[\mathscr{C}(X|Y^*)]\ge\mathscr{C}(X^*|Y^*)$ with equality if and only if $\mathscr{C}(X|Y^*)=\mathscr{C}(X^*|Y^*)$ $\mathbb{P}$-a.s. This shows the optimality of $X^*$ within the class $\mathscr{X}(x,[0,T])$ of adaptive strategies. Analogously we obtain the optimality of $Y^*$ within the class $\mathscr{X}(y,[0,T])$ of adaptive strategies. This completes the proof.
\end{proof}

 We now derive   first-order conditions for the optimality of $\mathbb{E}[\mathscr{C}(X|Y)]$ within the class $\mathscr{X}(x,[0,T])$ and for given $Y\in\mathscr{X}(y,[0,T])$. The first result is the following proposition, which, for our special case of exponential decay of price impact, extends \cite[Theorem 2.11]{GSS}, where, for $Y=0$ and $\theta=0$, the optimality of $X\in \mathscr{X}_{det}(x,[0,T])$ is characterized in terms of  a Fredholm integral equation.

\begin{proposition}\label{deterministic fredholm prop}Let  $x,y\in\mathbb{R}$ and  $Y\in\mathcal{X}_{det}(y,T)$ be given. Then a strategy  $X^{*}\in\mathcal{X}_{det}(x, [0,T])$ minimizes the liquidation costs $\mathscr{C}(X|Y)$ over $X\in \mathscr{X}_{\text{det}}(x,[0,T])$, if  and only if there exists a constant $\eta\in\mathbb{R}$ such that for all $t\in[0,T]$,
\begin{equation}\label{eqfred}
\int_{[0,T]}  e^{-\rho|t-s|}\,d X^*_s +\int_{[0, t)}  e^{-\rho(t-s)}\,d Y_s +\frac{1}{2}\Delta Y_t+2\theta\Delta X^*_t=\eta.
\end{equation}
\end{proposition}

\begin{proof}Suppose first that $X^*$ minimizes the liquidation costs $\mathscr{C}(X|Y)$ over $X\in \mathscr{X}_{\text{det}}(x,[0,T])$. We fix $t_0,t\in[0,T]$ and define $Z\in\mathscr{X}_{\text{det}}(0,[0,T])$ by $Z_s=\mathbbmss{1}_{\{s\ge t_0\}}-\mathbbmss{1}_{\{s\ge t\}}$. Admissible strategies $Z$ with initial value $Z_{0-}=0$ are often called \lq\lq round trips". The optimality of $X^*$ implies that the function 
\begin{equation}\label{Cost expansion in alpha eq}
\begin{split}
f(\alpha)&:=\mathscr{C}(X^*+\alpha Z|Y)\\
&=\frac{1}{2}C(X^*,X^*)+C_1(X^*,Y)+ \frac1{2} C_2(X^*,Y)+\frac{\alpha^2}{2}C(Z,Z)+\theta C_2(X^*,X^*)\\
&\qquad+\alpha C(Z,X^*)+\alpha C_1(Z,Y)+ \frac\alpha {2} C_2(Z,Y)+2\alpha \theta C_2(Z,X^*)+\alpha^2\theta C_2(Z,Z)
\end{split}
\end{equation}
 has a minimum at $\alpha=0$. Here we have used the decomposition \eqref{continuous-times costs shorthand decomposition}. Thus,
 \begin{align*}
 0&=\frac{d f(\alpha)}{d\alpha}\Big|_{\alpha=0}=C(Z,X^*)+C_1(Z,Y)+\frac1 {2} C_2(Z,Y)+2\theta C_2(Z,X^*)\\
 &=\int_{[0,T]}  e^{-\rho|t_0-s|}\,d X^*_s-\int_{[0,T]}  e^{-\rho|t-s|}\,d X^*_s +\int_{[0, t_0)}  e^{-\rho(t_0-s)}\,d Y_s-\int_{[0, t)}  e^{-\rho(t-s)}\,d Y_s \\
 &\qquad+\frac{1}{2}\Delta Y_{t_0}-\frac{1}{2}\Delta Y_t+2\theta\Delta X^*_{t_0}-2\theta\Delta X^*_t.
 \end{align*}
 Thus, \eqref{eqfred} follows if we let 
$$\eta:=\int_{[0,T]}  e^{-\rho|t_0-s|}\,d X^*_s +\int_{[0, t_0)}  e^{-\rho(t_0-s)}\,d Y_s +\frac{1}{2}\Delta Y_{t_0}+2\theta\Delta X^*_{t_0}.
$$

  Conversely, we now assume that $X^*\in\mathcal{X}_{det}(x, [0,T])$ satisfies \eqref{eqfred} and prove that $X^*$ is optimal. To this end, we take an arbitrary \lq\lq round trip"  $Z\in\mathscr{X}_{\text{det}}(0,[0,T])$.  Using \eqref{Cost expansion in alpha eq} for $\alpha=1$ and the facts that 
$C_2(Z,Z)\ge0$ and $ C(Z,Z) \ge0$ by \eqref{pos def Bochner costs eq}, we have
\begin{equation*}
\begin{split}
\lefteqn{\mathscr{C}(X^*+Z|Y)}\\
&\ge\mathscr{C}(X^*|Y)+C(Z,X^*)+C_1(Z,Y)+ \frac1{2} C_2(Z,Y)+2\theta C_2(Z,X^*)\\
&=\mathscr{C}(X^*|Y)+\int_{[0,T]}\Big(\int_{[0,T]}  e^{-\rho|t-s|}\,d X_t +\int_{[0, s)}  e^{-\rho|t-s|}\,d Y_t+ \frac1{2}\Delta Y_s+2\theta\Delta X^*_s\Big)\,d Z_s\\
&=\mathscr{C}(X^*|Y)+\eta(Z_T-Z_0)=\mathscr{C}(X^*|Y).
\end{split}
\end{equation*}
Since every strategy $X\in\mathscr{X}_{\text{det}}(x,[0,T])$ can be written as $X^*+Z$ for some \lq\lq round trip" $Z\in\mathscr{X}_{\text{det}}(0,[0,T])$, we obtain the optimality of $X^*$.
\end{proof}

 The following proposition extends the necessity of the first-order condition for optimality to the case of  strategies that are not necessarily deterministic.
 
\begin{proposition}\label{random fredholm prop}
Let $Y\in\mathscr{X}(y,[0,T])$ be given.  If there exists an optimal strategy $X^*$ minimizing the expected liquidation costs $\mathbb{E}[\mathscr{C}(X|Y)]$ in $\mathscr{X}(x,[0,T])$, then for any $[0,T]$-valued stopping time $\tau$, there exists an $\mathscr{F}_{\tau}$-measurable random variable $\eta$ such that  for every stopping time $\sigma$ taking values in $[\tau,T]$,
\begin{equation}\label{Fredholm formula}
\mathbb{E}\bigg[\int_{[0,T]}  e^{-\rho|\sigma-t|}\,dX^*_t+\int_{[0,\sigma)}  e^{-\rho(\sigma-t)}\,dY_t+\frac1{2}\Delta Y_{\sigma}+2\theta\Delta X^*_{\sigma}\,\Big|\,\mathscr{F}_{\tau}\bigg]=\eta\qquad\text{$\mathbb{P}$-a.s.}
\end{equation}
As a matter of fact, we can take 
\begin{equation*}
\eta:=\mathbb{E}\bigg[\int_{[0,T]}  e^{-\rho|\tau-t|}\,dX^*_t+\int_{[0,\tau)}  e^{-\rho(\tau-t)}\,dY_t+\frac1{2}\Delta Y_{\tau}+2\theta\Delta X^*_{\tau}\,\Big|\,\mathscr{F}_{\tau}\bigg]
\end{equation*}
\end{proposition}

\begin{proof}For $\tau$ and $\sigma$ as in the assertion and $A\in\mathscr{F}_\tau$, we define a \lq\lq round trip" $Z\in\mathscr{X}(0,[0,T])$ by
$$
Z_t=\mathbbmss{1}_A\Big(\Ind{\{t\ge\tau\}}-\Ind{\{t\ge\sigma\}}\Big).
$$
Expanding the expected costs $\mathbb{E}[\mathscr{C}( X^*+\alpha  Z| Y)]$ as in \eqref{Cost expansion in alpha eq} and taking derivatives with respect to $\alpha$ at $\alpha=0$ yields the following necessary first-order condition for optimality,
\begin{equation}\label{nece exp condition}
0=C( Z, X^*)+C_1(  Z,Y)+\frac1{2}C_2(  Z,Y)+2\theta C_2( Z,X^*).
\end{equation}
By exploiting the specific form of $Z$, \eqref{nece exp condition} becomes
\begin{align*}
0=&\mathbb{E}\bigg[\mathbbmss{1}_A\Big(\int_{[0,T]}\left(  e^{-\rho|\tau-t|}-  e^{-\rho|\sigma-t|}\right)\,d X^*_t+\int_{[0, \tau)}  e^{-\rho|\tau-t|}\,d Y_t-\int_{[0,\sigma)}  e^{-\rho|\sigma-t|}\,d Y_t\\
&\qquad+\frac1{2}\left(\Delta Y_{\tau}-\Delta Y_{\sigma}\right)+2\theta\left(\Delta X^*_{\tau}-\Delta X^*_{\sigma}\right)\Big)\bigg].
\end{align*}
This implies that for all $A\in\mathscr{F}_{\tau}$,
\begin{align*}
\lefteqn{\mathbb{E}\bigg[\mathbbmss{1}_A\Big(\int_{[0,T]}  e^{-\rho|\sigma-t|}\,dX^*_t+\int_{[0,\sigma)}  e^{-\rho(\sigma-t)}\,dY_t+\frac1{2}\Delta Y_{\sigma}+2\theta\Delta X^*_{\sigma}\Big)\bigg]}\\
&\qquad=\mathbb{E}\bigg[\mathbbmss{1}_A\Big(\int_{[0,T]}  e^{-\rho|\tau-t|}\,dX^*_t+\int_{[0,\tau)}  e^{-\rho(\tau-t)}\,dY_t+\frac1{2}\Delta Y_{\tau}+2\theta\Delta X^*_{\tau}\Big)\bigg].
\end{align*}
Note that the right-hand side is independent of $\sigma$. Taking conditional expectations thus gives the result.\end{proof}

\section[]{Expressing the discrete equilibrium strategies in closed form}\label{closed form section}

In this section, our aim is to compute the vectors ${\bm\nu}$ and $\bm\omega$. Our corresponding result will be Theorem~\ref{omega and nu closed form thm}
 at the end of this section. It will be needed for the proofs of our asymptotic results.
 
 As in~\cite{SchiedZhangHotPotato}, it will be convenient to define
$$\alpha := e^{-\rho T/N}\qquad\text{and}\qquad \kappa := 2\theta + \frac12.$$
Note that we have $\kappa\ge1/2$ with equality if and only if $\theta=0$ and that the critical value $\theta^*$ corresponds to $\kappa=1$.

To compute $\bm\nu$, we define the matrix 
\begin{align*}
	B := (\left.1-\alpha^2\right.)\left(\Id+\Gamma^{-1} (\left.\tilde{\Gamma}+2\theta\Id\right.)\right).
\end{align*}
To obtain a more explicit representation of $B$, recall first that the inverse of the Kac--Murdock--Szeg\H o matrix $\Gamma$ has a simple tridiagonal structure and is given by
\begin{equation}\label{KMS inverse eq}
\Gamma^{-1}=\frac{1}{1-\alpha^{{2}}}\begin{pmatrix}
  \,&{1}&{-\alpha}&0&\cdots&\cdots&0 \,{} \\
     &{-\alpha}&{1+\alpha^2}&{-\alpha}&0&\cdots&0\,{} \\
     &0&\ddots&\ddots&\ddots&\ddots&\vdots\,{} \\
     &\vdots&\ddots&\ddots&\ddots&\ddots&\vdots\,{} \\
     &\vdots&\ddots&\ddots&{-\alpha}&{1+\alpha^2}&{-\alpha}\,{} \\
     &0&\cdots&\cdots&0&{-\alpha}&{1}\,{}
 \end{pmatrix};
\end{equation}
see, e.g., \cite[Section 7.2, Problems 12-13]{HornJohnson}. Thus, 
\begin{eqnarray*}
B&=&(1-\alpha^2)\Id+
 \begin{pmatrix}
  \,&{1}&{-\alpha}&0&\cdots&\cdots&0 \,{} \\
     &{-\alpha}&{1+\alpha^2}&{-\alpha}&0&\cdots&0\,{} \\
     &0&\ddots&\ddots&\ddots&\ddots&\vdots\,{} \\
     &\vdots&\ddots&\ddots&\ddots&\ddots&\vdots\,{} \\
     &\vdots&\ddots&\ddots&{-\alpha}&{1+\alpha^2}&{-\alpha}\,{} \\
     &0&\cdots&\cdots&0&{-\alpha}&{1}\,{}
 \end{pmatrix} \begin{pmatrix}
  \,&\kappa&0&\cdots&\cdots&\cdots&0 &\,{} \\
     &\alpha&\kappa&0&\cdots&\cdots&0 &\,{} \\
     &\alpha^2&\alpha&\ddots&\ddots&\ddots&\vdots&\,{} \\
\,{}&\vdots&\vdots & & & 0& \vdots\\
     &\alpha^{N-1}&\alpha^{N-2}&\ddots&\ddots&\kappa&0&\,{} \\
     &\alpha^N&\alpha^{N-1}&\cdots&\cdots&\alpha&\kappa&\,{}
 \end{pmatrix}\\
 &=&\begin{pmatrix}
  \,&{1-2\alpha^2+\kappa}&{- \alpha \kappa}&0&\cdots&\cdots&0 \,{} \\
     &{-\alpha (\kappa -1)}&{1+\alpha^2(\kappa -2)+\kappa }&{- \alpha \kappa}&0&\cdots&0\,{} \\
     &0&\ddots&\ddots&\ddots&\ddots&\vdots\,{} \\
     &\vdots&\ddots&\ddots&\ddots&\ddots&\vdots\,{} \\
     &\vdots&\ddots&0&{-\alpha (\kappa -1)}&{1+\alpha^2(\kappa -2)+\kappa }&{- \alpha \kappa}\,{} \\
     &0&\cdots&\cdots&0&{-\alpha (\kappa -1)}&{1-\alpha^2+\kappa }\,{}
 \end{pmatrix}.
\end{eqnarray*}

\begin{lemma}\label{lemma1}
	For $k \le N$, the $k^{\text{th}}$ leading principal minor $\delta_k$ of $B$ is given by
	\begin{align*}
		\delta_k &= c_+m_+^k+c_-m_-^k,
	\end{align*}
	where, for the real number
	\begin{align*}
		R := \sqrt{\alpha^4\left(\kappa-2\right)^2-2\alpha^2\left(2+\left(\kappa-1\right)\kappa\right)+\left(\kappa+1\right)^2},
	\end{align*}
	the real numbers $c_\pm$ and $m_\pm$ are given by
	\begin{align*}
		c_\pm =\frac{\pm\left(1-\alpha^2\left(\kappa+2\right)+\kappa\right)+R}{2R} \hspace{1cm} \text{ and } \hspace{1cm}
		m_\pm =\frac{1+\alpha^2\left(\kappa-2\right)+\kappa\pm R}{2}.
	\end{align*}
\end{lemma}
\begin{proof}
	We have
	\begin{align}
		\delta_1 &= 1-2\alpha^2+\kappa,\label{a}\\
		\delta_2 &= -2\alpha^4\left(\kappa-2\right)-2\alpha^2\left(\kappa+2\right)+\left(\kappa+1\right)^2\label{b}.
	\end{align}
	For $k \in \left\{3,\dots,N\right\}$, the $k^{\text{th}}$ principal minor, $\delta_k$, is given by the recursion
	\begin{align*}
		\delta_k &= \left(1+\alpha^2\left(\kappa-2\right)+\kappa\right)\delta_{k-1}-\alpha^2\kappa\left(\kappa-1\right)\,\delta_{k-2}.
	\end{align*}
	This recursion is a homogeneous linear difference equation of second order.
	Its characteristic equation is
	\begin{align}\label{c}
		m^2-\left(1+\alpha^2\left(\kappa-2\right)+\kappa\right)m+\alpha^2\kappa\left(\kappa-1\right) = 0.
	\end{align}
	This equation has the two roots, $m_+$ and $m_-$. 
	
	We claim first that $m_+$ and $m_-$ are real for $\alpha \in \left[0,1\right]$ and $\kappa \ge 1/2$.
	This claim is equivalent to the nonnegativity of the argument of the square root in our formula for $R$. 
	We claim that this is in turn equivalent to  $f\left(t\right):= t^2\left(\kappa-2\right)^2-2t\left(2+\left(\kappa-1\right)\kappa\right)+\left(\kappa+1\right)^2 \ge 0$ for $0\le t\le 1$, where we have replaced $\alpha^2$ by the parameter $t$.
	The preceding claim is clearly true for $\kappa=2$.
	Otherwise, $f$ is minimized at $t_0 := \left(2+\left(\kappa-1\right)\kappa\right)/\left(\kappa-2\right)^2$.
	We have $t_0 < 1$ if and only if $\kappa<2/3$.
	In this case, we have $f\left(t\right)\ge f\left(t_0\right) = 8\left(1-\kappa\right)\kappa/\left(\kappa-2\right)^2>0$ for all $t$.
	For $\kappa\ge 2/3$ we have $t_0 \ge 1$ and in turn $f'\left(t\right)\le 0$ for $0\le t\le1$.
	This gives us $f\left(t\right) \ge f\left(1\right) = 1$ for $0\le t\le 1$ and proves our claim that the roots $m_\pm$ are real.
		
	It now follows from the general theory of homogeneous linear difference equations of second order that every solution to \eqref{c} is of the form $c_1\left(m_+\right)^k+c_2\left(m_-\right)^k$, where $c_1$ and $c_2$ are real constants; see \cite[Theorem 3.7]{KelleyPeterson}.
	Requiring the two initial conditions \eqref{a} and \eqref{b} yields $c_1 = c_+$ and $c_2 = c_-$.
%
\end{proof}

\begin{lemma}\label{lemma2}
	Define $\phi_n$ recursively by
	\begin{align*}
		\phi_{N+2} = 1, \hspace{2cm} \phi_{N+1} = 1-\alpha^2+\kappa,
	\end{align*}
	and for $k = N, N-1, \dots, 2$ by
	\begin{align*}
		\phi_k = \left(1+\alpha^2\left(\kappa-2\right)+\kappa\right)\phi_{k+1}-\alpha^2\kappa\left(\kappa-1\right)\,\phi_{k+2}.
	\end{align*}
	Then, for $k \in \left\{2,\dots, N+2\right\},$
	\begin{align*}
		\phi_k = d_+  m_+ ^{N+2-k}+d_- m_-^{N+2-k},
	\end{align*}
	where $m_\pm$ are as in Lemma~\ref{lemma1} and
	\begin{align*}
		d_\pm := \frac{\pm\left(1+\left(1-\alpha^2\right)\kappa\right)+R}{2R}.
	\end{align*}
\end{lemma}
\begin{proof}
	Let
	\begin{align}\label{g}
		\psi_0 = 1, \hspace{1cm} \psi_1 = 1-\alpha^2+\kappa,
	\end{align}
	and for $l \in \left\{2,\dots,N\right\}$, let
	\begin{align}\label{f}
		\psi_l = \left(1+\alpha^2\left(\kappa-2\right)+\kappa\right)\psi_{l-1}-\alpha^2\kappa\left(\kappa-1\right) \psi_{l-2}.
	\end{align}
	Then $\psi_k = \phi_{N+2-k}$.
	As in the proof of Lemma~\ref{lemma1} we see that the general solution to \eqref{f} is of the form $d_1  m_+ ^l+d_2 m_- ^l$, where $m_\pm$ are as above.
	Choosing $d_1 = d_+$ and $d_2 = d_-$ ensures that the initial conditions \eqref{g} are satisfied and completes the proof	 
\end{proof}

\begin{lemma}\label{B inverse lemma}
	The matrix $B$ is non-singular and its inverse is given by
\begin{align}\label{B inverse eq}
(B^{-1})_{ij}
&=\begin{cases}(\alpha\kappa)^{j-i}\delta_{i-1}\phi_{j+1}\delta_{N+1}^{-1}&\text{if $i\le j$,}\\
(\alpha(\kappa-1))^{i-j}\delta_{j-1}\phi_{i+1}\delta_{N+1}^{-1}&\text{if $i\ge j$,}
\end{cases}
\end{align}
	where $\delta_0=1$.
\end{lemma}
\begin{proof}
	It was shown in \cite[Lemma 3.2]{SchiedZhangHotPotato} that both $\Gamma$ and ${\Gamma}+\tilde{\Gamma}+2\theta\Id$ are  invertible. Thus, 
		$B  = (\left.1-\alpha^2\right.)\,\Gamma^{-1} ({\Gamma}+\tilde{\Gamma}+2\theta\Id ) $ is also invertible. Note that this implies $\delta_{N+1} \neq 0$, so that the right-hand side of \eqref{B inverse eq} is well-defined.
In view of  Lemmas~\ref{lemma1} and~\ref{lemma2}, the explicit form of the inverse now follows from Usmani's formula for the inversion of a tridiagonal Jacobi matrix~\cite{Usmani2,Usmani1}.
\end{proof}

\begin{theorem}\label{omega and nu closed form thm}The components of $\bm\omega$ are given by
	\begin{align}\label{omi formula}
		 \omega_i &= \frac{\left(1-\alpha\right)\kappa+\alpha\left(\frac{\alpha\left(\kappa-1\right)}{\kappa}\right)^{N+1-i}}{\kappa\left(\kappa-\alpha\left(\kappa-1\right)\right)},
	\end{align}
	for all $i \in \left\{1,\dots,N+1\right\}$. In particular, $ \omega_{N+1} = 1/\kappa$.

The components of $\bm\nu$ are given as follows,
	\begin{align*}
		\nu_1 &= \frac{1-\alpha}{\delta_{N+1}} \Bigg(\phi_2 + \left(1-\alpha\right) \sum_{j=2}^N \left(\alpha\kappa\right)^{j-1} \phi_{j+1} + \left(\alpha\kappa\right)^N\Bigg),\\
		\nu_{N+1} &= \frac{1-\alpha}{\delta_{N+1}} \Bigg(\left(\alpha\left(\kappa-1\right)\right)^N + \left(1-\alpha\right) \sum_{j=2}^N \left(\alpha\left(\kappa-1\right)\right)^{N+1-j} \delta_{j-1} + \delta_N\Bigg),
	\end{align*}
	and for $i =2,\dots,N$,
	\begin{equation*}
	\begin{split}
		\nu_i &= \frac{1-\alpha}{\delta_{N+1}} \Bigg(\left(\alpha\left(\kappa-1\right)\right)^{i-1} \phi_{i+1}+ \left(1-\alpha\right)\sum_{j=2}^{i-1} \left(\alpha\left(\kappa-1\right)\right)^{i-j} \delta_{j-1} \phi_{i+1}\\ 
		&\hspace{3.5cm}+ \left(1-\alpha\right)\sum_{j=i}^N\left(\alpha\kappa\right)^{j-i} \delta_{i-1} \phi_{j+1} + \left(\alpha\kappa\right)^{N+1-i}\delta_{i-1}\Bigg).
	\end{split}
	\end{equation*}
\end{theorem}
\begin{proof} The representation \eqref{omi formula}
 was proved in \cite[Eq. (16)]{SchiedZhangHotPotato} (note that our vector $\bm\omega$ is denoted by $\bm u$ in~\cite{SchiedZhangHotPotato}, that our $\alpha$ corresponds to $a^{1/N}$ in~\cite{SchiedZhangHotPotato}, and that $\lambda=1$ here).  
 
 To prove the formulas for $\bm\nu$, note that we have $(\Gamma + \tilde{\Gamma} + 2\theta \Id)^{-1}\mathbf{1} = ( 1-\alpha^2)B^{-1}\Gamma^{-1} \mathbf{1}$.	The result thus follows from  Lemma~\ref{B inverse lemma}
 together with the fact that $(1-\alpha^2)\Gamma^{-1}\bm1=(1,1-\alpha,\dots,1-\alpha,1)^\top$, which in turn follows from \eqref{KMS inverse eq}.
\end{proof}

\section{Conclusion}We have studied the high-frequency limits of strategies and costs in a Nash equilibrium for two agents that are competing to minimize liquidation costs in a discrete-time market impact model with exponentially decaying price impact and quadratic transaction costs of size $\theta\ge0$. Our results have permitted us to give mathematically rigorous proofs of numerical observations made in~\cite{SchiedZhangHotPotato}.  In particular, we have shown that, for $\theta=0$, equilibrium strategies and costs will oscillate indefinitely between two accumulation points, which were computed explicitly. For $\theta>0$,  strategies,  costs, and total tax revenues were shown to converge toward limits that are independent of $\theta$. We have considered Nash equilibria in continuous time and  shown that for $\theta>0$ the limiting strategies converge  to the unique continuous-time Nash equilibrium for the critical value $\theta^*$  and that the  high-frequency limits of the discrete-time equilibrium costs converge to the expected costs in the continuous-time Nash equilibrium. For $\theta\neq\theta^*$, however, it was shown that continuous-time Nash equilibria do not exist unless both inventories are zero.  Moreover, we have provided a range of model parameters for which the limiting expected costs of both agents are decreasing functions of $\theta$ so that raising additional transaction costs can reduce the expected costs of all agents. 

\appendix
\section{Proofs of  Theorems~\ref{limit strategies kappa thm} and~\ref{costs asymptotics thm} and  Corollaries~\ref{cost comparison cor} and~\ref{tax corollary}}

Quantities such as $\alpha$, $\nu$, or $\omega$ introduced in Section~\ref{closed form section} depend on the  parameter $N$ of the trading frequency. For the proofs of our asymptotic results, we need to send $N$ to infinity, but for the sake of reducing formula length, we will not always make the $N$-dependence of  quantities explicit. For instance, we will write
$$\lim_{N\uparrow\infty}\alpha=\lim_{N\uparrow\infty}e^{-\rho T/N}=0.
$$

\subsection{Proof of Theorem~\ref{limit strategies kappa thm} }

We first prove parts (c) and (d) of Theorem~\ref{limit strategies kappa thm}. The proofs of these parts are relatively easy. 

\begin{proof}[Proof of Theorem~\ref{limit strategies kappa thm} {\rm (c)}] Let $\theta>0$, which is equivalent to $\kappa>1/2$. 
We first sum over  \eqref{omi formula} to obtain that, for all $\kappa\ge1/2$ and $n=1,\dots, N+1$,
\begin{align}\label{partial sum omk kappa>1/2 eq}
\sum_{k=1}^n\omega_k=\frac1\kappa \bigg[n\Big(1-\frac\alpha{(1-\alpha)\kappa+\alpha}\Big)+\frac\alpha{(1-\alpha)\kappa+\alpha}\Big(\frac{\alpha(\kappa-1)}{\kappa}\Big)^{N+1-n}\frac{\big(\frac{\alpha(\kappa-1)}{\kappa}\big)^n-1}{\frac{\alpha(\kappa-1)}{\kappa}-1}\bigg];
\end{align}
Here we explicitly include the case $\kappa=1/2$ for later use. 
By taking $n=N+1$, formula \eqref{partial sum omk kappa>1/2 eq} gives
\begin{align*}
\bm\omega^\top\bm1=\frac1\kappa \bigg[(N+1)\Big(1-\frac\alpha{(1-\alpha)\kappa+\alpha}\Big)+\frac\alpha{(1-\alpha)\kappa+\alpha}\frac{\big(\frac{\alpha(\kappa-1)}{\kappa}\big)^{N+1}-1}{\frac{\alpha(\kappa-1)}{\kappa}-1}\bigg].
\end{align*}
Recalling that $\alpha = e^{-\rho T/N}$, we have that 
\begin{equation}\label{Thm 3.1 c proof aux eq1}
\lim_{N\uparrow\infty}(N+1)\Big(1-\frac\alpha{(1-\alpha)\kappa+\alpha}\Big)=\kappa\rho T.
\end{equation}
Since for    $\kappa>1/2$ we have $|\kappa-1|/\kappa<1$, this gives 
\begin{equation}\label{omega.1 asympt for kappa>1/2 eq}
\lim_{N\uparrow\infty}\bm\omega^\top\bm1=\rho T+1.
\end{equation}
Now, with $n_t:=\lceil Nt/T\rceil$,
$$W^{(N)}_t=1-\frac1{\bm\omega^\top\bm1}\sum_{k=1}^{n_t}\omega_k.
$$
Since for $t<T$ we have $ (\frac{\alpha(\kappa-1)}{\kappa})^{N+1-n_t}\to0$ as $N\uparrow\infty$, formula \eqref{partial sum omk kappa>1/2 eq} gives in this case that 
$$\sum_{k=1}^{n_t}\omega_k\longrightarrow \rho t\qquad\text{as $N\uparrow\infty$.}
$$
Putting everything together now yields the assertion. 
\end{proof}

\begin{proof}[Proof of Theorem~\ref{limit strategies kappa thm} {\rm (d)}] For $\kappa=1/2$, the formula \eqref{partial sum omk kappa>1/2 eq} simplifies as follows,
\begin{align}\label{partial sum omk kappa=1/2 eq}
\sum_{k=1}^n\omega_k=2 \bigg[n\Big(1-\frac{2\alpha}{1+\alpha}\Big)+(-1)^{N+1}\frac{2\alpha^{N+2-n}\big((-1)^n-\alpha^n\big)}{(1+\alpha)^2}\bigg].
\end{align}
Thus, for $\kappa=1/2$, we get with \eqref{Thm 3.1 c proof aux eq1}
 that
\begin{align}\label{sum omi limit formula kappa=1/2}
\lim_{\substack{N\uparrow\infty\\
	N\,\text{\rm even}}}\bm\omega^\top\bm1=\lim_{\substack{N\uparrow\infty\\
	N\,\text{\rm even}}}	\sum_{i=1}^{ N+1 } \omega_i = e^{-\rho T} +\rho T+1 \hspace{1cm}\text{ and }\hspace{1cm} \lim_{\substack{N\uparrow\infty\\
	N\,\text{\rm odd}}}	\bm\omega^\top\bm1 = -e^{-\rho T}+\rho T+1.
	\end{align}
If taking $n=n_t=\lceil Nt/T\rceil$, one easily shows that as $N\uparrow\infty$,
$$n\Big(1-\frac{2\alpha}{1+\alpha}\Big)\lra\frac{\rho t}2\qquad \text{and}\qquad \frac{2\alpha^{N+2-n}\big(\pm1-\alpha^n\big)}{(1+\alpha)^2}\lra e^{-\rho(T-t)}(\pm1-e^{-\rho t}).
$$
Plugging this into \eqref{partial sum omk kappa=1/2 eq} and using the definition of $W^{(N)}$ yields the result of Theorem~\ref{limit strategies kappa thm} (d)  after a short computation. 
\end{proof}

Now we prepare for the proofs of parts (a) and (b) of Theorem~\ref{limit strategies kappa thm}.  We first consider the  case ${\kappa=1}$; the corresponding proof of Theorem~\ref{limit strategies kappa thm} (a) will be given after the following lemma.

\begin{lemma}
	Let $\kappa =1$.
	Then
	\begin{align}\label{aa}
		\sum_{i=1}^n \nu_i &= \frac{1}{2+\alpha}\left(\left(1-\alpha\right)n+\alpha+\frac{\alpha\left(\alpha^2-2\right)}{2\left(2+\alpha\right)}\left(\frac{\alpha}{2-\alpha^2}\right)^{N+1}+\frac{\alpha\left(1+\alpha\right)}{2+\alpha}\left(\frac{\alpha}{2-\alpha^2}\right)^{N+1-n}\right)
	\end{align}
	for $n \in \left\{1,\dots,N+1\right\}$.
\end{lemma}
\begin{proof}
	Plugging in $\kappa=1$ yields that $\delta_k = 2\left(1-\alpha^2\right)\left(2-\alpha^2\right)^{k-1}$ for $k \in \left\{1,\dots,N+1\right\}$, as well as $\phi_k = \left(2-\alpha^2\right)^{N+2-k}$ for $k \in \left\{2,\dots,N+1\right\}$.
	Therefore,
	\begin{align*}
		\nu_1 = \frac{1}{2+\alpha}\left(1+\frac{2-\alpha^2}{2}\left(\frac{\alpha}{2-\alpha^2}\right)^{N+1}\right)\qquad\text{and}\qquad		\nu_i = \frac{1}{2+\alpha}\left(1-\alpha+\left(1-\alpha^2\right)\left(\frac{\alpha}{2-\alpha^2}\right)^{N+2-i}\right)
	\end{align*}
	for $i\in\left\{2,\dots,N+1\right\}$.
	Summing over $i = 1,\dots,n$  yields the result.	 
\end{proof}
 
 \begin{proof}[Proof of Theorem~\ref{limit strategies kappa thm} {\rm(a)} for $\kappa=1$]
	Recall that $\alpha = e^{-\rho T/N}$.
	Therefore, $\left(1-\alpha\right)n_t \to \rho t$ and $\left(2-\alpha^2\right)^{n_t} \to e^{2\rho t}$ for all $t \in \left(0,T\right]$.
	Taking limits in \eqref{aa} yields
	\begin{align}\label{(partial) sum nui kappa=1 eq}
		\sum_{i=1}^{n_t} \nu_i \longrightarrow \frac{e^{-3\rho T}\left(4e^{3\rho t}-1\right)+6\left(\rho t+1\right)}{18}.
	\end{align}
	Plugging this into the definition of $V^{(N)}_t$ yields the result. 	 
\end{proof}
 
Now we prepare for the proof of parts  (a) and (b)  of Theorem~\ref{limit strategies kappa thm} in case $\kappa\neq1$. For the remainder of this paper, we define the shorthand notation for $x\in\mathbb{R}$ and $m\in\mathbb{N}$,
\begin{align*}
	\left[x\right]^m := \frac{1-\alpha}{\delta_{N+1}} x^m.
\end{align*}
This notation will be convenient when computing limits of expressions like $[x]^N$.

\begin{lemma}
	Let $\kappa \ge 1/2$ and $\kappa \neq 1$.
	Define $C_1 := \alpha\left(1+\alpha\right)/\left(\kappa+1-\alpha\left(\kappa-2\right)\right)$.
	Then
	\begin{align}\label{ae}
		\sum_{i=1}^n \nu_i
		&= \sums \frac{d_\sigma\left(m_\sigma-\alpha^2\kappa \right)}{m_\sigma- \alpha\kappa} \left[m_\sigma\right]^N \\
			&\hspace{.5cm}{}+ \left(1-\alpha\right)\left(n-1\right)\sums c_\sigma d_\sigma \left(\frac{\alpha\left(\kappa-1\right)}{m_\sigma-\alpha\left(\kappa-1\right)}+\frac{m_\sigma}{m_\sigma-\alpha\kappa}\right) \left[m_\sigma\right]^N\nonumber\\
			&\hspace{.5cm}{}+C_1\left(1+\sums \frac{c_\sigma m_\sigma\left(\left(\frac{m_\sigma}{\alpha\kappa}\right)^{n-1}-1\right)}{m_\sigma-\alpha\kappa}\right)\alpha^N \left[\kappa\right]^N\nonumber\\
			&\hspace{.5cm}{}+2C_1 \sums \frac{d_\sigma m_\sigma\left(\frac{\alpha\left(\kappa-1\right)}{m_\sigma}-\left(\frac{\alpha\left(\kappa-1\right)}{m_\sigma}\right)^n\right)}{m_\sigma-\alpha\left(\kappa-1\right)}  \left[m_\sigma\right]^N,\nonumber
	\end{align}
	for $n \in \left\{1,\dots,N\right\}$, and
	\begin{equation}\label{af}
	\begin{aligned}
		\nu_{N+1}&= \sums \frac{c_\sigma\left(m_\sigma-\alpha^2\left(\kappa-1\right)\right)}{m_\sigma-\alpha\left(\kappa-1\right)}  \left[m_\sigma\right]^N+2C_1\alpha^N  \left[\kappa-1\right]^N.
	\end{aligned}
	\end{equation}
\end{lemma}
\begin{proof}
	For $i \in \left\{3,\dots,N\right\}$,
	\begin{align}\label{ab}
		&\hphantom{=\,}\sum_{j=2}^{i-1} \left(\alpha\left(\kappa-1\right)\right)^{i-j} \delta_{j-1} \phi_{i+1}\\ 
		&=\alpha\left(\kappa-1\right)
		\Bigg(
			\sums\frac{c_\sigma d_\sigma}{m_\sigma-\alpha\left(\kappa-1\right)}\left(m_\sigma\right)^N\nonumber\\
			&\hspace{2.5cm}{}+ \frac{c_+ d_- \left(m_-\right)^{N+1}}{m_+\left(m_+-\alpha\left(\kappa-1\right)\right)} \left(\frac{m_+}{m_-}\right)^i + \frac{c_- d_+ \left(m_+\right)^{N+1}}{m_-\left(m_--\alpha\left(\kappa-1\right)\right)}\left(\frac{m_-}{m_+}\right)^i\nonumber\\
			&\hspace{2.5cm}{}- \sums \frac{c_\sigma m_\sigma}{m_\sigma-\alpha\left(\kappa-1\right)}\sumt \frac{d_\tau\left(m_\tau\right)^{N+1}}{\left(\alpha\left(\kappa-1\right)\right)^2}\left(\frac{\alpha\left(\kappa-1\right)}{m_\tau}\right)^i\nonumber
		\Bigg)
	\end{align}
	and
	\begin{align}\label{ac}
		&\hphantom{=\,}\sum_{j=i}^{N} \left(\alpha\kappa\right)^{j-i} \delta_{i-1} \phi_{j+1}\\ 
		&=\sums\frac{c_\sigma d_\sigma}{m_\sigma-\alpha\kappa}\left(m_\sigma\right)^{N+1}
			+ \frac{c_+ d_-\left(m_-\right)^{N+2}}{m_+\left(m_--\alpha\kappa\right)}\left(\frac{m_+}{m_-}\right)^i +\,\frac{c_- d_+\left(m_+\right)^{N+2}}{m_-\left(m_+-\alpha\kappa\right)}\left(\frac{m_-}{m_+}\right)^i\nonumber\\
		&\hspace{.5cm}{}-\sums\frac{d_\sigma m_\sigma}{m_\sigma-\alpha\kappa}\sumt\frac{c_\tau\left(\alpha\kappa\right)^{N+1}}{m_\tau}\left(\frac{m_\tau}{\alpha\kappa}\right)^i.\nonumber
	\end{align}
	Since
	\begin{equation*}
	\begin{aligned}
		\alpha\left(\kappa-1\right)\left(m_--\alpha\kappa\right)+m_-\left(m_+-\alpha\left(\kappa-1\right)\right) 
		&=\alpha\left(\kappa-1\right)\left(m_+-\alpha\kappa\right)+m_+\left(m_--\alpha\left(\kappa-1\right)\right)\\
		&=m_+m_--\alpha^2\kappa\left(\kappa-1\right)=0,
	\end{aligned}
	\end{equation*}
	the second and third summands in \eqref{ab} and \eqref{ac} cancel each other out.
	Simplifying further, we arrive at
	\begin{align*}
		\nu_i &= \left(1-\alpha\right)\sums c_\sigma d_\sigma\left(\frac{\alpha\left(\kappa-1\right)}{m_\sigma-\alpha\left(\kappa-1\right)}+\frac{m_\sigma}{m_\sigma-\alpha\kappa}\right)  \left[m_\sigma\right]^N\\
		&\hspace{.5cm}{}+2C_1\sums \frac{d_\sigma m_\sigma   \left[m_\sigma\right]^N}{\alpha\left(\kappa-1\right)}\left(\frac{\alpha\left(\kappa-1\right)}{m_\sigma}\right)^i+C_1\sums \frac{c_\sigma \alpha^{N+1}\kappa   \left[\kappa\right]^N}{m_\sigma}\left(\frac{m_\sigma}{\alpha\kappa}\right)^i,\nonumber
	\end{align*}
	for $i \in \left\{2,\dots,N\right\}$.
	Similar calculations yield
	\begin{align*}
		\nu_1 &= \sums\frac{d_\sigma\left(m_\sigma-\alpha^2\kappa\right)}{m_\sigma-\alpha\kappa}  \left[m_\sigma\right]^N+C_1 \alpha^N   \left[\kappa\right]^N,\\
		\nu_{N+1} &= \sums\frac{c_\sigma\left(m_\sigma-\alpha^2\left(\kappa-1\right)\right)}{m_\sigma-\alpha\left(\kappa-1\right)}  \left[m_\sigma\right]^N+2C_1\alpha^N  \left[\kappa-1\right]^N.
	\end{align*}
	Noting that
	\begin{equation*}
	\begin{aligned}
		\sum_{i=2}^n \sums \frac{d_\sigma m_\sigma   \left[m_\sigma\right]^N}{\alpha\left(\kappa-1\right)}\left(\frac{\alpha\left(\kappa-1\right)}{m_\sigma}\right)^i= \sums \frac{d_\sigma m_\sigma\left(\left(\frac{m_\sigma}{\alpha\left(\kappa-1\right)}\right)^{N-1}-\left(\frac{m_\sigma}{\alpha\left(\kappa-1\right)}\right)^{N-n}\right)}{m_\sigma-\alpha\left(\kappa-1\right)}\alpha^N  \left[\kappa\right]^N
	\end{aligned}
	\end{equation*}
	and
	\begin{equation*}
	\begin{aligned}
		\sum_{i=2}^n \sums \frac{c_\sigma\alpha^{N+1}\kappa   \left[\kappa\right]^N}{m_\sigma}\left(\frac{m_\sigma}{\alpha\kappa}\right)^i
		= \sums \frac{c_\sigma m_\sigma\left(\left(\frac{m_\sigma}{\alpha\kappa}\right)^{n-1}-1\right)}{m_\sigma-\alpha\kappa}\alpha^N   \left[\kappa\right]^N
	\end{aligned}
	\end{equation*}
	 for all $n \in \left\{2,\dots,N\right\}$ completes the proof.	 
\end{proof}

The following lemma summarizes the limit behaviour of all objects that we will encounter while obtaining the limiting strategy and, later, the limiting costs. Recall that  $n_t := \lceil N t/T\rceil$.
For a sequence of real numbers $\left(a_N\right)_{N \in \mathbb{N}}$ and a real number $a$, we use the shorthand notation $\left(a_N\right)^{n_t} \to \pm a$ to state that $\left(a_N\right)^{n_t} = \left(-1\right)^{n_t} \left| a_N\right|^{n_t}$ and $\lim_{N \to \infty} \left|a_N\right|^{n_t} = a$.

\begin{lemma}\label{lemma4}

	For $\kappa\ge1/2$ and  $\kappa \neq 1$, we have the following limits for $N\uparrow\infty$.
	\begin{enumerate}
	\itemsep1em
		\item $\alpha \to 1$ and $\alpha^{n_t} \to e^{-\rho t}$;
		\item $R \to 1,\quad  c_+ \to 0,\quad  c_- \to 1,\quad  d_+ \to 1,\quad  d_- \to 0,\quad  m_+ \to \kappa$, \  and \  $m_- \to \kappa-1$;
		\item $\frac{c_+}{m_+-\kappa} \to 2, \quad \frac{c_+}{m_+-\alpha\kappa} \to \frac{4}{3}, \quad  \frac{c_+}{m_+-\alpha^2\kappa} \to 1$, \ and \ $\frac{c_+}{1-\alpha^2} \to 2\kappa$;
		\item $\frac{d_-}{m_--\left(\kappa-1\right)} \to -\frac{1}{2},\quad  \frac{d_-}{m_--\alpha\left(\kappa-1\right)} \to -\frac{2}{3}, \quad \frac{d_-}{m_--\alpha^2\left(\kappa-1\right)} \to -1$, \ and \ $\frac{d_-}{1-\alpha^2} \to \kappa-1$;
		\item $\left(1-\alpha\right)n_t \to \rho t$. 
	\end{enumerate}
\noindent	If additionally $\kappa > 1/2$, then also the following limits are true.
	\begin{enumerate}
		\itemsep1em
	\setcounter{enumi}{5}
		\item $\left(\frac{\kappa-1}{\kappa}\right)^{n_t}\to 0,\quad  \left(\frac{m_+}{\kappa}\right)^{n_t} \to e^{2\rho t}, \quad \left(\frac{\kappa-1}{m_+}\right)^{n_t} \to 0, \quad  \left(\frac{m_-}{\kappa}\right)^{n_t} \to 0$, \ and \ $\left(\frac{\kappa-1}{m_-}\right)^{n_t} \to e^{4\rho t}$;
		\item $  \left[m_+\right]^N \to \frac{1}{4\kappa}, \quad  \left[m_-\right]^N \to 0, \quad   \left[\kappa\right]^N \to \frac{e^{-2\rho T}}{4\kappa}$, \ and \ $  \left[\kappa-1\right]^N\to 0$;
		\item $\frac{\left(\left(\kappa-1\right)/\kappa\right)^N}{1-\alpha^2} \to 0, \quad  \frac{  \left[m_-\right]^N}{1-\alpha^2} \to 0$, \ and \ $\frac{  \left[\kappa-1\right]^N}{1-\alpha^2} \to 0$.
			\end{enumerate}
\noindent	If, on the other hand,  $\kappa = 1/2$ then the preceding limits no longer hold. Instead, we have the following.	\begin{enumerate}[{\normalfont (a')}]
	\itemsep1em
	\setcounter{enumi}{5}
		\item $\left(\frac{\kappa-1}{\kappa}\right)^{n_t} \to \pm1, \quad \left(\frac{m_+}{\kappa}\right)^{n_t} \to e^{2\rho t}, \quad \left(\frac{\kappa-1}{m_+}\right)^{n_t} \to \pm e^{-2\rho t}, \quad \left(\frac{m_-}{\kappa}\right)^{n_t} \to \pm e^{-4\rho t}$,  and $\left(\frac{\kappa-1}{m_-}\right)^{n_t} \to e^{4\rho t}$;
	\item $ \left[m_+\right]^{2N} \to \frac{1}{e^{-6\rho T}+2},  \quad \left[m_-\right]^{2N}\to \frac{1}{2e^{6\rho T}+1},  \quad \left[\kappa\right]^{2N}\to \frac{e^{4\rho T}}{2e^{6\rho T}+1}, \quad  \left[\kappa-1\right]^{2N} \to \frac{e^{4\rho T}}{2e^{6\rho T}+1}$,\\[3pt]
		$ \left[m_+\right]^{2N+1} \to \frac{1}{-e^{-6\rho T}+2},  \quad \left[m_-\right]^{2N+1} \to \frac{1}{-2e^{6\rho T}+1},  \quad \left[\kappa\right]^{2N+1} \to \frac{e^{4\rho T}}{2e^{6\rho T}-1}$, and $ \left[\kappa-1\right]^{2N+1} \to \frac{e^{4\rho T}}{-2e^{6\rho T}+1}$;
	\item $\frac{m_++\kappa-1}{m_++\alpha\left(\kappa-1\right)} \to \frac{2}{3}, \quad  \frac{m_-+\alpha^2\kappa}{m_-+\alpha\kappa} \to \frac{2}{3}$, \ and \ $\frac{\kappa+\alpha\left(\kappa-1\right)}{1-\alpha^2} \to \frac{1}{4}$.\\[5pt]
	\end{enumerate}
\end{lemma}
\begin{proof}
	(a) and (b) are obvious, (c)--(e) follow by applying L'H\^opital's rule.
	The first statement in (f) follows from the fact that $\kappa > 1/2$.
	To prove the second, write 
	\begin{align*}
		\left(m_+/\kappa\right)^N = \exp\left(N\log\left(m_+/\kappa\right)\right)
	\end{align*}	
	and apply L'H\^opital's rule.
	The third statement follows directly, since
	\begin{align*}
		\left(\left(\kappa-1\right)/\kappa\right)^{n_t}=\left(m_+/\kappa\right)^{n_t}\left(\left(\kappa-1\right)/m_+\right)^{n_t}\to 0.
	\end{align*} 
	The fourth and fifth statements can be proved in a similar fashion.
	With regard to (g) and (h), recall that
	\begin{equation*} 
	\begin{aligned}
		&\hphantom{{}=}\frac{1-\alpha}{\delta_{N+1}}\\
		&=\left(\frac{c_+\left(1-\alpha^2+\kappa-\frac{\alpha^2\kappa\left(\kappa-1\right)}{m_+}\right)}{1-\alpha}\left(m_+\right)^N+ \frac{c_-\left(\left(1-\alpha^2+\kappa\right)m_--\alpha^2\kappa\left(\kappa-1\right)\right)}{m_-\left(1-\alpha\right)}\left(m_-\right)^N\right)^{-1}\hspace{-.4cm}.
	\end{aligned}
	\end{equation*}
	Applying L'H\^opital's rule: 
	\begin{align*}
		\frac{c_+\left(1-\alpha^2+\kappa-\frac{\alpha^2\kappa\left(\kappa-1\right)}{m_+}\right)}{1-\alpha} &\to 4\kappa,\hspace{1cm}\text{ and }\\
		\frac{c_-\left(\left(1-\alpha^2+\kappa\right)m_--\alpha^2\kappa\left(\kappa-1\right)\right)}{m_-\left(1-\alpha\right)} &\to -2\left(\kappa-1\right)^2.
	\end{align*}
	It follows from (vi) that $\left(m_-/m_+\right)^N \to 0$ and, using L'H\^opital's rule again, 
	\begin{align*}
		\frac{\left(m_-/m_+\right)^N}{1-\alpha^2}\to 0, &&\frac{\left(m_-/\left(\kappa-1\right)\right)^N}{1-\alpha^2} \to 0, &&\text{ and }  &&\frac{\left(\left(\kappa-1\right)/\kappa\right)^N}{1-\alpha^2} \to 0.&
	\end{align*}
	Plugging in and taking limits yields the results.
	
	If $\kappa = 1/2$, observe that
	\begin{align*}
		m_- = \frac{3}{4} \left(1-\alpha^2-\sqrt{\left(1-\alpha^2\right)^2+4/9\alpha^2}\right) < 0 < m_+.
	\end{align*}
	With this in mind, statements (f') and (g') can be proved in the same way as statements (f) and (g).
	(h') is another application of L'H\^opital's rule.
\end{proof}

\begin{proof}[Proof of Theorem~\ref{limit strategies kappa thm} {\rm (a)} for $\kappa\neq1$]
	Let $\kappa>1/2$ and $\kappa \neq 1$.
	The limits of \eqref{ae} and \eqref{af} are easily calculated using Lemma~\ref{lemma4}.
	In total, we see that, as $N\uparrow\infty$,
	\begin{align}\label{sum nuit kappa>1/2, kappa neq 1 eq}
		\sum_{i=1}^{n_t} \nu_i &\lra \frac{e^{-3\rho T}\left(6e^{3\rho T}\left(\rho t+1\right)+4e^{3\rho t}-1\right)}{18},
	\end{align}
	for $t \in \left(0,T\right)$, and
	\begin{align}\label{sum nui kappa>1/2, kappa neq 1 eq}
		\sum_{i=1}^{N+1} \nu_i &\lra \frac{e^{-3\rho T}\left(2e^{3\rho T}\left(3\rho T+5\right)-1\right)}{18}.
	\end{align}
	Note that \eqref{sum nui kappa>1/2, kappa neq 1 eq} coincides with the value of the right-hand side of \eqref{sum nuit kappa>1/2, kappa neq 1 eq} for $t=T$ and that  \eqref{sum nuit kappa>1/2, kappa neq 1 eq} coincides with the limit from \eqref{(partial) sum nui kappa=1 eq}, which was obtained for $\kappa=1$. Plugging this result into the definition of $V^{(N)}$ yields the result.	 
\end{proof}

\begin{proof}[Proof of Theorem~\ref{limit strategies kappa thm} {\rm (b)}]	For $\kappa=1/2$, the limits of \eqref{ae} and \eqref{af} can be obtained using Lemma~\ref{lemma4}.
	We find that $\sum_{i=1}^{n_t}\nu_i$ has two cluster points for every $t \in \left(0,T\right]$. One is approached if $n_t$ is even, the other one if $n_t$ is odd.
	The same is true for $\sum_{i=1}^{N+1}\nu_i$, depending on whether $N$ is even or odd.
	For future reference, we now state the limits of $\mathbf{1}^\top {\bm\nu}$ as $N\uparrow\infty$;
	\begin{equation}\label{limits sum nui kappa=1/2 eq}
	\begin{split}
	\lim_{\substack{N\uparrow\infty\\
	N\,\text{\rm even}}}	\mathbf{1}^\top {\bm\nu} &=\frac{2e^{6\rho T}\left(3\rho T+5\right)+e^{3\rho T}+3\rho T+7}{18e^{6\rho T}+9}, \\
	\lim_{\substack{N\uparrow\infty\\
	N\,\text{\rm odd}}}	\mathbf{1}^\top {\bm\nu} &= \frac{2e^{6\rho T}\left(3\rho T+5\right)-3e^{3\rho T}-3\rho T-7}{18e^{6\rho T}-9}.
	\end{split}	 
	\end{equation}
\end{proof}

\subsection{Proof of Theorem~\ref{costs asymptotics thm} (a)}

We  start preparing the proof with the following simple lemma, which holds for all $\kappa\ge1/2$. 

\begin{lemma}
	We have	\begin{align}\label{am}
		\mathbb{E}\left[\mathscr{C}_\mathbb{T}\left(\bm\xi\mid\bm\eta\right)\right] &= \frac{1}{8}
		\Bigg(
			\frac{\left(x+y\right)^2}{\mathbf{1}^\top {\bm\nu}}+\frac{\left(x^2-y^2\right)\left(\mathbf{1}^\top {\bm\nu}+\mathbf{1}^\top {\bm\omega}\right)}{\left(\mathbf{1}^\top {\bm\nu}\right)\left(\mathbf{1}^\top {\bm\omega}\right)}+\frac{\left(x-y\right)^2}{\mathbf{1}^\top {\bm\omega}}
			\\
			&\hspace{1cm}{}+\left(\frac{x+y}{\mathbf{1}^\top {\bm\nu}}\right)^2{\bm\nu}^\top \tilde{\Gamma}\,{\bm\nu} + \frac{x^2-y^2}{\left(\mathbf{1}^\top {\bm\nu}\right)\left(\mathbf{1}^\top {\bm\omega}\right)}{\bm\omega}^\top \left(\tilde{\Gamma}-\tilde{\Gamma}^\top \right){\bm\nu}-\left(\frac{x-y}{\mathbf{1}^\top {\bm\omega}}\right)^2 {\bm\omega}^\top \tilde{\Gamma}\,{\bm\omega}\nonumber
		\Bigg).
	\end{align}
\end{lemma}
\begin{proof}
	We have
	\begin{align*}
		\mathbb{E}\left[\mathscr{C}_\mathbb{T}\left(\bm\xi\mid\bm\eta\right)\right] &= \frac{1}{2} \bm\xi^\top \left({\Gamma}+2\theta\Id\right)\bm\xi+\bm\xi^\top \tilde{\Gamma}\,\bm\eta\\
		&=\frac{1}{2}
		\Bigg(
			\left(\frac{x+y}{2\left(\mathbf{1}^\top {\bm\nu}\right)}\right)^2 \left({\bm\nu}^\top \left({\Gamma}+\tilde{\Gamma}+2\theta\Id\right){\bm\nu}\right)\\
			&\hspace{1cm}{}+ \frac{\left(x+y\right)\left(x-y\right)}{4\left(\mathbf{1}^\top {\bm\nu}\right)\left(\mathbf{1}^\top {\bm\omega}\right)}\left({\bm\nu}^\top \left({\Gamma}-\tilde{\Gamma}+2\theta\Id\right){\bm\omega}+{\bm\omega}^\top \left({\Gamma}+\tilde{\Gamma}+2\theta\Id\right){\bm\nu}\right)\nonumber\\
			&\hspace{1cm}{}+ \left(\frac{x-y}{2\left(\mathbf{1}^\top {\bm\omega}\right)}\right)^2 \left({\bm\omega}^\top \left({\Gamma}-\tilde{\Gamma}+2\theta\Id\right){\bm\omega}\right)\nonumber
		\Bigg)
		+\frac{1}{2}\bm\xi^\top \tilde{\Gamma}\,\bm\eta.
	\end{align*}	
	By definition, $(\left.{\Gamma}+\tilde{\Gamma}+2\theta \Id\right.){\bm\nu} = (\left.{\Gamma}-\tilde{\Gamma}+2\theta\Id\right.){\bm\omega} = \mathbf{1}$.
	Since also ${\bm\nu}^\top \mathbf{1} = \mathbf{1}^\top {\bm\nu}$, ${\bm\omega}^\top \mathbf{1} = \mathbf{1}^\top {\bm\omega}$, and ${\bm\nu}^\top \tilde{\Gamma}\,{\bm\omega} = {\bm\omega}^\top \tilde{\Gamma}^\top {\bm\nu}$, representation \eqref{am} follows.	 
\end{proof}

\begin{lemma}\label{lemma7}
	For $\kappa > 1/2$, as $N\uparrow\infty$,
		\begin{align*}
		{\bm\nu}^\top \tilde{\Gamma}\,{\bm\nu} &\lra (\left.-e^{-6\rho T}-8e^{-3\rho T}+24\rho T+36\right.)/216,\\
		{\bm\omega}^\top (\left.\tilde{\Gamma}-\tilde{\Gamma}^\top \right.){\bm\nu} &\lra (\left.-e^{-3\rho T}+4\right.)/6,\text{ and}\\
		{\bm\omega}^\top \tilde{\Gamma}\,{\bm\omega} &\lra \left(2\rho T+1\right)/2.
	\end{align*}
\end{lemma}
\begin{proof}
	First let $\kappa =1$.
	Then
	\begin{align*}
		{\bm\nu}^\top \tilde{\Gamma}\,{\bm\nu} &= 
		\frac{\left(\nu_1\right)^2}{2}+\frac{1}{2}\sum_{i=2}^{N+1}\left(\nu_i\right)^2+ \nu_1 \sum_{i=2}^{N+1}\nu_i\alpha^{i-1} + \sum_{i=3}^{N+1} \sum_{j=2}^{i-1}\nu_i\nu_j\alpha^{i-j}\\
		&= \frac{1}{\left(2+\alpha\right)^2} \Bigg(\frac{\left(1-\alpha^2\right)N}{2}+\frac{-\alpha^4+2\alpha^2+4\alpha+4}{2\left(4-\alpha^2\right)}\nonumber\\
		&\hspace{2.3cm}-\frac{\alpha^2\left(\alpha+1\right)}{2\left(\alpha+2\right)}\left(\frac{\alpha}{2-\alpha^2}\right)^N-\frac{\alpha^4}{8\left(4-\alpha^2\right)}\left(\frac{\alpha}{2-\alpha^2}\right)^{2N}\Bigg)\nonumber\\
		&\to \frac{-e^{-6\rho T}-8e^{-3\rho T}+24\rho T+36}{216},\nonumber
	\end{align*}
	as well as
	\begin{align*}
		{\bm\omega}^\top \left(\tilde{\Gamma}-\tilde{\Gamma}^\top \right){\bm\nu} &= 
		\nu_1 \sum_{i=2}^{N+1}  \omega_i \alpha^{i-1} +  \omega_{N+1} \sum_{i=2}^N \nu_i \alpha^{N+1-i} + \sum_{i=1}^N  \omega_i\left(\sum_{j=2}^{i-1} \nu_j \alpha^{i-j} - \sum_{j=i+1}^{N+1} \nu_j \alpha^{j-i}\right)\\
		&= - \frac{\alpha^2}{(2-\alpha^2)(2+\alpha)} \left( \frac{2-\alpha^2}{2} \left( \frac{\alpha}{2-\alpha^2}\right)^N + \alpha^2 - 3 \right) \nonumber\\
	&\to \frac{-3e^{-3\rho T}+4}{6},\nonumber
	\end{align*}
	and
	\begin{align*}
		{\bm\omega}^\top \tilde{\Gamma}\, {\bm\omega} &= \frac{1}{2} \sum_{i=1}^N \left(1-\alpha\right)^2 + \frac{1}{2} + \sum_{i=2}^N \sum_{j=1}^{i-1} \left(1-\alpha\right)^2 \alpha^{i-j} + \sum_{j=1}^N \left(1-\alpha\right) \alpha^{N+1-j}\\
	&=\frac{N\left(1-\alpha^2\right)+1}{2}\\
	&\to \frac{2\rho T+1}{2}.
	\end{align*}
	
	Now let $\kappa \ge 1/2$ and $\kappa \neq 1$. Note that we explicitly include the case $\kappa=1/2$, because partial results we will obtain in the following computations will also be needed to treat this case.

	The following calculations are standard but tedious.
	We need to find the limits of  ${\bm\nu}^\top \tilde{\Gamma}{\bm\nu}$,  of ${\bm\omega}^\top (\left.\tilde{\Gamma}-\tilde{\Gamma}^\top \right.){\bm\nu}$, and of ${\bm\omega}^\top \tilde{\Gamma}\,{\bm\omega}$.	
	We first compute $\tilde{\Gamma}{\bm\nu}$.
	Define $C_1 := \alpha\left(1+\alpha\right)/\left(1-\alpha\left(\kappa-2\right)+\kappa\right)$ as above and
\begin{align*}
	C_2 &:= \sums c_\sigma d_\sigma \left(\frac{\alpha\left(\kappa-1\right)}{m_\sigma-\alpha\left(\kappa-1\right)}+\frac{m_\sigma}{m_\sigma-\alpha\kappa}\right)  \left[m_\sigma\right]^N,\\
	C_3 &:= -C_2 + \sums d_\sigma \left(\frac{m_\sigma-\alpha^2\kappa}{m_\sigma-\alpha\kappa}+\frac{2C_1\left(\kappa-1\right)}{m_\sigma-\left(\kappa-1\right)}\right)   \left[m_\sigma\right]^N.
\end{align*}
	Also for $\sigma \in \left\{+,-\right\}$, let $\overline{\sigma} = -$ if $\sigma = +$ and $\overline{\sigma} = +$ if $\sigma = -$.

	\begin{align*}
		(\left.\tilde{\Gamma}{\bm\nu}\right.)_1 &= \frac{1}{2}\sums \frac{d_\sigma\left(m_\sigma-\alpha^2\kappa\right)}{m_\sigma-\alpha\kappa}   \left[m_\sigma\right]^N + \frac{C_1 \alpha^N }{2}   \left[\kappa\right]^N,\\
		(\left.\tilde{\Gamma}{\bm\nu}\right.)_2 &= \sums d_\sigma\left(\frac{\alpha\left(m_\sigma-\alpha^2\kappa\right)}{m_\sigma-\alpha\kappa}+\frac{C_1\alpha\left(\kappa-1\right) }{m_\sigma}\right)  \left[m_\sigma\right]^N \\
		&\hspace{1cm}{}+ \frac{C_2\left(1-\alpha\right)}{2} + \frac{C_1\left(1+2\alpha^2\left(\kappa-1\right)+\kappa\right)\alpha^N}{2\alpha\kappa}  \left[\kappa\right]^N,\nonumber
	\end{align*}
	and for $i\in \left\{3,\dots,N\right\}$ we have
	\begin{align*}
		(\left.\tilde{\Gamma}{\bm\nu}\right.)_i &= \frac{C_2\left(1+\alpha\right)}{2}+\frac{C_1}{\alpha\left(\kappa-1\right)} \sums \frac{d_\sigma m_\sigma\left(1-\kappa-m_\sigma\right)   \left[m_\sigma\right]^N}{1-\kappa+m_\sigma}\left(\frac{\left(\kappa-1\right)\alpha}{m_\sigma}\right)^i\\
		&\hspace{1cm}{}+ \frac{C_1\alpha^{N+1}\kappa  \left[\kappa\right]^N}{2} \sums \frac{c_\sigma\left(m_\sigma+\alpha^2\kappa\right)}{m_\sigma\left(m_\sigma-\alpha^2\kappa\right)}\left(\frac{m_\sigma}{\alpha\kappa}\right)^i+C_3\alpha^{i-1}.\nonumber
	\end{align*}
	Also,
	\begin{align*}
		(\left.\tilde{\Gamma}{\bm\nu}\right.)_{N+1} &= \sums c_\sigma\left(\frac{d_\sigma m_\sigma \alpha }{m_\sigma-\alpha\kappa}+\frac{m_\sigma+\left(2d_\sigma-1\right)\alpha^2\left(\kappa-1\right)}{2\left(m_\sigma-\alpha\left(\kappa-1\right)\right)}+\frac{C_1\alpha^2\kappa}{m_\sigma-\alpha^2\kappa}\right)  \left[m_\sigma\right]^N\\
		&\hspace{.5cm} {}+C_3\alpha^N -C_1\alpha^N\kappa  \left[\kappa-1\right]^N.\nonumber
	\end{align*}
	Next we compute $\nu_i (\tilde{\Gamma}\,{\bm\nu})_{i}$ for $i \in \left\{3,\dots,N\right\}$.
	It holds that
	\begin{align*}
		\nu_i (\tilde{\Gamma}\,{\bm\nu})_{i} = D_1^i + D_2^i + D_3^i + D_4^i
	\end{align*}
	with $D_1^i := C_2\left(1+\alpha\right) \nu_i /2$,
	\begin{align*}
		D_2^i := C_2 \left(1-\alpha\right) \Bigg(C_3\alpha^{i-1} + \frac{C_1}{\alpha\left(\kappa-1\right)} \sums \frac{d_\sigma m_\sigma\left(1-\kappa-m_\sigma\right)  \left[m_\sigma\right]^N}{1-\kappa+m_\sigma}\left(\frac{\alpha\left(\kappa-1\right)}{m_\sigma}\right)^i\\
		\hspace{1cm}{}+ \frac{C_1 \alpha^{N+1}\kappa \left[\kappa\right]^N}{2} \sums \frac{c_\sigma\left(m_\sigma+\alpha^2\kappa\right)}{m_\sigma\left(m_\sigma-\alpha^2\kappa\right)}\left(\frac{m_\sigma}{\alpha\kappa}\right)^i\Bigg),\nonumber
	\end{align*}
	\begin{align*}
		D_3^i := \frac{C_1C_3}{\alpha} \left(2\sums\frac{d_\sigma m_\sigma   \left[m_\sigma\right]^N}{\alpha\left(\kappa-1\right)}\left(\frac{\alpha^2\left(\kappa-1\right)}{m_\sigma}\right)^i + \sums \frac{c_\sigma \alpha^{N+1}\kappa    \left[\kappa\right]^N}{m_\sigma}\left(\frac{m_\sigma}{\kappa}\right)^i\right),
	\end{align*}
	and
	\begin{align*}
		D_4^i &:= \left(C_1\right)^2 \Bigg( \frac{2}{\left(\alpha\left(\kappa-1\right)\right)^2}\Bigg(\sums \left(d_\sigma m_\sigma   \left[m_\sigma\right]^N\right)^2 \frac{1-\kappa-m_\sigma}{1-\kappa+m_\sigma}\left(\frac{\alpha\left(\kappa-1\right)}{m_\sigma}\right)^{2i}\\
		&\hspace{3cm}{}+ d_+d_-m_+m_-  \left[m_+\right]^N \left[m_-\right]^N\frac{\left(\kappa-1\right)^2-m_+m_-}{\left(1-\alpha^2\right)\left(1-\kappa\right)}\left(\frac{\left(\alpha\left(\kappa-1\right)\right)^2}{m_+m_-}\right)^i\Bigg)\nonumber\\
		&\hspace{2cm}{}+\frac{\alpha^{2\left(N+1\right)}\kappa^2(\left[\kappa\right]^N)^2}{2} \Bigg(\sums \frac{\left(c_\sigma\right)^2\left(m_\sigma+\alpha^2\kappa\right)}{\left(m_\sigma\right)^2\left(m_\sigma-\alpha^2\kappa\right)}\left(\frac{m_\sigma}{\alpha\kappa}\right)^{2i}\nonumber\\ 
		&\hspace{6cm}{}+ \frac{c_+c_-\left(\left(\alpha^2\kappa\right)^2-m_+m_-\right)}{m_+m_-\alpha^2\left(1-\alpha^2\right)\kappa}\left(\frac{m_+m_-}{\left(\alpha\kappa\right)^2}\right)^i\Bigg)\nonumber\\
		&\hspace{2cm}{}+\frac{\alpha^N\kappa  \left[\kappa\right]^N}{\kappa-1} \Bigg( \frac{1-\left(1-\alpha^2\right)\kappa}{1-\alpha^2}\sums c_\sigma d_\sigma  \left[m_\sigma\right]^N\left(\frac{\kappa-1}{\kappa}\right)^i\nonumber\\ 
		&\hspace{3.5cm}{}+ \sums\frac{c_{\overline{\sigma}}d_\sigma m_\sigma  \left[m_\sigma\right]^N}{m_{\overline{\sigma}}}\left(\frac{1-\kappa-m_\sigma}{1-\kappa+m_\sigma}+\frac{m_{\overline{\sigma}}+\alpha^2\kappa}{m_{\overline{\sigma}}-\alpha^2\kappa}\right)\left(\frac{m_{\overline{\sigma}}\left(\kappa-1\right)}{m_\sigma\kappa}\right)^i\Bigg).\nonumber
	\end{align*}
	Summing over $i$, we find that
	\begin{align*}
		\sum_{i=3}^N D_1^i &= \frac{C_2\left(1+\alpha\right)}{2} \Bigg(C_2\left(1-\alpha\right)\left(N-2\right)+C_1 \sums \frac{c_\sigma m_\sigma\left(\frac{\alpha\kappa}{m_\sigma} \left[m_\sigma\right]^N-\frac{m_\sigma}{\alpha\kappa}\alpha^N  \left[\kappa\right]^N\right)}{m_\sigma-\alpha\kappa}\\ 
		&\hspace{4cm}{}+ 2C_1 \sums \frac{d_\sigma m_\sigma \left(\left(\frac{\alpha\left(\kappa-1\right)}{m_\sigma}\right)^2  \left[m_\sigma\right]^N -\alpha^N  \left[\kappa-1\right]^N\right)}{m_\sigma-\alpha\left(\kappa-1\right)}\Bigg)\nonumber
	\end{align*}
	Moreover,
	\begin{align*}
		\sum_{i=3}^N D_2^i &= C_2 C_3\left(\alpha^2-\alpha^N\right)\\ 
		&\hspace{.3cm}{}+ \frac{C_1C_2\left(1-\alpha^2\right)}{1+\alpha}\sums \frac{\left(d_\sigma\right)^2\left(1-\kappa-m_\sigma\right)\left(\left(m_\sigma\right)^2\alpha^N \left[\kappa-1\right]^N-\left(\alpha\left(\kappa-1\right)\right)^2 \left[m_\sigma\right]^N\right)}{d_\sigma m_\sigma\left(m_\sigma-\left(\kappa-1\right)\right)\left(\alpha\left(\kappa-1\right)-m_\sigma\right)}\nonumber\\
		&\hspace{.3cm}{}+\frac{C_1C_2\left(1-\alpha^2\right)}{2\left(1+\alpha\right)} \sums \frac{\left(c_\sigma\right)^2\left(m_\sigma+\alpha^2\kappa\right)\left(\left(\alpha\kappa\right)^2 \left[m_\sigma\right]^N-\left(m_\sigma\right)^2\alpha^N \left[\kappa\right]^N\right)}{c_\sigma \alpha\kappa\left(m_\sigma-\alpha\kappa\right)\left(m_\sigma-\alpha^2\kappa\right)},\nonumber
	\end{align*}
	\begin{align*}
		\sum_{i=3}^N D_3^i &= C_1C_3 \Bigg(\sums \frac{2d_\sigma\left((\left.m_\sigma \alpha^N\right.)^2 \left[\kappa-1\right]^N-\alpha^N\left(\alpha\left(\kappa-1\right)\right)^2 \left[m_\sigma\right]^N\right)}{m_\sigma\left(\alpha^2\left(\kappa-1\right)-m_\sigma\right)}\\
		&\hspace{5.8cm}{}+\sums \frac{c_\sigma\left(\alpha^N\kappa^2 \left[m_\sigma\right]^N-\left(m_\sigma\right)^2\alpha^N \left[\kappa\right]^N\right)}{\kappa\left(m_\sigma-\kappa\right)}\Bigg),\nonumber
	\end{align*}
	and
	\begin{align}\label{fa}
	\lefteqn{\sum_{i=3}^N D_4^i}\\
		&= \left(C_1\right)^2 \Bigg(2  \sums \frac{\left(d_\sigma m_\sigma\right)^2\left(1-\kappa-m_\sigma\right)}{\left(m_\sigma-\left(\kappa-1\right)\right)\left(\alpha\left(\kappa-1\right)-m_\sigma\right)\left(m_\sigma+\alpha\left(\kappa-1\right)\right)}\left(\alpha^N \left[\kappa-1\right]^N\right)^2\nonumber\\
		&\hspace{.5cm}{}-\frac{1}{2\left(\alpha\kappa\right)^2} \sums \frac{\left(c_\sigma\right)^2\left(m_\sigma\right)^4\left(m_\sigma+\alpha^2\kappa\right)}{\left(m_\sigma-\alpha\kappa\right)\left(m_\sigma-\alpha^2\kappa\right)\left(m_\sigma+\alpha\kappa\right)} \left(\alpha^N \left[\kappa\right]^N\right)^2\nonumber\\
		&\hspace{.5cm}{}+\sums \Bigg(\frac{\left(c_\sigma\right)^2\left(m_\sigma+\alpha^2\kappa\right)\left(\alpha\kappa\right)^2}{2\left(m_\sigma-\alpha\kappa\right)\left(m_\sigma-\alpha^2\kappa\right)\left(m_\sigma+\alpha\kappa\right)}\nonumber\\
		&\hspace{3cm}-\frac{2\left(d_\sigma\right)^2\left(\alpha\left(\kappa-1\right)\right)^4\left(1-\kappa-m_\sigma\right)}{\left(m_\sigma\right)^2\left(m_\sigma-\left(\kappa-1\right)\right)\left(\alpha\left(\kappa-1\right)-m_\sigma\right)\left(m_\sigma+\alpha\left(\kappa-1\right)\right)}\Bigg)\left(\left[m_\sigma\right]^N\right)^2\nonumber\\
		&\hspace{.5cm}{}+\sums \alpha^N\kappa\left(\frac{c_\sigma d_{\overline{\sigma}} m_{\overline{\sigma}}\left(\frac{1-\kappa-m_{\overline{\sigma}}}{1-\kappa+m_{\overline{\sigma}}}+\frac{m_\sigma+\alpha^2\kappa}{m_\sigma-\alpha^2\kappa}\right)}{m_\sigma\left(\kappa-1\right)-m_{\overline{\sigma}}\kappa}-\frac{c_\sigma d_\sigma \left(1-\left(1-\alpha^2\right)\kappa\right)}{1-\alpha^2}\right)\left[m_\sigma\right]^N\left[\kappa-1\right]^N\nonumber\\
		&\hspace{.5cm}{}+\sums\frac{\alpha^N\left(\kappa-1\right)^2}{\kappa}\left(\frac{c_\sigma d_\sigma \left(1-\left(1-\alpha^2\right)\kappa\right)}{1-\alpha^2}-\frac{c_{\overline{\sigma}}d_\sigma \left(m_{\overline{\sigma}}\right)^2\left(\frac{1-\kappa-m_\sigma}{1-\kappa+m_\sigma}+\frac{m_{\overline{\sigma}}+\alpha^2\kappa}{m_{\overline{\sigma}}-\alpha^2\kappa}\right)}{m_\sigma\left(m_{\overline{\sigma}}\left(\kappa-1\right)-m_\sigma\kappa\right)}\right)\left[m_\sigma\right]^N\left[\kappa\right]^N\nonumber\\
		&\hspace{.5cm}{}+\kappa\left(\left(\alpha^2-1\right)\kappa+1\right)\left(\alpha^N\right)^2\left(\frac{c_+c_-\left(\alpha\left(\kappa-1\right)\right)^2}{2\left(\alpha\kappa\right)^2\left(1-\alpha^2\right)} \left(\left[\kappa\right]^N\right)^2-\frac{2d_+d_-}{1-\alpha^2}\left(\left[\kappa-1\right]^N\right)^2\right)\nonumber\\
		&\hspace{.5cm}{}+\frac{\left(1-\alpha^2\right)\kappa-1}{2\kappa}\left(\frac{c_+c_-\kappa^2}{1-\alpha^2}-\frac{4d_+d_-\left(\kappa-1\right)^2}{1-\alpha^2}\right)  \left[m_+\right]^N \left[m_-\right]^N\Bigg).\nonumber
	\end{align}
	Note that 
	\begin{align*}
		\frac{c_+}{m_+\left(\kappa-1\right)-m_-\kappa} = \left(R\left(\frac{d_-}{1-\alpha^2}\frac{1-\alpha^2}{c_+}\kappa+\left(\kappa-1\right)\right)\right)^{-1}.
	\end{align*}
	The limit of ${\bm\nu}^\top \tilde{\Gamma}\,{\bm\nu}$ is a combination of the limits found in Lemma~\ref{lemma4}. 
	Thus,
	\begin{align*}
		{\bm\nu}^\top \tilde{\Gamma}\,{\bm\nu} &= \nu_1 (\left.\tilde{\Gamma}\,{\bm\nu}\right.)_1 + \nu_2 (\left.\tilde{\Gamma}\,{\bm\nu}\right.)_2 + \sum_{k=1}^4 \sum_{i=3}^N D_k^i + \nu_{N+1} (\left.\tilde{\Gamma}\,{\bm\nu}\right.)_{N+1}\\
		&\to \frac{e^{-6\rho T}\left(2e^{3\rho T}+1\right)^2}{72 \kappa^2} + \frac{e^{-6\rho T}\left(2e^{3\rho T}+1\right)^2\left(\kappa-1\right)\left(3\kappa-1\right)}{72\kappa^4}\\ 
		&\hspace{.5cm}{}+ \frac{e^{-6\rho T}}{216\kappa^4} \Big(12e^{6\rho T}\left(\kappa^4\left(2\rho T+3\right)-4\kappa^2+4\kappa-1\right)\nonumber\\
		&\hspace{2.2cm}{}-4e^{3\rho T}\left(2\kappa^4+12\kappa^2-12\kappa+3\right)-\kappa^4-12\kappa^2+12\kappa-3\Big)+ 0\nonumber\\
		&= \frac{-e^{-6\rho T}-8e^{-3\rho T}+24\rho T+36}{216}.
	\end{align*}
	
	Now we turn to the computation of ${\bm\omega}^\top  (\left.\tilde{\Gamma}-\tilde{\Gamma}^\top \right.){\bm\nu}$.
	Define
	\begin{align*}
		C_4 &:= \frac{\left(\alpha^2\left(\kappa-1\right)-\kappa\right)-\alpha\left(\frac{\alpha\left(\kappa-1\right)}{\kappa}\right)^{N+1}}{\left(\kappa-\alpha\left(\kappa-1\right)\right)\left(\alpha^2\left(\kappa-1\right)-\kappa\right)}\qquad \text{and}\qquad 
		C_5 := \frac{\alpha^2\left(\kappa-1\right)\left(\kappa+\alpha\left(\kappa-1\right)\right)}{\kappa^2\left(\alpha^2\left(\kappa-1\right)-\kappa\right)}.
	\end{align*}
	We have
	\begin{align*}
		(\left.{\bm\omega}^\top  (\left.\tilde{\Gamma}-\tilde{\Gamma}^\top \right.)\right.)_1 &= \frac{\alpha}{\kappa-\alpha\left(\kappa-1\right)}\left(1-\left(\frac{\alpha\left(\kappa-1\right)}{\kappa}\right)^N\right),\\
		(\left.{\bm\omega}^\top  (\left.\tilde{\Gamma}-\tilde{\Gamma}^\top \right.)\right.)_i &= \frac{\left(\alpha^2\left(\kappa-1\right)-\kappa\right)-\alpha\left(\frac{\alpha\left(\kappa-1\right)}{\kappa}\right)^{N+1}}{\left(\kappa-\alpha\left(\kappa-1\right)\right)\left(\alpha^2\left(\kappa-1\right)-\kappa\right)}\alpha^i+\frac{\alpha\left(\kappa+\alpha\left(\kappa-1\right)\right)}{\kappa\left(\alpha^2\left(\kappa-1\right)-\kappa\right)}\left(\frac{\alpha\left(\kappa-1\right)}{\kappa}\right)^{N+1-i}\hspace{-.8cm},\nonumber
	\end{align*}
	for $i \in \left\{2,\dots,N\right\}$, and
	\begin{align*}
		(\left.{\bm\omega}^\top  (\left.\tilde{\Gamma}-\tilde{\Gamma}^\top \right.)\right.)_{N+1} &= \frac{\alpha\Big(\alpha^N\left(\kappa-\alpha^2\left(\kappa-1\right)\right)\kappa+\alpha^2\left(\kappa-1\right)\left(\kappa-1+\left(\frac{\alpha^2\left(\kappa-1\right)}{\kappa}\right)^N\right)-\kappa^2\Big)}{\kappa\left(\kappa-\alpha\left(\kappa-1\right)\right)\left(\kappa-\alpha^2\left(\kappa-1\right)\right)}.
	\end{align*}
	For $i \in \left\{2,\dots,N\right\}$:
	\begin{align*}
		\lefteqn{(\left.{\bm\omega}^\top  (\left.\tilde{\Gamma}-\tilde{\Gamma}^\top \right.) {\bm\nu} \right.)_i}\\ 
		&= C_2 \left(1-\alpha\right)\left(C_4\alpha^i+C_5\left(\frac{\alpha\left(\kappa-1\right)}{\kappa}\right)^N\left(\frac{\kappa}{\alpha\left(\kappa-1\right)}\right)^i\right)\nonumber\\
		&\hspace{.5cm}{}+2C_1\sums \frac{d_\sigma m_\sigma}{\alpha\left(\kappa-1\right)}\left(C_4\left(\frac{\alpha^2\left(\kappa-1\right)}{m_\sigma}\right)^i+C_5\left(\frac{\alpha\left(\kappa-1\right)}{\kappa}\right)^N\left(\frac{\kappa}{m_\sigma}\right)^i\right)  \left[m_\sigma\right]^N\nonumber\\
		&\hspace{.5cm}{}+C_1 \sums \frac{c_\sigma \alpha^{N+1}\kappa }{m_\sigma}\left(C_4\left(\frac{m_\sigma}{\kappa}\right)^i+C_5\left(\frac{\alpha\left(\kappa-1\right)}{\kappa}\right)^N\left(\frac{m_\sigma}{\alpha^2\left(\kappa-1\right)}\right)^i\right)  \left[\kappa\right]^N.\nonumber
	\end{align*}
	Summing over $i$, we arrive at
	\begin{align}\label{fc}
	\lefteqn{\sum_{i=2}^N(\left.{\bm\omega}^\top  (\left.\tilde{\Gamma}-\tilde{\Gamma}^\top \right.) {\bm\nu} \right.)_i}\\ 
		&= C_2\left(C_4\left(\alpha^2-\alpha^{N+1}\right)+C_5\left(1-\alpha\right)\frac{\alpha\kappa\left(\kappa-1\right)-\alpha^N\kappa^2\left(\left(\kappa-1\right)/\kappa\right)^N}{\alpha\left(\kappa-1\right)\left(\kappa-\alpha\left(\kappa-1\right)\right)}\right)\nonumber\\
		&\hspace{.5cm}{}+ C_1 \sums \Bigg(C_4\left(\frac{2d_\sigma\alpha^3\left(\kappa-1\right)}{m_\sigma-\alpha^2\left(\kappa-1\right)}-\frac{c_\sigma\alpha^{N+1}\kappa}{\kappa-m_\sigma}\right)\nonumber\\
		&\hspace{3cm}{}+C_5\Bigg(\frac{2d_\sigma \alpha^N\kappa^2 \left(\frac{\left(\kappa-1\right)}{\kappa}\right)^N}{\alpha\left(\kappa-1\right)\left(m_\sigma-\kappa\right)}-\frac{c_\sigma \alpha\kappa}{\alpha^2\left(\kappa-1\right)-m_\sigma}\Bigg)\Bigg)  \left[m_\sigma\right]^N\nonumber\\
		&\hspace{.5cm}{}-C_1\alpha\kappa\left(\frac{C_4\alpha^N\left(\kappa+1\right)}{\kappa}+\frac{C_5\left(1-2\alpha^2\right)\left(\frac{\alpha^2\left(\kappa-1\right)}{\kappa}\right)^N}{\alpha^2\left(1-\alpha^2\right)\left(\kappa-1\right)}\right) \left[\kappa\right]^N\nonumber\\
		&\hspace{.5cm}{}-2C_1\left(C_4\alpha^{2N+1}+\frac{C_5\kappa\left(1-\left(\alpha^2-1\right)\kappa\right)\alpha^N}{\alpha\left(1-\alpha^2\right)\left(\kappa-1\right)}\right) \left[\kappa-1\right]^N\nonumber.
	\end{align}
	Note that
	\begin{align*}
		\frac{\left(\left(\kappa-1\right)/\kappa\right)^N}{\kappa-m_+} &= \frac{\left(\left(\kappa-1\right)/\kappa\right)^N}{1-\alpha^2} \frac{1-\alpha^2}{c_+}\frac{c_+}{\kappa-m_+},\hspace{1cm}\text{ and}\\
		\frac{ \left[m_-\right]^N}{m_--\alpha\left(\kappa-1\right)} &= \frac{ \left[m_-\right]^N}{1-\alpha^2} \frac{1-\alpha^2}{d_-}\frac{d_-}{m_--\alpha\left(\kappa-1\right)}.
	\end{align*}
	Again, Lemma~\ref{lemma4} gives us all the necessary limits; we find that
	\begin{align*}
		{\bm\omega}^\top (\left.\tilde{\Gamma}-\tilde{\Gamma}^\top \right.){\bm\nu} &= (\left.{\bm\omega}^\top (\left.\tilde{\Gamma}-\tilde{\Gamma}^\top \right.)\right.)_1 \nu_1 + \sum_{i=2}^N(\left.{\bm\omega}^\top  (\left.\tilde{\Gamma}-\tilde{\Gamma}^\top \right.) {\bm\nu} \right.)_i + (\left.{\bm\omega}^\top (\left.\tilde{\Gamma}-\tilde{\Gamma}^\top \right.)\right.)_{N+1} \nu_{N+1}\\
		&\to \frac{e^{-3\rho T}+2}{6\kappa}+\frac{2\left(2\kappa-1\right)-\left(\kappa+1\right)e^{-3\rho T}}{6\kappa}+0
		= \frac{-e^{-3\rho T}+4}{6}.
	\end{align*}
	Finally,
	\begin{align}
		{\bm\omega}^\top \tilde{\Gamma}\,{\bm\omega} &= \frac{1}{2}\sum_{i=1}^N \left( \omega_i\right)^2 + \frac{\left( \omega_{N+1}\right)^2}{2}+\sum_{i=1}^{N-1}\sum_{j=i+1}^N  \omega_i \omega_j \alpha^{j-i} +  \omega_{N+1} \sum_{i=1}^N  \omega_i \alpha^{N+1-i}\nonumber\\
		&= \frac{\kappa^2-2\alpha\left(1-\alpha^2\right)\kappa\left(\kappa-1\right)-\alpha^2\left(2-\alpha^2\right)\left(\kappa-1\right)^2}{2\left(\kappa-\alpha\left(\kappa-1\right)\right)^3\left(\kappa+\alpha\left(\kappa-1\right)\right)}\label{fd}\\
		&\hspace{.5cm}{}+\frac{\left(1+\alpha\right)\left(1-\alpha\right)N}{2\left(\kappa-\alpha\left(\kappa-1\right)\right)^2}-\frac{\alpha^{N+2}\left(1-\left(1-\alpha\right)\kappa\right)\left(\kappa-1\right)\left(\left(\kappa-1\right)/\kappa\right)^N}{\kappa\left(\kappa-\alpha\left(\kappa-1\right)\right)^3}\nonumber\\
		&\hspace{.5cm}{}+\frac{\alpha^{2\left(N+2\right)}\left(\kappa-1\right)^2\left(2\kappa-1\right)\left(\left(\kappa-1\right)/\kappa\right)^{2N}}{2\left(\kappa-\alpha\left(\kappa-1\right)\right)^3\left(\kappa+\alpha\left(\kappa-1\right)\right)\kappa^2}.\nonumber
	\end{align}
	For $\kappa > 1/2$, we find ${\bm\omega}^\top \tilde{\Gamma}\,{\bm\omega} \to \left(2\rho T+1\right)/2$.	 
\end{proof}

\begin{proof}[Proof of Theorem~\ref{costs asymptotics thm} {\rm (a)}] It was shown in \eqref{omega.1 asympt for kappa>1/2 eq} that 
\begin{equation}\label{1top omega convergence eq}
\mathbf{1}^\top {\bm\omega}= \sum_{i=1}^{N+1}  \omega_i\lra \rho T+1\quad\text{as $N\uparrow\infty$.}
\end{equation}
	The limit  of $\mathbf{1}^\top {\bm\nu} = \sum_{i=1}^{N+1} \nu_i$ was obtained in \eqref{(partial) sum nui kappa=1 eq} and  \eqref{sum nui kappa>1/2, kappa neq 1 eq} for the respective cases $\kappa=1$ and $\kappa\neq1$, $\kappa>1/2$.
	The limits of ${\bm\nu}^\top \tilde{\Gamma}\,{\bm\nu}, {\bm\omega}^\top (\left.\tilde{\Gamma}-\tilde{\Gamma}^\top \right.){\bm\nu}$, and ${\bm\omega}^\top \tilde{\Gamma}\,{\bm\omega}$ can be found in Lemma~\ref{lemma7}.
	Plugging these results into \eqref{am} yields assertion (a).	 
\end{proof}

\subsection{Proof of Theorem~\ref{costs asymptotics thm} (b)}


\begin{lemma}\label{lemma7 kappa =1/2}
	Let $\kappa = 1/2$.
	It holds that
	\begin{align*}
		\lim_{\substack{N\uparrow\infty\\
	N\,\text{\rm even}}} {\bm\nu}^\top \tilde{\Gamma}\,{\bm\nu}&= \frac{2e^{6\rho T}\left(3\rho T+5\right)+e^{3\rho T}+3\rho T+7}{54e^{6\rho T}+27},\\
	\lim_{\substack{N\uparrow\infty\\
	N\,\text{\rm even}}}	 {\bm\omega}^\top ( \tilde{\Gamma}-\tilde{\Gamma}^\top  ){\bm\nu}&= \frac{4e^{6\rho T}-6e^{5\rho T}+e^{3\rho T}-3e^{-\rho T}+4}{6e^{6\rho T}+3},\\
	\lim_{\substack{N\uparrow\infty\\
	N\,\text{\rm even}}}{\bm\omega}^\top \tilde{\Gamma}\,{\bm\omega}&= e^{-\rho T}+\rho T+1.
	\end{align*}
Moreover,
	\begin{align*}
	\lim_{\substack{N\uparrow\infty\\
	N\,\text{\rm odd}}}{\bm\nu}^\top \tilde{\Gamma}\,{\bm\nu}&= \frac{2e^{6\rho T}\left(3\rho T+5\right)-3e^{3\rho T}-3\rho T-7}{54e^{6\rho T}-27},\\
	\lim_{\substack{N\uparrow\infty\\
	N\,\text{\rm odd}}}{\bm\omega}^\top (\left.\tilde{\Gamma}-\tilde{\Gamma}^\top \right.){\bm\nu}&=\frac{-4e^{6\rho T}-6e^{5\rho T}+3e^{3\rho T}+3e^{-\rho T}+4}{-6e^{6\rho T}+3},\text{ and}\\
		\lim_{\substack{N\uparrow\infty\\
	N\,\text{\rm odd}}}{\bm\omega}^\top \tilde{\Gamma}\,{\bm\omega}&=  -e^{-\rho T}+\rho T+1.
	\end{align*}
\end{lemma}
\begin{proof}
	As explained in the proof of Lemma~\ref{lemma7}, the representations of ${\bm\nu}^\top \tilde{\Gamma}\,{\bm\nu}$,  ${\bm\omega}^\top (\left.\tilde{\Gamma}-\tilde{\Gamma}^\top \right.){\bm\nu}$, and ${\bm\omega}^\top \tilde{\Gamma}\,{\bm\omega}$ obtained in the proof of that lemma for $\kappa\neq1$ are also valid for $\kappa = 1/2$.
	In \eqref{fa}, note that for $\sigma \in \left\{+,-\right\}$,
	\begin{align*}
		\frac{c_{\overline{\sigma}} d_\sigma \left(m_{\overline{\sigma}}\right)^2\left(\frac{1-\kappa-m_\sigma}{1-\kappa+m_\sigma}+\frac{m_{\overline{\sigma}}+\alpha^2\kappa}{m_{\overline{\sigma}}-\alpha^2\kappa}\right)}{m_\sigma\left(m_{\overline{\sigma}}\left(\kappa-1\right)-m_\sigma\kappa\right)} = \frac{16\left(m_{\overline{\sigma}}\right)^2\alpha^2}{m_\sigma\left(27\alpha^4-42\alpha^2+27\right)}.
	\end{align*}
	Also in \eqref{fc},
	\begin{align*}
		\frac{C_5}{1-\alpha^2} &= \frac{\alpha^2\left(\kappa-1\right)}{\kappa^2\left(\alpha^2\left(\kappa-1\right)-\kappa\right)} \frac{\kappa+\alpha\left(\kappa-1\right)}{1-\alpha^2},\\
		\frac{C_5}{m_+-\kappa} &= \frac{C_5}{1-\alpha^2} \frac{1-\alpha^2}{c_+}\frac{c_+}{m_+-\kappa},\hspace{2.5cm}\text{ and}\\
		\frac{C_5}{m_--\alpha^2\left(\kappa-1\right)} &= \frac{C_5}{1-\alpha^2} \frac{1-\alpha^2}{d_-}\frac{d_-}{m_--\alpha^2\left(\kappa-1\right)}.\hspace{1cm}
	\end{align*}
Equation	\eqref{fd} simplifies to
	\begin{align*}
		-\frac{2\Big(\alpha^2\left(1-2\left(-1\right)^N \alpha^N\right)-2\alpha-1-\left(1+\alpha\right)\left(1-\alpha\right)N\Big)}{\left(1+\alpha\right)^2}.
	\end{align*}
 	Plugging in the limits from Lemma~\ref{lemma4} completes the proof.	 
\end{proof}
 
\begin{proof}[Proof of Theorem~\ref{costs asymptotics thm} {\rm (b)}]	To prove part (b) of Theorem~\ref{costs asymptotics thm}, we now proceed just as in the proof of part (a), this time using the limits obtained in \eqref{sum omi limit formula kappa=1/2},  \eqref{limits sum nui kappa=1/2 eq}, and Lemma~\ref{lemma7 kappa =1/2}.	 
\end{proof}

\subsection{Proof of Corollary~\ref{cost comparison cor}}
Let $x=y.$
	For $z > 0,$ define
	\begin{align*}
		c_0^+\left(z\right) &:= \frac{\left(2x\right)^2\left(6e^{6z}+3\right)}{2\left(2e^{6z}\left(3z+5\right)+e^{3z}+3z+7\right)},\\
		c_0^-\left(z\right) &:= \frac{\left(2x\right)^2\left(6e^{6z}-3\right)}{2\left(2e^{6z}\left(3z+5\right)-3e^{3z}-3z-7\right)},\\
		c_{1/4}\left(z\right) &:= \frac{\left(2x\right)^2\left(36e^{6z}\left(8z+13\right)-60e^{3z}-3\right)}{16\left(2e^{3z}\left(3z+5\right)-1\right)^2}.
	\end{align*}
	By Theorem~\ref{costs asymptotics thm}, these functions correspond to the limiting expected costs of trader $X$  if $\rho T = z$, trader $X$ has initial endowment  $x$,  trader $Y$ has initial endowment $y=x$, and we have $\theta = 0$ (for $c_0^\pm$) or $\theta > 0$ (for $c_{1/4}$).
	Without loss of generality, we can take $x=y=1/2$. 
	We find that
	\begin{align*}
		 &\hphantom{{}=}{} \frac{6e^{6z}+3}{2\left(2e^{6z}\left(3z+5\right)+e^{3z}+3z+7\right)}
		 - \frac{6e^{6z}-3}{2\left(2e^{6z}\left(3z+5\right)-3e^{3z}-3z-7\right)} \\
		 &= - \frac{3e^{3z} (2e^{3z}+1)^2}{(2e^{6z} (3z+5) + e^{3z} + 3z + 7) (2e^{6z} (3z+5) - 3e^{3z} - 3z - 7)}.
	\end{align*}
	Rewriting the second factor in the denominator shows that the expression above is negative:
	$$
		2e^{6z} (3z+5) - 3e^{3z} - 3z - 7 = (e^z - 1) (e^{2z} + e^z + 1) (10e^{3z}+7) + 3z (2e^{6z}-1) > 0.
	$$
	This shows that $c_0^+ < c_0^-$ for all $z > 0.$
	It remains to show:
	If $z > \log(4+\sqrt{62}/3)/3,$ then
	$$
		\frac{6e^{6z}+3}{2\left(2e^{6z}\left(3z+5\right)+e^{3z}+3z+7\right)} 
		> \frac{36e^{6z}\left(8z+13\right)-60e^{3z}-3}{16\left(2e^{3z}\left(3z+5\right)-1\right)^2}.
	$$
	Multiplying by both denominators, we find that this is true if and only if
	\begin{align*}
		0 < 16\left(2e^{3z}\left(3z+5\right)-1\right)^2 \left(6e^{6z}+3\right) - 2\left(2e^{6z}\left(3z+5\right)+e^{3z}+3z+7\right) \left(36e^{6z}\left(8z+13\right)-60e^{3z}-3\right)
	\end{align*}
	A tedious rearrangement of terms shows that the expression on the right hand side equals
	$$
		6 \left(2e^{3z} +1\right)^2 \left(3z+5\right) \left(2e^{6z} - e^{3z} \frac{48z+79}{3z+5} + \frac{3z+15}{3z+5} \right).
	$$
	Since $48z+79 < 16(3z+5)$ and $3z+15 > 3z+5$ for all $r > 0$, this expression is larger than
	$$
		6 \left(2e^{3z} + 1\right)^2 \left(3z+5\right) \left(2e^{6z} - 16e^{3z} + 1\right).
	$$
	The real-valued function $a \mapsto 2a^2 - 16a + 1$ has the two roots $a_1 \coloneqq 4 - \sqrt{62}/2$ 
	and $a_2 \coloneqq 4 + \sqrt{62}/2.$
	Since $a_1 < 1 < a_2,$ the real-valued function $b \mapsto 2e^{6b} - 16 e^{3b} + 1$ has exactly one positive root $\log(a_2)/3.$
	Hence, for every $z > \log(a_2)/3 = \log(4 + \sqrt{62}/2) / 3$ it holds that
	\begin{align*}
		0 &< 6 \left(2e^{3z} + 1\right)^2 \left(3z+5\right) \left(2e^{6z} - 16e^{3z} + 1\right)\\
		&< 16\left(2e^{3z}\left(3z+5\right)-1\right)^2 \left(6e^{6z}+3\right) - 2\left(2e^{6z}\left(3z+5\right)+e^{3z}+3z+7\right) \left(36e^{6z}\left(8z+13\right)-60e^{3z}-3\right).
	\end{align*}
\qed

\subsection{Proof of Corollary~\ref{tax corollary}}

\begin{lemma}\label{tax lemma}For $\theta>0$, as $N\uparrow\infty$,
\begin{align*}
	\theta\bm v^\top\bm v&\lra \frac{9(1+2e^{3\rho T})^2}{4(1-2e^{3\rho T}(5+3\rho T))^2},\\
	\theta\bm w^\top\bm w&\lra\frac{1}{4(\rho T+1)^2}.
\end{align*}
\end{lemma}

\begin{proof}As for the first limit, since $(\Gamma+\tilde\Gamma+2\theta\Id)\bm\nu=\bm1$ and $\Gamma=\tilde\Gamma+\tilde\Gamma^\top$, we have 
$$\bm\nu^\top\bm1=\bm\nu^\top\Gamma\bm\nu+\bm\nu^\top\tilde\Gamma\bm\nu+2\theta\bm\nu^\top \bm\nu=3\bm\nu^\top\tilde\Gamma\bm\nu+2\theta\bm\nu^\top\bm\nu.
$$
Solving for $\theta\bm\nu^\top\bm\nu$, applying the limiting formula for $\bm\nu^\top\bm1$ from \eqref{sum nui kappa>1/2, kappa neq 1 eq}
 and the one for $\bm\nu^\top\tilde\Gamma\bm\nu$ from Lemma~\ref{lemma7}, and then simplifying the result leads to 
 $$\lim_{N\uparrow\infty}\theta\bm\nu^\top\bm\nu=\frac{e^{-6\rho T}(1+2e^{3\rho T})^2}{144}.
 $$
Since $\bm v=(\bm\nu^\top\bm1)^{-1}\bm\nu$, the preceding limit, another application of \eqref{sum nui kappa>1/2, kappa neq 1 eq}, and a straightforward computation  finally yield the asserted convergence of $\theta\bm v^\top\bm v$.

To prove the second limit, we use that $(\Gamma-\tilde\Gamma+2\theta\Id)\bm\omega=\bm1$ and that $\bm\omega^\top(\Gamma-\tilde\Gamma)\bm\omega=\bm\omega^\top\tilde\Gamma^\top\bm\omega=\bm\omega^\top\tilde\Gamma\bm\omega$. Hence, using the limits from \eqref{1top omega convergence eq} and Lemma~\ref{lemma7},
$$\theta\bm\omega^\top\bm\omega=\frac12\bm\omega^\top\bm1-\frac12\bm\omega^\top\tilde\Gamma\bm\omega\lra\frac{\rho T+1}{2}-\frac{2\rho T+1}{4}=\frac14\qquad\text{as $N\uparrow\infty$.}
$$
Since $\bm w=(\bm1^\top\bm\omega)^{-1}\bm\omega$, the preceding limit and another application of  \eqref{1top omega convergence eq} conclude the proof.
\end{proof}

 \begin{proof}[Proof of Corollary~\ref{tax corollary}] We have 
$$\text{TR}_N=\theta\bm\xi^\top\bm\xi+\theta\bm\eta^\top\bm\eta=\frac12\big(\theta\bm v^\top\bm v+\theta\bm w^\top\bm w).
$$
Therefore, the formula for the limit of $\text{TR}_N$ follows from Lemma~\ref{tax lemma}. Next, we get from 
 Theorem~\ref{costs asymptotics thm} and a tedious though straightforward computation that
 \begin{align*}\liminf_{N\uparrow\infty}(\text{TR}_N-\text{TC}_N)&=\lim_{N\uparrow\infty}(\text{TR}_{2N}-\text{TC}_{2N})\\
 &=\frac{(x+y)^2\,3 \left(2 e^{3 \rho T }+1\right)^2 \left(3 (\rho T +3)+2 e^{6 \rho T } (3 \rho T +5)-e^{3
   \rho T } (12 \rho T +19)\right)}{2 \left(1-2 e^{3 \rho T } (3 \rho T +5)\right)^2 \left(3 \rho T
   +e^{3 \rho T }+2 e^{6 \rho T } (3 \rho T +5)+7\right)}.
 \end{align*}
Clearly, the latter expression vanishes for $x=-y$ and it is strictly positive if and only if $x\neq -y$ and $f(\rho T)>0$, where
$$f(z):=3 (z +3)+2 e^{6 z } (3 z +5)-e^{3
  z } (12 z +19).
$$
But one easily sees that $f(0)=0$ and $f'(z)=3(e^{3z}-1)(2e^{3z}(11+6z)-1)>0$ for all $z>0$. This proves the assertion.
 \end{proof}

 \section{Proofs of Theorem~\ref{main th continuous} and Corollary~\ref{cont costs cor}}

\subsection{Proof of Theorem~\ref{main th continuous}}

 \begin{proof}[Proof of Theorem~\ref{main th continuous} {\rm (a)}] Let $X^*$ and $Y^*$ be as in \eqref{equilibrium strategies continuous}  and $\theta=\theta^*=1/4$. A straightforward computation yields that, for all $t\in[0,T]$, 
$$
 \int_{[0,T]}  e^{-\rho|t-s|}\,d X^*_s +\int_{[0, t)}  e^{-\rho(t-s)}\,d Y^*_s+\frac1{2}\Delta Y^*_t+2\theta\Delta X^*_t=-\frac{1}{2}\Big(\frac{ x-y}{\rho T+1}+\frac{18 (x+y)}{10+6\rho T-e^{-3\rho T}}\Big)
 $$
 and
 $$
 \int_{[0,T]}  e^{-\rho|t-s|}\,d Y^*_s +\int_{[0, t)}  e^{-\rho(t-s)}\,d X^*_s+\frac1{2}\Delta X^*_t+2\theta\Delta Y^*_t =\frac{1}{2}\Big(\frac{ x-y}{\rho T+1}-\frac{18 (x+y)}{10+6\rho T-e^{-3\rho T}}\Big).
$$
Proposition~\ref{deterministic fredholm prop} therefore implies that $(X^*,Y^*)$ is a Nash equilibrium in $\mathscr{X}_{\text{det}}(x,[0,T])\times\mathscr{X}_{\text{det}}(y,[0,T])$. Thus, Lemma~\ref{det is adp lm} yields that $(X^*,Y^*)$ is a Nash equilibrium also in $\mathscr{X}(x,[0,T])\times\mathscr{X}(y,[0,T])$. The uniqueness of the latter Nash equilibrium was proved in Proposition~\ref{unique prop}. This concludes the proof of part (a).
 \end{proof}

\begin{proof}[Proof of Theorem~\ref{main th continuous} {\rm (b)}]

Now we consider the case $\theta\neq \theta^* = 1/4$.  Suppose first that $x = y = 0$.  It is obvious from Proposition~\ref{deterministic fredholm prop} that $X^* = Y^* = 0$ is a Nash equilibrium in $\mathscr{X}_{\text{det}}(0,[0,T])\times\mathscr{X}_{\text{det}}(0,[0,T])$.  Lemma~\ref{det is adp lm}  and Proposition~\ref{unique prop} thus yield that this is also the unique Nash equilibrium in  $\mathscr{X}(0,[0,T])\times\mathscr{X}(0,[0,T])$. 

Now we will prove that the existence of  a Nash equilibrium implies the condition  $x = y = 0$. To this end, let $(X^*,Y^*)$ be a Nash equilibrium in $\mathscr{X}(x,[0,T])\times\mathscr{X}(y,[0,T])$. 

{\it Step 1.} In the first step of the proof, we show that $X^*$ and $Y^*$ are necessarily continuous on $[0,T)$. To this end, we pick $\varepsilon>0$ and define  $\tau:=\inf\{t>0\,|\,|\Delta X^*_t|\ge\varepsilon\}\wedge T$. Then $\tau$ is a stopping time for the right-continuous filtration $(\mathscr{F}_t)_{t\ge0}$. Next, we take a sequence $(\tau_n)$ of stopping times satisfying the following three conditions: $\tau\le\tau_n\le T$;  on $\{\tau<T\}$ we have $\tau_n\downarrow\tau$ $\mathbb{P}$-a.s.;
and on $\{\tau_n<T\}$ we have $\Delta X^*_{\tau_n}=\Delta Y^*_{\tau_n}=0$   $\mathbb{P}$-a.s. The existence of such a sequence will be proved in Lemma~\ref{stopping time lemma} below.

Since the total variations of $X^*(\omega)$ and $Y^*(\omega)$  are uniformly bounded in $\omega$,  dominated convergence yields that
\begin{equation}\label{stime limit eq 1}
\lim_{n\uparrow\infty}\mathbb{E}\bigg[\int_{[0,T]}  e^{-\rho|\tau_n-s|}\,d X^*_s\bigg]=\mathbb{E}\bigg[\int_{[0,T]}  e^{-\rho|\tau-s|}\,d X^*_s\bigg].
\end{equation}
and that
\begin{equation}\label{stime limit eq 2}
\lim_{n\uparrow\infty}\mathbb{E}\bigg[\int_{[0, \tau_n)}  e^{-\rho(\tau_n-s)}\,d Y^*_s\bigg]=\mathbb{E}\bigg[\int_{[0, \tau]}  e^{-\rho(\tau-s)}\,d Y^*_s\bigg].
\end{equation}
Now let $\eta$ be as in Proposition~\ref{random fredholm prop}. 
Then  \eqref{stime limit eq 1} and \eqref{stime limit eq 2} yield  that
\begin{equation}\label{tau limit eq}
\begin{split}
\eta&=\lim_{n\to\infty}\mathbb{E}\bigg[\int_{[0,T]}  e^{-\rho|\tau_n-s|}\,d X^*_s +\int_{[0, \tau_n)}  e^{-\rho(\tau_n-s)}\,d Y^*_s +\frac1{2}\Delta Y^*_{\tau_n}+2\theta\Delta X^*_{\tau_n}\,\Big|\,\mathscr{F}_{\tau}\bigg]\\
&=\mathbb{E}\bigg[\int_{[0,T]}  e^{-\rho|\tau-s|}\,d X^*_s +\int_{[0, \tau]}  e^{-\rho(\tau-s)}\,d Y^*_s+\Big(\frac12\Delta Y^*_T+2\theta \Delta X^*_T\Big)\Ind{\{\tau=T\}}\,\Big|\,\mathscr{F}_{\tau}\bigg].
\end{split}
\end{equation}
On the other hand, taking  $\sigma=\tau$ in \eqref{Fredholm formula} yields that
\begin{equation}\label{eq sigma=tau}
\eta=\mathbb{E}\bigg[\int_{[0,T]}  e^{-\rho|\tau-s|}\,d X^*_s +\int_{[0, \tau)}  e^{-\rho(\tau-s)}\,d Y^*_s +\frac1{2}\Delta Y^*_{\tau}+2\theta\Delta X^*_{\tau}\,\Big|\,\mathscr{F}_{\tau}\bigg].
\end{equation}
By subtracting \eqref{eq sigma=tau} from \eqref{tau limit eq}, we obtain
\begin{equation*}
\Big(\frac1{2}\Delta Y^*_{\tau}-2\theta\Delta X^*_{\tau}\Big)\Ind{\{\tau<T\}}=0.
\end{equation*}
Keeping the definition of $\tau$ but otherwise reversing the roles of $X^*$ and $Y^*$ in the preceding argument gives
\begin{equation*}
\Big(\frac1{2}\Delta X^*_{\tau}-2\theta\Delta Y^*_{\tau}\Big)\Ind{\{\tau<T\}}=0.
\end{equation*}
In view of our assumption $\theta\neq1/4$ and the definition of $\tau$ we must conclude that $\tau=T$ $\mathbb{P}$-a.s. Sending $\varepsilon$ to zero along a countable sequence now yields  that $X^*$ must be continuous on $[0,T)$. Reversing the roles of $X^*$ and $Y^*$ in the entire argument gives the same result for $Y^*$.

{\it Step 2.}  We will show here that $\Delta X^*_T=\Delta Y^*_T=0$. To this end, let $(\tau_n)_{n\in\mathbb{N}}$ be a sequence of strictly positive stopping times such that $\tau_n\uparrow T$. For each $n\in\mathbb{N}$, let $\big(\sigma^n_m\big)_{m\in\mathbb{N}}$ be a sequence of stopping times such that $\sigma_m^n\ge\tau_n$ and  $\sigma^n_m\uparrow T$  as $m\uparrow\infty$. For each $\tau_n$, let $\eta_n$ be as in Proposition~\ref{random fredholm prop}. We know from the preceding step that $X^*$ and $Y^*$ are continuous on $[0,T)$, so that $\Delta Y^*_{\sigma^n_m}=\Delta X^*_{\sigma^n_m}=0$. We therefore have
\begin{equation}\label{contradic T eq1}
\begin{split}
\eta_n&=\lim_{m\uparrow\infty}\mathbb{E}\bigg[\int_{[0,T]}  e^{-\rho|\sigma^n_m-s|}\,d X^*_s +\int_{[0, \sigma^n_m)}  e^{-\rho(\sigma^n_m-s)}\,d Y^*_s \,\Big|\,\mathscr{F}_{\tau_n}\bigg]\\
&=\mathbb{E}\bigg[\int_{[0,T]}  e^{-\rho|T-s|}\,d X^*_s +\int_{[0, T)}  e^{-\rho(T-s)}\,d Y^*_s \,\Big|\,\mathscr{F}_{\tau_n}\bigg].
\end{split}
\end{equation}
On the other hand, we also have
\begin{equation}\label{contradic T eq2}
\eta_n=\mathbb{E}\bigg[\int_{[0,T]}  e^{-\rho|T-s|}\,d X^*_s +\int_{[0,T)}  e^{-\rho(T-s)}\,d Y^*_s +\frac1{2}\Delta Y^*_{T}+2\theta\Delta X^*_{T}\,\Big|\,\mathscr{F}_{\tau_n}\bigg].
\end{equation}
Comparing \eqref{contradic T eq1} and \eqref{contradic T eq2} yields that
\begin{align}\label{DeltaXT=0}
\mathbb{E}\Big[\frac12\Delta Y^*_T+2\theta\Delta X^*_T\,\Big|\,\mathscr{F}_{\tau_n}\Big]=0.
\end{align}
Reversing the roles of $X^*$ and $Y^*$ yields that\begin{align}\label{DeltaYT=0}
\mathbb{E}\Big[\frac12\Delta X^*_T+2\theta\Delta Y^*_T\,\Big|\,\mathscr{F}_{\tau_n}\Big]=0.
\end{align}
Sending $n$ to infinity and using martingale convergence together with the fact that $\sigma (\bigcup_{n=0}^{\infty}\mathscr{F}_{\tau_n})=\mathscr{F}_{T-}$ yields that we may replace $\mathscr{F}_{\tau_n}$ in \eqref{DeltaXT=0} and \eqref{DeltaYT=0} by $\mathscr{F}_{T-}$. But the fact that $X^*_T=Y^*_T=0$ implies that 
$\Delta X^*_T=-X^*_{T-}$ and $ \Delta Y^*_T=-Y^*_{T-}$ are $\mathscr{F}_{T-}$-measurable. Thus,  $\mathbb{P}$-a.s., $\frac12\Delta Y^*_T+2\theta\Delta X^*_T=0$ and $\frac12\Delta X^*_T+2\theta\Delta Y^*_T=0$. Our assertion $
\Delta X^*_T=\Delta Y^*_T=0$ now follows as in the first step of the proof by using the assumption $\theta\neq1/4$.

{\it Step 3.} We now prove $\Delta X^*_0=\Delta Y^*_0=0$. Proceeding analogously to Step 2, we let $\tau=0$ and take a sequence of deterministic times $s_n\downarrow 0$. As before, an application of Proposition~\ref{random fredholm prop} yields that 
\begin{align*}
\mathbb{E}\bigg[\int_{[0,T]}  e^{-\rho t}\,d X^*_t +\frac12\Delta Y^*_0+2\theta\Delta X^*_0\bigg]=\lim_{n\uparrow\infty}\mathbb{E}\bigg[\int_{[0,T]}  e^{-\rho |s_n-t|}\,d X^*_t +\int_{[0,s_n)}e^{-\rho(s_n-t)}\,dY^*_t\bigg].
\end{align*}
After exchanging limit and integration, we arrive at $\frac12\Delta Y^*_0=2\theta\Delta X^*_0$. Reversing the roles of $X^*$ and $Y^*$ gives $\frac12\Delta X^*_0=2\theta\Delta Y^*_0$, and as above our assumption $\theta\neq\frac14$ yields $\Delta X^*_0=\Delta Y^*_0=0$.

{\it Step 4.}  In this step, we will show  that $X^*_t=Y^*_t\text{ }\mathbb{P}$-a.s.~for all $t\in[0,T]$. To this end, we take an arbitrary $[0,T]$-valued stopping time $\tau$ and apply 
 Proposition~\ref{random fredholm prop} to obtain two $\mathscr{F}_{\tau}$-measurable random variables $\eta_1$ and $\eta_2$ such that  for all stopping times $\sigma$ with $\tau\le\sigma\le T$, 
\begin{equation}\label{continuous contradiction eq2}
\begin{split}
&\mathbb{E}\bigg[\int_{[0,T]}  e^{-\rho|\sigma-t|}\,dX^*_t+\int_{[0,\sigma)}  e^{-\rho(\sigma-t)}\,dY^*_t\,\Big|\,\mathscr{F}_{\tau}\bigg]=\eta_1,\\
&\mathbb{E}\bigg[\int_{[0,T]}  e^{-\rho|\sigma-t|}\,dY^*_t+\int_{[0,\sigma)}  e^{-\rho(\sigma-t)}\,dX^*_t\,\Big|\,\mathscr{F}_{\tau}\bigg]=\eta_2.
\end{split}
\end{equation}
Here we have used the already established continuity of $X^*$ and $Y^*$ on $[0,T]$. 
Subtracting  the first  equation in \eqref{continuous contradiction eq2} from the second one yields
$$
\mathbb{E}\bigg[\int_{[\sigma,T]}  e^{-\rho(t-\sigma)}\,d(X^*_t-Y^*_t)\,\Big|\,\mathscr{F}_{\tau}\bigg]=\eta_1-\eta_2.
$$
Setting $\sigma=T$ and using that $\Delta X^*_T=\Delta Y^*_T=0$  yields that $\eta_1-\eta_2=0$. Taking then $\sigma=\tau$ yields
$$
\mathbb{E}\bigg[\int_{[\tau,T]}  e^{-\rho t}\,d(X^*_t-Y^*_t)\,\Big|\,\mathscr{F}_{\tau}\bigg]e^{\rho\tau}=0\qquad\text{$\mathbb{P}$-a.s.}
$$
That is, 
$$
\mathbb{E}\bigg[\int_{[\tau,T]}e^{-\rho t}\,d(X^*_t-Y^*_t)\,\Big|\,\mathscr{F}_{\tau}\bigg]=0\qquad\text{$\mathbb{P}$-a.s.~for all $[0,T]$-valued stopping times $\tau$.}
$$
It follows that $M_t:=\int_{[0,t]}e^{-\rho s}\,d(X^*_s-Y^*_s)$ is a 
continuous martingale of bounded variation and hence constant in $t$. Taking derivatives with respect to $t$ yields our claim that $X^*_t=Y^*_t\text{ }\mathbb{P}$-a.s.~for all $t\in[0,T]$.
In particular, we must have that $x = X^*_{0-} = X^*_0 = Y^*_0 = Y^*_{0-} = y$.

{\it Step 5.} 
Since $X^*$ and $Y^*$ are optimal strategies, Proposition~\ref{random fredholm prop} applies. In particular, Equation \eqref{Fredholm formula} holds for $\tau = 0$ and every $\sigma = t \in [0,T]$.
Given that $X^* = Y^*$ has no jumps, we conclude that there exists a constant $\eta$ such that
\begin{equation}\label{equation fredholm ibp}
\begin{aligned}
	\eta 
		&= \mathbb{E}\Big[ 2\int_0^t e^{-\rho(t-s)} \dif X_s^* + \int_t^T e^{\rho(t-s)} \dif X_s^* \Big] \\
	&= \mathbb{E} \Big[ 2 \Big( X_{t}^* - e^{-\rho t} X_{0}^* - \rho \int_0^t e^{-\rho(t- s)} X_s^* \dif s \Big) 
		+ e^{\rho(t-T)} X_{T}^* - X_{t}^* - \rho \int_t^T e^{\rho(t-s)}X_s^* \dif s \Big] \\
	&= - 2 e^{-\rho t} x + \mathbb{E} [ X_t^*] 
		 - 2\rho \int_0^t e^{-\rho(t-s)} \mathbb{E}[X_s^*] \dif s+ \rho \int_t^T e^{\rho(t-s)} \mathbb{E}[X_s^*] \dif s
\end{aligned}
\end{equation}
for all $t \in [0,T].$
The second equality follows from Stieltjes integration by parts, the third from interchanging expectation and integration.
Define
$$
	f(t) \coloneqq \rho \int_0^t e^{-\rho(t-s)} \mathbb{E}[X_s^*] \dif s - \rho \int_t^T e^{\rho(t-s)} \mathbb{E}[X^*_s] \dif s,
	\qquad 0 \le t \le T.
$$
Notice that $f$ is continuously differentiable. Plugging in from \eqref{equation fredholm ibp}, we see that $f$ solves an ordinary differential equation:
\begin{align*}
	f'(t) &= - \rho^2  \int_0^t e^{-\rho(t-s)} \mathbb{E}[X_s^*] \dif s - \rho^2 \int_t^T e^{\rho(t-s)} \mathbb{E}[X^*_s] \dif s 
		+ 2 \rho \mathbb{E}[X_t^*] \\
	&= - \rho^2  \int_0^t e^{-\rho(t-s)} \mathbb{E}[X_s^*] \dif s - \rho^2 \int_t^T e^{\rho(t-s)} \mathbb{E}[X^*_s] \dif s  \\
	&\hphantom{{}=} {}+2\rho \Big( \eta + 2e^{-\rho t} x + 2 \rho \int_0^t e^{-\rho (t-s)} \mathbb{E} [X_s^*] \dif s - \rho \int_t^T e^{\rho(t-s)} \mathbb{E}[X_s^*] \dif s \Big) \\
	&= 3\rho f(t) + 2\rho \big(\eta + 2e^{-\rho t}x\big)
\end{align*}
for all $t \in [0,T]$.
Define further
$$
	g(t) \coloneqq \mathbb{E}[X_t^*] - \rho \int_t^T e^{\rho(t-s)} \mathbb{E}[X_s^*] \dif s,
	\qquad 0 \le t \le T.
$$
Then \eqref{equation fredholm ibp} is equivalent to 
\begin{equation}\label{equation fredholm f and g}
	g(t) = \eta + 2e^{-\rho t} x +2 f(t)
\end{equation}
for all $t \in [0,T],$
showing that $g$ is differentiable and solves an ordinary differential equation as well:
\begin{align*}
	g'(t) &= - 2 \rho e^{-\rho t} x + 6\rho f(t) + 4\rho (\eta + 2e^{-\rho t} x)\\
	&= -2\rho e^{-\rho t} x + 3\rho \big( g(t) - \eta - 2e^{-\rho t}x \big) + 4\rho (\eta + 2e^{-\rho t} x) \\
	&= 3\rho g(t) + \rho \eta
\end{align*}
for all $t \in [0,T].$
It follows from Step 2 and Definition \ref{def ad cont strategy} that $g(T) = \mathbb{E}[X_T^*] = 0.$
Furthermore, letting $t=T$ in \eqref{equation fredholm f and g} yields $f(T) = - ( \eta + 2 e^{-\rho T} x) / 2.$
Solving the ordinary differential equations above with these boundary conditions shows
$$
	f(t) = \frac{e^{-3\rho(T-t)}-4}{6}\,\eta - e^{-\rho t} x
	\qquad\text{ and }\qquad
	g(t) = \frac{e^{-3\rho(T-t)} - 1}{3}\,\eta,
	\qquad 0 \le t \le T.
$$
According to Step 4, $\mathbb{E}[X_0^*] =x$, hence
$$
	x = g(0) - f(0) = \frac{e^{-3\rho T} + 2}{6} \,\eta + x.
$$
We conclude $\eta = 0,$ which implies $g(t) = 0$ for all $t \in [0,T].$
Now notice that
$$
	\frac{d}{dt} \Big[ \int_t^T e^{\rho(t-s)} \mathbb{E}[X_s^*] \dif s \Big]
		= - \mathbb{E}[X_t^*] + \rho \int_t^T e^{\rho(t-s)} \mathbb{E}[X_s^*] \dif s
		= - g(t)
		= 0
$$
 for all $t \in [0,T].$
It follows that 
$$
	\int_t^T e^{\rho(t-s)} \mathbb{E}[X_s^*] \dif s = \int_T^T e^{\rho(t-s)} \mathbb{E}[X_s^*] \dif s = 0,
$$
and therefore 
$
	\mathbb{E}[X_t^*] = g(t) = 0
$
for all $t \in [0,T].$
In particular, $x = \mathbb{E}[X_0^*] = 0.$
\end{proof}

\begin{lemma}\label{stopping time lemma} 
Let $X\in\mathscr{X}(x,[0,T])$, $Y\in\mathscr{X}(y,[0,T])$, and $\tau$ be a stopping time with $\tau\le T$. Then there exists a sequence $(\tau_n)$ of stopping times satisfying the following three conditions: $\tau\le\tau_n\le T$;  on $\{\tau<T\}$ we have $\tau_n\downarrow\tau$ $\mathbb{P}$-a.s.;
and on $\{\tau_n<T\}$ we have $\Delta X^*_{\tau_n}=\Delta Y^*_{\tau_n}=0$   $\mathbb{P}$-a.s.  \end{lemma}

\begin{proof} For each $\omega$, the functions $t\mapsto X_t(\omega)$ and $t\mapsto Y_t(\omega)$ are of bounded variation and hence can be written as differences of increasing functions $X_t^\pm(\omega)$, $Y_t^\pm(\omega)$. That is, $X_t(\omega)=x+X^+_t(\omega)-X^-_t(\omega)$ and $Y_t(\omega)=y+Y^+_t(\omega)-Y^-_t(\omega)$. As a matter of fact, $X^+_t(\omega)$ can be taken as the total variation of $X(\omega)$ over the interval $[0,t]$, and $X^-_t(\omega)=x+X^+_t(\omega)-X_t(\omega)$. Then both $X^+$ and $X^-$ are adapted, and they are also right-continuous. Now take a strictly increasing function $\phi$ mapping $[0,\infty)$ bijectively onto $[0,1)$. Then, for 
$Z_t:= X^+_{\tau+t}+X^-_{\tau+t}+Y^+_{\tau+t}+Y^-_{\tau+t}$, the function 
$f(t):=\mathbb{E} [\,\phi(Z_t)\,]
$ 
is right-continuous and increasing and so has only countably many discontinuity points (recall from Definition~\ref{def ad cont strategy} that $X_t$ and $Y_t$ are defined for all $t\ge0$). Hence, for each $n$, there exists $t_n\in [2^{-n},2^{-n-1})$ such that $f(t_n)-f(t_n-)=0$. Dominated convergence implies that $f(t-)=\mathbb{E} [\,\phi(Z_{t-})\,]$ and in turn that $\Delta Z_{t_n}=0$ $\mathbb{P}$-a.s. By construction, this entails that also $\Delta X_{t_n}=\Delta Y_{t_n}=0$ $\mathbb{P}$-a.s. Hence, letting $\tau_n:=(\tau+t_n)\wedge T$ yields the desired sequence of stopping times.\end{proof}

\subsection{Proof of Corollary~\ref{cont costs cor}}

	Define the two functions
	\begin{align*}
		\varphi_\pm \left(t\right) := -\rho \left(\frac{3\left(x+y\right)\left(e^{3\rho T}+2e^{3\rho t}\right)}{2e^{3\rho T}\left(3\rho T+5\right)-1}\pm \frac{x-y}{2\left(\rho T+1\right)}\right).
	\end{align*}
	Let $\varphi_\pm '$ denote the first derivative of $\varphi_\pm$.
	We see that $\,d  X^*_t = \varphi_+\left(t\right)\,d  t$ and $\,d  Y^*_t = \varphi_- \left(t\right) \,d  t$ on $(0,T)$.
	In addition,
	\begin{align*}
		\Delta X^*_0 = \Delta Y^*_0 = -\frac{3\left(x+y\right)(\left.2e^{3\rho T}+1\right.)}{2\left(2e^{3\rho T}\left(3\rho T+5\right)-1\right)},\hspace{1cm}\text{ and }
		\Delta X^*_T = - \Delta Y^*_T = -\frac{x-y}{2\left(\rho T+1\right)}.
	\end{align*}
	It holds that
	\begin{align*}
		&\hphantom{{}=}\int_{\left[0,T\right]} \int_{\left[0,T\right]} e^{-\rho \left|t-s\right|} \,d  X^*_s \,d  X^*_t\\
		&= \int_0^T \int_0^T  e^{-\rho\left|t-s\right|} \varphi_+'\left(s\right) \varphi_+' \left(t\right) \,d  s \,d  t+ \left(\Delta X^*_0\right)^2 + 2e^{-\rho T} \Delta X^*_0 \Delta X^*_T \nonumber\\
		&\hspace{.5cm}{}+ \left(\Delta X^*_T\right)^2+ 2\Delta X^*_0 \int_0^T  e^{-\rho t} \varphi_+'\left(t\right) \,d  t + 2 \Delta X^*_T \int_0^T  e^{-\rho\left(T-t\right)} \varphi_+'\left(t\right) \,d  t,\nonumber
	\end{align*}
	and (using the fact that $\Delta Y^*_0 = \Delta X^*_0$):
	\begin{align*}
		\int_{\left[0,T\right]} \int_{\left[0,t\right)} e^{-\rho \left(t-s\right)} \,d  Y^*_s \,d  X^*_t
		&= \int_0^T  \int_0^t e^{-\rho\left(t-s\right)} \varphi_-'\left(s\right) \varphi_+'\left(t\right)\,d  s \,d  t + e^{-\rho T} \Delta X^*_0 \Delta X^*_T \\
		&\hspace{.5cm}{}+ \Delta X^*_0 \int_0^T  e^{-\rho t} \varphi_+'\left(t\right)\,d  t + \Delta X^*_T \int_0^T  e^{-\rho\left(T-t\right)} \varphi_-'\left(t\right) \,d  t.\nonumber
	\end{align*}
The computation of the integrals is straightforward:
	\begin{align*}
		&\hphantom{{}=}\int_0^T  e^{-\rho \left|t-s\right|} \varphi_+'\left(s\right) \,d  s\\ 
		&= \frac{3\left(x+y\right)\left(-4e^{3\rho T}+e^{3\rho t}+e^{-\rho t}\left(2e^{3\rho T}+1\right)\right)}{2\left(2e^{3\rho T}\left(3\rho T+5\right)-1\right)} + \frac{\left(x-y\right)\left(e^{-\rho\left(T-t\right)}+e^{-\rho t}-2\right)}{2\left(\rho T+1\right)},\nonumber
	\end{align*}
	and then:
	\begin{align*}
		&\hphantom{{}=}\frac{1}{2} \int_0^T \int_0^T  e^{-\rho \left|t-s\right|} \varphi_+'\left(s\right)\varphi_+'\left(t\right)\,d  s \,d  t\\
		&= \frac{3\left(x+y\right)^2\left(3\rho T e^{6\rho T}-e^{3\rho T}+1\right)}{\left(2e^{3\rho T}\left(3\rho T+5\right)-1\right)^2}+ \frac{\left(x^2-y^2\right)\left(e^{3\rho T}\left(12\rho T-7\right)+6e^{2\rho T}+3e^{-\rho T}-2\right)}{4\left(2e^{3\rho T}\left(3\rho T+5\right)-1\right)\left(\rho T+1\right)}\nonumber\\
		&\hspace{.5cm}{}+ \frac{\left(x-y\right)^2\left(e^{-\rho T}+\rho T-1\right)}{4\left(\rho T+1\right)^2}.\nonumber
	\end{align*}
	Similarly,
	\begin{align*}
		&\hphantom{{}=}\int_0^t e^{-\rho\left(t-s\right)} \varphi_-'\left(s\right) \,d  s\\
		&= \frac{3\left(x+y\right)\left(-2e^{3\rho T}-e^{3\rho t}+e^{-\rho t}\left(2e^{3\rho T}+1\right)\right)}{2\left(2e^{3\rho T}\left(3\rho T+5\right)-1\right)} + \frac{\left(x-y\right)\left(1-e^{-\rho t}\right)}{2\left(\rho T+1\right)},\nonumber
	\end{align*}
	and:	
	\begin{align*}
		&\hphantom{{}=}\int_0^T  \int_0^t  e^{-\rho \left(t-s\right)} \varphi_-'\left(s\right) \varphi_+'\left(t\right) \,d  s \,d  t\\
		&= \frac{3\left(x+y\right)^2\left(3\rho T e^{6\rho T}-e^{3\rho T}+1\right)}{\left(2e^{3\rho T}\left(3\rho T+5\right)-1\right)^2} + \frac{\left(x^2-y^2\right)\left(-3e^{3\rho T}+6e^{2\rho T}+3e^{-\rho T}-6\right)}{4\left(2e^{3\rho T}\left(3\rho T+5\right)-1\right)\left(\rho T+1\right)}\nonumber\\
		&\hspace{.5cm}{}-\frac{\left(x-y\right)^2\left(e^{-\rho T}+\rho T-1\right)}{4\left(\rho T+1\right)^2}.\nonumber
	\end{align*}
	Finally, we see that
	\begin{align*}
	\int_0^T  e^{-\rho t}\varphi_+'\left(t\right) \,d  t &= \Delta X^*_0 \frac{2e^{3\rho T}-2}{2e^{3\rho T}+1} + \Delta X^*_T \left(1-e^{-\rho T}\right),
	\end{align*}
	and
	\begin{align*}
	\int_0^T  e^{-\rho\left(T-t\right)} \left(\varphi_+'\left(t\right) + \varphi_-'\left(t\right)\right) \,d  t &= 2 \Delta X^*_0 \frac{3e^{3\rho T}-2e^{2\rho T}-e^{-\rho T}}{2e^{3\rho T}+1}.
	\end{align*}
	Adding all components and simplifying yields the assertion.	 \qed

	 \parskip-0.5em\renewcommand{\baselinestretch}{0.9}\normalsize
\bibliography{MarketImpact}{}
\bibliographystyle{abbrv}

\end{document}